%% file: main.tex
\documentclass[acmsmall, authorversion, nonacm, screen]{acmart}

\usepackage{xfrac}
\usepackage{flushend}
\usepackage{amsmath}
\usepackage{multicol}
\usepackage{enumitem}
\usepackage{tikz}
\usetikzlibrary{tikzmark}
\usepackage{mathtools}
\AtBeginDocument{%
  \providecommand\BibTeX{{%
    \normalfont B\kern-0.5em{\scshape i\kern-0.25em b}\kern-0.8em\TeX}}}

\usepackage{algorithm}
\usepackage{xspace}
\usepackage[noend]{algpseudocode}
\usepackage{nicefrac}

\usepackage{appendix}
\usepackage{lipsum}
\usepackage{geometry}

\setcopyright{acmlicensed}
\copyrightyear{2024}
\acmYear{2024}
\acmDOI{}

\acmJournal{POMACS}
\acmBooktitle{Proceedings of the 2024 ACM SIGMETRICS / IFIP Performance Joint International Conference on Measurement and Modeling of Computer Systems (SIGMETRICS / Performance '24), June 10--14, 2024, Venice, Italy}

\usepackage[]{algpseudocode}

\makeatletter
\newcommand{\algmargin}{\the\ALG@thistlm}

\renewcommand\paragraph{\@startsection{paragraph}{4}{\parindent}%
  {2pt}
  {-\parindent}
  {\ACM@NRadjust{\@parfont\@adddotafter}}}
\makeatother
\newlength{\whilewidth}
\algdef{SE}[parWHILE]{parWhile}{EndparWhile}[1]
{\parbox[t]{\dimexpr\linewidth-\algmargin}{%
		\hangindent\whilewidth\strut\algorithmicwhile\ #1\ \algorithmicdo\strut}}{\algorithmicend\ \algorithmicwhile}%
\algnewcommand{\parState}[1]{\State%
	\parbox[t]{\dimexpr\linewidth-\algmargin}{\strut #1\strut}}
\makeatother

\newtheorem{thm}{Theorem}[section]
\newtheorem{lem}[thm]{Lemma}

\newtheorem{conj}[thm]{Conjecture}
\newtheorem{obs}[thm]{Observation}
\newtheorem{dfn}[thm]{Definition}

\newtheorem{claim}[thm]{Claim}

\colorlet{blue}{black}

\usepackage{etoolbox}
\makeatletter
\patchcmd{\hyper@makecurrent}{%
    \ifx\Hy@param\Hy@chapterstring
        \let\Hy@param\Hy@chapapp
    \fi
}{%
    \iftoggle{inappendix}{%
        \@checkappendixparam{chapter}%
        \@checkappendixparam{section}%
        \@checkappendixparam{subsection}%
        \@checkappendixparam{subsubsection}%
        \@checkappendixparam{paragraph}%
        \@checkappendixparam{subparagraph}%
    }{}%
}{}{\errmessage{failed to patch}}

\newcommand*{\@checkappendixparam}[1]{%
    \def\@checkappendixparamtmp{#1}%
    \ifx\Hy@param\@checkappendixparamtmp
        \let\Hy@param\Hy@appendixstring
    \fi
}
\makeatletter

\newtoggle{inappendix}
\togglefalse{inappendix}

\apptocmd{\appendix}{\toggletrue{inappendix}}{}{\errmessage{failed to patch}}
\apptocmd{\subappendices}{\toggletrue{inappendix}}{}{\errmessage{failed to patch}}

\DeclareMathOperator*{\argmax}{arg\,max}
\DeclareMathOperator*{\argmin}{arg\,min}

\newcommand{\defeq}{\vcentcolon=}
\newcommand{\ROROmin}{\texttt{RORO-min}\xspace}
\newcommand{\ROROmax}{\texttt{RORO-max}\xspace}
\newcommand{\RORO}{\texttt{RORO}\xspace}
\newcommand{\ROAdv}{\texttt{RO-Advice}\xspace}
\newcommand{\ROAdvmin}{\texttt{RO-Advice-min}\xspace}
\newcommand{\ROAdvmax}{\texttt{RO-Advice-max}\xspace}
\newcommand{\ALG}{\texttt{ALG}\xspace}
\newcommand{\ADV}{\texttt{ADV}\xspace}
\newcommand{\OPT}{\texttt{OPT}\xspace}
\newcommand{\PrOb}{\texttt{OCS}\xspace}
\newcommand{\PrObmin}{\texttt{OCS-min}\xspace}
\newcommand{\PrObmax}{\texttt{OCS-max}\xspace}
\newcommand{\OCS}{\texttt{OCS}\xspace}
\newcommand{\OCSmin}{\texttt{OCS-min}\xspace}
\newcommand{\OCSmax}{\texttt{OCS-max}\xspace}

\newcommand{\OPR}{\texttt{OPR}\xspace}
\newcommand{\OPRmin}{\texttt{OPR-min}\xspace}
\newcommand{\OWT}{\texttt{OWT}\xspace}
\newcommand{\CFC}{\texttt{CFC}\xspace}

\begin{document}

\title[Online Conversion with Switching Costs]{Online Conversion with Switching Costs:\\ \textit{Robust and Learning-augmented Algorithms}}%

\author{Adam Lechowicz}
\affiliation{%
  \institution{University of Massachusetts Amherst}
  \country{USA}
}
\email{alechowicz@cs.umass.edu}

\author{Nicolas Christianson}
\affiliation{%
  \institution{California Institute of Technology}
  \country{USA}
}
\email{nchristianson@caltech.edu}

\author{Bo Sun}
\affiliation{%
  \institution{University of Waterloo}
  \country{Canada}
}
\email{bo.sun@uwaterloo.ca}

\author{Noman Bashir}
\affiliation{%
  \institution{Massachusetts Institute of Technology}
  \country{USA}
}
\email{nbashir@mit.edu}

\author{Mohammad Hajiesmaili}
\affiliation{%
  \institution{University of Massachusetts Amherst}
  \country{USA}
}
\email{hajiesmaili@cs.umass.edu}

\author{Adam Wierman}
\affiliation{%
  \institution{California Institute of Technology}
  \country{USA}
}
\email{adamw@caltech.edu}

\author{Prashant Shenoy}
\affiliation{%
  \institution{University of Massachusetts Amherst}
  \country{USA}
}
\email{shenoy@cs.umass.edu}

\renewcommand{\shortauthors}{Lechowicz et al.}

\begin{abstract}
\input{0-abstract}

\end{abstract}

\maketitle

\section{Introduction}
\label{sec:intro}

\input{1-intro}

\section{Problem Formulation and Preliminaries} %
\label{sec:prob}
\input{2-problem}

\section{Competitive Algorithms}
\label{sec:roro}
\input{3-algos}

\section{Learning-augmented Algorithms}
\label{sec:preds}
\input{4-advice}

\section{Case Study: Carbon-Aware EV Charging} %
\label{sec:eval}
\input{5-experiments}

\section{Conclusion}
\label{sec:conclusion}
\input{6-conclusion}

\begin{acks}
We thank our shepherd Arpan Mukhopadhyay and the anonymous SIGMETRICS / Performance reviewers for their valuable insight and feedback.

This research is supported by National Science Foundation grants CAREER-2045641, CNS-2102963, CNS-2106299, CNS-2146814, CNS-1518941, CPS-2136197, CPS-2136199, NGSDI-2105494, NGSDI-2105648, 1908298, 2020888, 2021693, 2045641, 2213636, and 2211888.

This material is based upon work supported by the U.S. Department of Energy, Office of Science, Office of Advanced Scientific Computing Research, Department of Energy Computational Science Graduate Fellowship under Award Number DE-SC0024386, and an NSF Graduate Research Fellowship (DGE-1745301).
\end{acks}

\section*{Disclaimers}
This report was prepared as an account of work sponsored by an agency of the United States Government. Neither the United States Government nor any agency thereof, nor any of their employees, makes any warranty, express or implied, or assumes any legal liability or responsibility for the accuracy, completeness, or usefulness of any information, apparatus, product, or process disclosed, or represents that its use would not infringe privately owned rights. Reference herein to any specific commercial product, process, or service by trade name, trademark, manufacturer, or otherwise does not necessarily constitute or imply its endorsement, recommendation, or favoring by the United States Government or any agency thereof. The views and opinions of authors expressed herein do not necessarily state or reflect those of the United States Government or any agency thereof.

\bibliographystyle{ACM-Reference-Format}
\bibliography{main}
\clearpage

\appendix
\input{Z-appendix}

\end{document}

%% file: 0-abstract.tex
We introduce and study online conversion with switching costs, a family of online problems that capture emerging problems at the intersection of energy and sustainability.  In this problem, an online player attempts to purchase (alternatively, sell) fractional shares of an asset during a fixed time horizon with length $T$.  At each time step, a cost function (alternatively, price function) is revealed, and the player must irrevocably decide an amount of asset to convert.  The player also incurs a \textit{switching cost} whenever their decision changes in consecutive time steps, i.e., when they increase or decrease their purchasing amount.  
We introduce competitive (robust) threshold-based algorithms for both the minimization and maximization variants of this problem, and show they are optimal among deterministic online algorithms.  We then propose learning-augmented algorithms that take advantage of untrusted black-box advice (such as predictions from a machine learning model) to achieve significantly better average-case performance without sacrificing worst-case competitive guarantees.  Finally, we empirically evaluate our proposed algorithms using a carbon-aware EV charging case study, showing that our algorithms substantially improve on baseline methods for this problem.

%% file: 1-intro.tex
In classic online conversion, an online player buys/sells shares of an item over a fixed time horizon where future prices are unknown, and attempts to minimize their cost or maximize their profit. This paper introduces and studies
\textit{online conversion with switching costs} (\OCS), a novel class of online problems motivated by emerging control problems in the design of sustainable systems.
We consider both minimization (\OCSmin) and maximization (\OCSmax) variants of the problem.  
In \OCSmin, an online player aims to purchase one item over a sequence of time-varying cost functions and decides the \textit{fractional} amount of item to purchase in each round. %
The player must purchase the entire item before a deadline $T$, and they incur a \textit{movement cost} whenever their decision changes, i.e., whenever they purchase different amounts of the item in consecutive time steps.  From the player's perspective, the goal is to minimize their total cost, including the total purchasing cost and any movement cost incurred over the time horizon.  In \OCSmax, the setting is almost the same, except the player sells an item fractionally according to time-varying price functions, so the goal is to maximize their total profit, and any movement costs are subtracted from the revenue.  In both settings, the cost/price functions are revealed one by one in an online manner, and the player makes an irrevocable decision at each time step without the knowledge of future cost/price functions.

Our motivation behind introducing \OCS is an emerging class of \textit{carbon-aware problems} such as carbon-aware electric vehicle (EV) charging~\cite{Cheng:22} and carbon-aware compute shifting~\cite{radovanovic2022carbon,acun2022holistic,bashir2021enabling,Wiesner:21, Hanafy:23, Hanafy:23:CarbonScaler}, which have attracted significant attention in recent years.  A common thread is the goal of reducing carbon emissions by temporally shifting flexible workloads to take advantage of times when low-carbon forms of electricity generation are available.

\OCS builds on a long line of work in the literature of related online algorithms, which can roughly be classified into two problem types: \textit{online search problems}, which feature deadline constraints but do not consider a penalty for switching decisions~\cite{ElYaniv:01, Lorenz:08, Mohr:14, SunZeynali:20, SunLee:21}, and \textit{online metric problems}, which feature switching costs but do not consider a long-term deadline constraint~\cite{Borodin:92, Koutsoupias:09, chenSmoothedOnlineConvex2018a, bubeckChasingNestedConvex2019, bubeckMetricalTaskSystems2021, bubeckRandomizedServerConjecture2022, Christianson:22, Rutten:23, Christianson:23MTS}. %
We briefly review the most relevant problems to our work here. %
In the online search literature, \OCS is most closely aligned with the \textit{one-way trading problem} (\OWT) introduced by~\citet{ElYaniv:01}, in which an online player sells fractional shares of an asset over a sequence of time-varying price functions to maximize their profit.  In the online metric algorithms literature, \OCS is also closely aligned with the problem of \textit{convex function chasing} (\CFC) introduced by~\citet{FriedmanLinial:93}, where an online player makes online decisions $x_t$ in a normed vector space $(X, \lVert \cdot \rVert)$ over a sequence of time-varying cost functions in order to minimize their total cost, which includes the value of the cost function plus a penalty for switching decisions. 

Despite extensive prior literature on the above two categories of problems, there is limited work that includes both long-term constraints (as in \OWT) and switching costs (as in \CFC). To our knowledge, the only prior work is the online pause and resume (\OPR) problem~\cite{Lechowicz:23}, which combines switching costs and a long-term deadline constraint. 
Specifically, \OPR is a generalization of $k$-search, where the online player's overarching objective is to accept the $k$ lowest prices over a time-varying sequence, and the decision space is \textit{binary} (the player either accepts or rejects a price). Similarly to \OCS, \OPR adds a penalty for switching the decision over consecutive slots. The binary design space in \OPR makes it natural to extend the prior single-threshold-based algorithms for $k$-search~\cite{Lorenz:08}. Therefore, the authors in~\cite{Lechowicz:23} develop a \textit{double threshold} algorithm for \OPR and show that the proposed algorithms can achieve the best competitive ratios possible.  
In \OCS, however, the intuitive idea of the double threshold design fails because the design space of the algorithm is continuous. Therefore, this work aims to answer the following question:

\begin{center}
\textit{ Is it possible to design competitive algorithms for \OCS that simultaneously balance \\ switching costs and deadline constraints in continuous decision spaces?}
\end{center}

Although the previous literature in both \OWT and \CFC primarily focuses on providing algorithms with near-optimal adversarial competitive guarantees, moving beyond worst-case analysis has received considerable attention in recent years.  In particular, because algorithms designed to optimize the worst case may be overly cautious or pessimistic, a line of work focused on using machine-learned advice has emerged.  These \textit{learning-augmented algorithms}~\cite{Lykouris:18, Purohit:18} have, in many settings, been able to take advantage of untrusted machine-learned advice to improve the average-case performance (\textit{consistency}) while still maintaining strong adversarial guarantees when the advice is bad (\textit{robustness}). Recent works proposed learning-augmented algorithms for both \OWT~\cite{SunLee:21} and \CFC~\cite{Christianson:22, Christianson:23MTS}; however,
incorporating machine-learned predictions into algorithm design for \OCS is uniquely challenging since even the choice of advice model is not clear for \OCS. The prior literature for online algorithms with advice for online search and metrics problems uses two different advice models. 
In \citet{SunLee:21}, based on an advice model for predicting an \textit{input} parameter (i.e., the \textit{maximum price}), the authors present learning-augmented algorithms for the one-way trading problem that achieves the optimal consistency-robustness trade-off.
For metrics problems and, specifically, convex function chasing, \citet{Christianson:23MTS} give a learning-augmented algorithm that uses black-box advice of the optimal \textit{decision} at each time step and achieves an optimal consistency-robustness trade-off. Given these challenges, a natural open question is: 

\begin{center}
\textit{Can predictions be effectively integrated into an online algorithm for \OCS to provably improve the average case performance of competitive algorithms while maintaining their worst-case guarantees?} 
\end{center}

\subsection{Contributions and Techniques} 
We present online algorithms for both variants of \OCS and show that they achieve the best possible competitive ratios. We then extend our algorithms to the learning-augmented setting, significantly improving practical performance when predictions are accurate while maintaining competitive guarantees. Our technical contributions are summarized in detail below.

\paragraph{Competitive results} \ \
The core algorithmic idea is based on a novel optimization-based framework called ``Ramp-On Ramp-Off'' (\RORO), which we instantiate to develop competitive algorithms for \OCSmin and \OCSmax.  \RORO extends the framework of online threshold-based algorithms (\texttt{OTA}), a preeminent design paradigm for related online search problems~\cite{ElYaniv:01, SunLee:21,Lechowicz:23,Lorenz:08}.
We give upper bounds on the competitive ratio of \RORO in the minimization and maximization settings as a function of the problem's parameters (see Theorems~\ref{thm:roromin} and~\ref{thm:roromax}).  Furthermore, we show lower bounds for the competitive ratio of any deterministic online algorithm, which imply that \RORO is optimal for this problem (see Theorems~\ref{thm:lowerboundmin} and~\ref{thm:lowerboundmax}).  

To achieve the above competitive results for \OCS, the \RORO framework has to tackle substantially different technical challenges than prior work.  
Existing work leveraging the \texttt{OTA} framework either does not consider the presence of switching costs~\cite{ElYaniv:01, SunLee:21, SunZeynali:20}, or only considers a binary decision space~\cite{Lechowicz:23}.  For the online pause and resume  (\OPR) problem, authors in~\cite{Lechowicz:23} show a ``double threshold'' algorithm achieving the optimal competitive ratio.
However, the double threshold technique cannot be generalized to $\OCS$. That is,  
each of the two defined threshold families corresponds to a single case for the previous decision made by the algorithm -- making the decision space continuous in $[0,1]$ breaks this relationship between decisions and the threshold.  
Thus, the primary technical challenge in applying \texttt{OTA} design to \OCS is deriving a technique that can specify continuous decisions while incorporating ``awareness'' of the switching cost.  
The key enabling result for the \RORO framework is a \textit{dynamic threshold} approach (in \autoref{sec:warmup}). Instead of the double threshold idea, the \RORO framework uses a single threshold function complemented by an adaptive mechanism that dynamically adjusts admission criteria based on the previous online decision.  Compared to existing approaches, this framework is more general and can be applied to the continuous $\OCS$ setting. \RORO also recovers the double threshold algorithm for \OPR~\cite{Lechowicz:23} as a special case.

\paragraph{Learning-augmented results} \ \
To go beyond worst-case analysis and attain better \textit{average-case} performance for $\OCS$, we then turn our attention to learning-augmented design. 
The existing learning-augmented results for related problems can be roughly split into two advice models. Firstly, the advice model in online search (\OWT and $k$-search) leverages \textit{input predictions}, which, e.g., predicts the ``best price'' of the instance for \OWT. In contrast, prior results for \texttt{CFC} and \texttt{MTS} typically leverage ``black-box''-type advice, where the online player receives a prediction of an \textit{optimal decision} at each time step.
In learning-augmented design for \OCS, we consider both advice models. For the input advice model, we present an impossibility result: our lower bound result proves that ``best price'' predictions cannot achieve $1$-consistency for any value of robustness in \OCS. This result implies that switching costs complicates the problem beyond the expressiveness of simple predictions (see \autoref{thm:bestPriceConsistency}), and that an advice model providing more information is necessary.  %
Leveraging proof techniques from the advice complexity literature~\cite{Bockenhauer:09, Bockenhauer:14, Gupta:13}, we further show that any $1$-consistent algorithm for $\OCS$ must use advice that grows linearly in the sequence length.   This motivates our usage of an \textit{untrusted black-box advice model} which predicts the online decision at each time step.  Using such advice, we present \ROAdv, a meta-algorithm that uses \RORO as a subroutine to achieve $(1+\epsilon)$-consistency and bounded robustness (see Theorems~\ref{thm:advice-consist-robust-min} and~\ref{thm:advice-consist-robust-max}). We expect that our technique of using advice complexity to prove lower bounds for a prediction model will be broadly applicable to other online problems where the choice of a prediction model in learning augmentation is not obvious.

\paragraph{Case study}

We evaluate the performance of the proposed algorithms using a case study on the carbon-aware electric vehicle (EV) charging problem. In this problem, there is an interruptible and deferrable workload in the form of a charging schedule for an EV.  From the user, we receive a requested amount of charge and a deadline of $T$, representing the time the vehicle will be plugged into an adaptive charger.  Such a charger can change its charging rate based on grid constraints or external signals~\cite{Lee:21:ACN}. Since frequent fluctuations in the charging rate are undesirable from a lithium-ion battery health perspective~\cite{Zhang:06}, the operator can define a \textit{switching penalty} to incentivize a ``smooth'' charging rate.  The objective is to minimize the charging carbon footprint while satisfying the charging demand and maintaining a smooth charging rate within the deadline limit.  As renewable sources such as solar and wind make up greater fractions of power grids worldwide, their intermittency contributes to time-varying carbon intensities, which can be volatile~\cite{Lechowicz:23, Maji:22, Maji:22:CC}, so future values are unknown.  This problem could be exactly modeled by \OCS.

In \autoref{sec:eval}, we evaluate the performance of our robust \RORO and learning-augmented \ROAdv algorithms through the EV charging scheduling case study. We utilize real-world data describing EV charging sessions from the Adaptive Charging Network data set~\cite{Lee:19:ACNData}, combined with carbon intensity traces for the California ISO~\cite{electricity-map}, and simulated local solar generation computed using the National Solar Radiation Data Base~\cite{Sengupta:18:NSRDB}. Our experiments simulate different strategies to charge an EV to the requested capacity before a deadline while dealing with uncertain future carbon intensities. We show that our proposed robust algorithms outperform baseline methods and adapted algorithms for related problems such as \OWT.  
Furthermore, we evaluate our learning-augmented $\ROAdv$ algorithm when advice is provided by an \textit{``off-the-shelf'' ML model} providing time-series predictions of carbon intensity~\cite{Maji:22:CC}.  We show that the imperfect predictions provided by this model are powerful enough for $\ROAdv$ to outperform existing algorithms significantly.

Lastly, we note that $\OCS$ also captures other interesting applications in the sustainability space. 
One such example is the shifting of carbon-aware workloads~\cite{radovanovic2022carbon,acun2022holistic,bashir2021enabling,Wiesner:21,Sukprasert:23}, where time-varying resources are allocated to an interruptible and deferrable workload so that total carbon emissions from the workload are minimized.  Recent work~\cite{Hanafy:23:CarbonScaler, Souza:23} has additionally considered scalability (i.e., when training an ML model, a parallelizable training job can take more resources to ``run harder'' when low-carbon energy is available), which is analogous to the continuous decision space we consider in \OCS. Progress in this field builds on a long line of work considering the temporal shifting of workloads for more sustainable data centers, e.g.,~\cite{gupta2019combining,liu2011greening,liu2012renewable,lin2012dynamic}.
Carbon awareness has also been explored in other contexts, such as HVAC in residential buildings, which considers, e.g., scheduling heating and cooling cycles to coincide with the availability of low carbon energy~\cite{Bovornkeeratiroj:23}. For many such applications, the \OCS formulation is applicable.

%% file: 2-problem.tex
We first introduce and formulate the \OCS problem, then provide brief background on design paradigms which will be useful in the design of our proposed algorithms. \autoref{tab:notations} (in the Appendix) summarizes core \OCS notations for convenience.

\subsection{Problem Formulation} \label{sec:probform}

We formulate two variants of the Online Conversion (\OCS) problem, {\color{blue} primarily considering the minimization variant (\OCSmin) in the main body, and deferring the maximization variant (\OCSmax) to \autoref{apx:ocsmax}}. \footnote{We use \OCS whenever the context applies to both minimization (\PrObmin) and maximization (\PrObmax) variants of the problem; otherwise, we refer to the specific variant. The same notation style extends to our proposed algorithms (\RORO and \ROAdv) throughout the paper.}
In \OCSmin, an online player must buy an asset with total size %
$C$, while minimizing their total cost.  Without loss of generality and for notational simplicity, we assume that $C = 1$.
At each time step $t \in [T]$, a \textit{convex cost function} $g_t(\cdot)$ %
arrives online. 
The player can buy $x_t \in [0, d_t]$ amount of the asset at a cost of $g_t( x_t )$.  {\color{blue} Following convention, each cost function satisfies $g_t(0) = 0$; i.e., if the player purchases nothing, they pay no cost, and $g_t(x_t) \geq 0$ for any valid $x_t$.}  Furthermore, $d_t \leq C$ is a \textit{rate constraint}, which limits the amount of purchases at a given time step $t$.\footnote{For motivating applications such as carbon-aware EV charging, this rate constraint primarily models the case where the requested capacity cannot be met in a single time step (e.g., an EV requests 20 kWh of charge from a 10kW source).}  \ \
Whenever the player's decision changes in consecutive time steps, they incur a \textit{switching cost}, which is proportional to the ``distance'' between $x_t$ and $x_{t-1}$.  For one time step, this is formalized as $\beta \lvert x_t - x_{t-1} \rvert$.  For modeling purposes, we assume $x_0 = 0$ and $x^{T+1} = 0$, which jointly imply that the player must incur \textit{some} switching costs to ``turn on'' and ``turn off'', respectively.  Note that $\OPT$ must incur at least $\frac{2}{T} \beta$ in switching cost\footnote{This follows by observing that the sequence consists of $T$ cost functions.  If each $g_t$ is identical, $\OPT$ can minimize their switching cost by choosing to purchase $1/T$ of the asset at every possible time step.  $\OPT$ then incurs one switching cost to turn on ($\beta \lvert 1/T - x_{0} \rvert = \beta/T$), and one switching cost to turn off ($\beta \lvert x_{t+1} - 1/T \rvert = \beta/T$), for a total of $\frac{2}{T} \beta$.}
and that the maximum switching cost is $2\beta$. The switching cost parameter $\beta$ can be considered a linear coefficient that charges the online player proportionally to the movement between consecutive time steps. 

For simplicity in modeling, this formulation assumes that the cost to ``turn on'' is equivalent to the cost to ``turn off'', and that the magnitude of $\beta$ is constant across all time steps.  
Our formulation and algorithms can be extended to cases where the switching cost is time-varying or asymmetric by setting $\beta$ to be an upper bound on the actual switching cost (e.g $\beta \lvert x_t - x_{t-1} \rvert \geq s_t (x_t, x_{t-1})$ for all $t \in [T]$, where $s_t(x, x')$ is the time-varying or asymmetric switching cost).

Importantly, \OCS requires the online player to buy/sell the entire demand $C$ before the end of the sequence (the ``deadline'').  If the player has bought $w^{(t)}$ fraction of the utilization at time $t$, a \textit{compulsory trade} begins at time step $t$ as soon as $\sum_{\tau=t+1}^T d_{\tau} < 1 - w^{(t)}$ (i.e. when the remaining purchase opportunities in future time steps are not enough to fulfill the demand).  During this compulsory trade, a cost-agnostic algorithm takes over and trades maximally at each time step to meet the constraint -- intuitively, this algorithm comes with no competitive guarantees.  We assume that the rate constraints $\{ d_t \}_{t \in [T]}$ are known in advance, so the player can always identify when the compulsory trade should begin.\footnote{We note that if the rate constraints are not known in advance (i.e., revealed online), the formulation still holds as long as the player receives a signal near the end of the sequence indicating that they must begin the compulsory trade.} 

To ensure that the problem is non-trivial, we generally assume that the compulsory trade does not make up a large fraction of the sequence.  Concretely, the earliest time step $j$ at which the compulsory trade begins (i.e. the first time step such that $\sum_{\tau=j+1}^T d_{\tau} < 1 $) is $j \gg 1$, which implies that $T$ and the rate constraints $\{ d_t \}_{t \in [T]}$ are sized appropriately for the demand.  This assumption is reasonable in practice, since if $T$ is small or the rate constraints are small, the rate constraints will be binding or near-binding in each time step -- i.e., the task's \textit{temporal flexibility} will be low, so that even an omnipotent solution will be unable to take much advantage of ramping up and down to improve its performance.
In summary, the offline version of \PrObmin can be formalized as follows:
\begin{align}
\min_{\{x_t\}_{ t \in [T] }} & \underbrace{ \ \sum_{t=1}^T g_t( x_t ) }_{\text{Purchasing cost}} + \underbrace{ \sum_{t=1}^{T+1}\beta \lvert x_t - x_{t-1} \rvert }_{\text{Switching cost}}, \ \text{s.t., }  \underbrace{\sum_{t=1}^T x_t = 1,}_{\text{Deadline constraint}} \underbrace{ x_t \in [0, d_t] \ \forall t \in [T]}_{\text{Rate constraint}}. \label{align:objMin}
\end{align}

Our focus is on the online version of \OCS, where the player must make irrevocable decisions $x_t$ at each time step without the knowledge of future inputs.  The most important unknowns are the cost/price functions $g_t(\cdot)$, which are revealed online.  

\paragraph{Assumptions and additional notation}

{\color{blue}
For \OCSmin, we assume that cost functions $\{ g_t( \cdot ) \}_{t \in [T]}$ have a bounded derivative, i.e. $L \leq d g_t / d x_t \leq U$, where $L$ and $U$ are known positive constants.  Furthermore, we assume that all cost functions $g_t(\cdot)$ are \textit{convex} -- this assumption is important as a way to model \textit{diminishing returns}, and is empirically valid for the applications of interest.

The switching cost coefficient $\beta$ is assumed to be known to the player, and is bounded within an interval ($\beta \in (0, \nicefrac{U-L}{2})$). 
If $\beta = 0$, the problem is equivalent to one-way trading, and if $\beta$ is ``too large'' (i.e., $> \nicefrac{U-L}{2}$), we can show that any competitive algorithm should \textit{only} consider the switching cost.
\footnote{As brief justification for the bounds on $\beta$, consider the following solutions for \OCSmin.  
In \OCSmin, a feasible solution can have objective value $L + 2\beta$.  Note that if $\beta > \nicefrac{U-L}{2}$, $L + 2\beta > U$, and we argue that the incurred switching cost is more important than the cost values accepted.
}
}

In applications, the deadline $T$ is generally known in advance, although our formulation and algorithms do not require this to be true.  The most important technicality that changes when $T$ is unknown is the compulsory trading at the end of the period, which is used to ensure that the entire asset is purchased/sold before the deadline.  If $T$ is unknown, the player must receive a signal to indicate that the deadline is close and compulsory trading should begin.

\paragraph{Competitive analysis} 

Our goal is to design an online algorithm that maintains a small \textit{competitive ratio}~\cite{Manasse:88, Borodin:92}, i.e., performs nearly as well as the offline optimal solution.   Given an online algorithm $\ALG$, an offline optimal solver $\OPT$, and a valid input sequence $\mathcal{I}$, we denote the optimal cost by $\OPT(\mathcal{I})$, and $\ALG(\mathcal{I})$ is the cost of the solution obtained by running the online algorithm over this input.  Then if $\Omega$ is the set of all feasible input instances, the \textit{competitive ratio} for a minimization problem is defined as:
$
\textnormal{CR}(\ALG) = \max_{\mathcal{I} \in \Omega} \nicefrac{ \ALG(\mathcal{I}) }{ \OPT(\mathcal{I}) }.
$  
Under these definitions, the competitive ratio is always greater than or equal to one, and a \textit{lower} competitive ratio implies that the online algorithm is guaranteed to be \textit{closer} to the offline optimal solution. 

In the emerging literature on learning-augmented algorithms, the competitive ratio is interpreted via the notions of \textit{consistency} and \textit{robustness}, introduced in~\cite{Lykouris:18, Purohit:18}.  {\color{blue} Formally, let $\ALG(\mathcal{I}, \varepsilon)$ denote the cost of a learning-augmented online algorithm on input sequence $\mathcal{I}$ when provided predictions with error factor $\varepsilon$. Then the notion of consistency describes the competitive ratio of a learning-augmented algorithm when the predictions provided to it are exactly correct, i.e. $\max_{\mathcal{I} \in \Omega} \nicefrac{ \ALG(\mathcal{I}, 0) }{ \OPT(\mathcal{I}) }$.  Conversely, robustness captures the competitive ratio when predictions are adversarially incorrect, i.e. $\max_{\mathcal{I} \in \Omega} \nicefrac{ \ALG(\mathcal{I}, \mathbf{E}) }{ \OPT(\mathcal{I}) }$, where $\mathbf{E}$ is a maximum error factor (or $\infty$). }

\subsection{Background: Threshold-based Algorithms for Online Problems} \label{sec:OTA}

Existing work on related problems such as online pause and resume \cite{Lechowicz:23}, online search \cite{Lee:22,Lorenz:08,ElYaniv:01}, one-way trading \cite{SunLee:21,SunZeynali:20}, and online knapsack \cite{Zhou:08,sun2022online,Yang2021Competitive} often use of a technique known as \textit{threshold-based algorithm design}. Threshold-based algorithms are a family of online algorithms where a thoughtfully-designed \textit{threshold function} is used to guide an online %
decision at each time step.  In broad terms, this threshold establishes a ``minimal acceptable quality'' that an incoming input must meet to be accepted by the algorithm. The threshold values are deliberately chosen to ensure that an algorithm consistently accepting/rejecting inputs based on this (potentially changing) ``standard'' at each step is guaranteed to obtain a certain competitive ratio.

This algorithmic framework has seen success in the related problems mentioned above; many of the threshold-based online algorithms presented in these works are optimal in the sense that the threshold function guarantees a competitive ratio that matches the information-theoretic lower bound for the online problem.  
Threshold-based design has not historically been applied to problems with switching costs such as \OCS.  In \cite{Lechowicz:23}, the authors present the first (to our knowledge) threshold-based algorithm for a problem with switching costs, namely the online pause and resume problem.  Although these results and the general framework do not directly translate to the \OCS setting, we make use of ideas from these existing works.  We briefly detail the most relevant results from prior work before discussing how this framework generalizes to \OCS in \autoref{sec:roro}.

\paragraph{One-way trading.}  
In the canonical one-way trading problem, a sequence of time-varying prices is revealed to an online player one-by-one~\cite{ElYaniv:01}.  The player's objective is to maximize their profit from selling an asset.  When each price $c_t$ arrives, the player must immediately decide a fraction of their asset $x_t \in [0,1]$ to sell at the current price, and they must sell the entire asset by the end of the sequence.  
A generalization of the problem considers price functions $g_t(\cdot)$, which must be \textit{concave}, but are not necessarily linear.  Assuming bounded derivative for the price functions, i.e., $L \le \nicefrac{d g_t}{d x} \le U$, \citet{SunZeynali:20} show a deterministic threshold-based algorithm which achieves the optimal competitive ratio of $1 + W\left( \left( \nicefrac{U}{L} - 1 \right)/e \right)$, where $W(\cdot)$ is the Lambert $W$ function.

\paragraph{Pause and resume (and $k$-search).}
In recent work, \citet{Lechowicz:23} presents the online pause and resume problem (\OPR), which can be summarized as $k$-min/$k$-max search~\cite{Lorenz:08} with an added switching cost.
In both \texttt{OPR} and $k$-search, a sequence of time-varying prices is revealed to an online player one by one.  When each price $c_t$ arrives, the player must immediately decide whether to reject or accept the price (i.e., $x_t \in \{0, 1\}$).  The player's objective in both problems is to find the $k$ lowest (conversely, highest) prices in the sequence to minimize their cost (conversely, maximize their profit), and they must accept $k$ prices before the end of the sequence.  In \OPR, the player additionally pays a switching cost $\beta$ whenever their decision changes in consecutive time steps.  Assuming bounded prices in the range $[L, U]$, \cite{Lechowicz:23} shows deterministic \textit{double threshold} algorithms for \OPR, where each threshold separately corresponds to each of the possible prior decisions $x_{t-1} \in \{0, 1\}$.  These algorithms achieve the optimal competitive ratio for both variants of \OPR.

%% file: 3-algos.tex
We first introduce a general algorithmic framework that we term Ramp-On, Ramp-Off (\RORO). 
Then, we demonstrate the application of \RORO in the context of the pause and resume problem for simplicity before applying it to the \OCSmin problem {\color{blue} (and \OCSmax in \autoref{apx:ocsmax-roro}).} %

\subsection{The online Ramp-On, Ramp-Off (\RORO) framework} \label{sec:warmup}

\begin{algorithm}[!t]
    \small
	\caption{Online Ramp-On, Ramp-Off (\RORO) framework}
	\label{alg:roro}
	\begin{algorithmic}[1]
		\State \textbf{input:} ramping-on problem $\textsc{RampOn}(\cdot)$, ramping-off problem $\textsc{RampOff}(\cdot)$,\\ pseudo-cost function $\textsc{PCost}(\cdot)$
        \State \textbf{initialization:} initial decision $x_0 = 0$, initial utilization $w^{(0)} = 0$;
		\While{cost/price function $g_t(\cdot)$ is revealed and $w^{(t-1)} < 1$}
		\State solve the \textbf{(ramping-on problem)} to obtain decision $x_t^+$ and its pseudo cost $r_t^+$, 
        \begin{align}
            {x}_t^+ &= \textsc{RampOn}(g_t( \cdot ), x_{t-1} ),\\
            r_t^+ &= \textsc{PCost}(g_t( \cdot ), x_{t}^+, x_{t-1}).
        \end{align}
        \State solve the \textbf{(ramping-off problem)} to obtain decision $x_t^-$ and its pseudo cost $r_t^-$, 
        \begin{align}
            {x}_t^- &= \textsc{RampOff}(g_t( \cdot ), x_{t-1} ),\\
            r_t^- &= \textsc{PCost}(g_t( \cdot ), x_{t}^-, x_{t-1}).
        \end{align}
        \State \textbf{if } $r^+_t \le r^-_t$ \textbf{ then } set $x_t = x_t^+$ \textbf{ else } set $x_t = x_t^-$;
        \State update the utilization $w^{(t)} = w^{(t-1)} + x_t$;
        \EndWhile
	\end{algorithmic}
\end{algorithm}

This section presents a general online optimization framework called Ramp-On, Ramp-Off (\RORO) (see \autoref{alg:roro}).
In the \RORO framework, whenever an input arrives online, we solve two \textit{pseudo-cost minimization problems}, with a restricted decision space in each.  We denote these problems as the \textbf{ramping-on} and \textbf{ramping-off} problems.  
In this framework, we employ the \textit{pseudo-cost}, a notion inherited from the online search literature~\cite{SunZeynali:20}.  At a high level, the pseudo-cost generalizes the traditional concept of threshold-based algorithms, where, recall, a threshold value/function establishes the minimum/maximum acceptable inputs for the competitive algorithm to accept, given the current state.  Then, at each time step, when input arrives, the pseudo-cost of a particular decision gives the actual cost of the decision at the current time step, \textit{minus} a threshold value term describing the maximum/minimum acceptable input that would justify the decision.
Intuitively, the pseudo-cost serves as a mechanism preventing an algorithm from ``waiting too long'' to accept any prices. For example, in \OCS, naïvely minimizing the cost function $g_t$ at each time step results in decisions $x_t = 0$ for all $t \in [T]$, which incur no cost.  However, this sequence of decisions is undesirable since the final compulsory trade needed to satisfy the deadline constraint may force the player to pay a large cost. Thus, the pseudo-cost ``balances'' between the two extremes of buying immediately and waiting indefinitely by enforcing the acceptance of inputs that are ``good enough''.

The namesakes of the \RORO framework -- the \textit{ramping-on} and \textit{ramping-off} problems -- find the \textit{best decision} in either of two restricted decision spaces that minimize this pseudo-cost.  In the ramping-on problem, the decision space is restricted to online decisions where $x_t$ increases or stays the same, i.e., $ x_t \geq x_{t-1}$.  In the ramping-off problem, the considered decisions are all $\leq x_{t-1}$ (i.e., the online decision decreases or stays the same). The framework then chooses whichever of the two resulting decisions yields a better pseudo-cost.  
In what follows, we discuss these as two distinct problems for instructive value and ease of presentation.

\paragraph{Intuition and a motivating example} To provide intuition for the design of the \RORO framework, we show that the double threshold algorithm shown for online pause and resume in~\cite{Lechowicz:23} is a special case of this framework.
We focus on the minimization version of pause and resume (\OPRmin), summarized in \autoref{sec:OTA}.
At a high level, \OPRmin includes many of the same components as \OCSmin, with a few key differences.  Most notably, the decision space in \OPRmin is binary, i.e., $x_t \in \{0, 1\}$, while decisions in \OCS are continuous ($x_t \in [0, d_t]$).  Furthermore, the online player in \OPRmin is restricted to accept at most $1/k$ of their total demand in a single time step (recall that $k$ prices must be accepted before the deadline), while the online player in \OCS can choose to purchase their entire demand $C$ in a single time step (if rate constraint $d_t$ permits).

{\color{blue} 
In~\cite{Lechowicz:23}, the authors present a \textit{double threshold} algorithm for \OPRmin.  At a high level, the idea behind this algorithm is to change the criteria for accepting a price based on the switching cost that will be incurred with respect to the previous online decision.  Intuitively, when a price arrives and a previously inactive player considers whether to accept it, they must add the extra switching cost of $\beta$ they will incur to ``switch on''.   Conversely, if the player has accepted the previous price, once a new price arrives, they should discount the price by $\beta$, since \textit{not} accepting this new price will force them to ``switch off'' and incur an extra switching cost of $\beta$.

This idea is formalized in the specification of two threshold families $\{u_i\}_{i\in[k]}$ and $\{\ell_i\}_{i\in[k]}$, each representing $k$ ``minimum acceptable prices''.  The algorithm chooses which family to use based on the previous online decision $x_{t-1}$ as described above.  Then if $i-1$ prices have been accepted so far, the $i$th price accepted will be the first price which is at most $\ell_i$ if $x_{t-1} = 0$, or at most $u_i$ if $x_{t-1} = 1$.
}

\smallskip

We will illustrate the special case of \RORO by first considering the \textit{pseudo-cost} in the context of \OPRmin, where inputs are scalar values arriving online, and decisions are binary.  When a price $c_t$ arrives, the algorithm can consider accepting the price $(x_t = 1)$ or rejecting it $(x_t = 0)$.  The pseudo-cost is the actual cost $c_t x_t + \beta \lvert x_t - x_{t-1} \rvert$, minus a threshold value determining the maximum acceptable price.  Then, if $x_t = 1$ and the pseudo-cost $< 0$, we know that the threshold value is \textit{greater than} the actual cost and the price should be accepted. %
In the following, we formalize the above intuition by defining the ramping-on problem $\textsc{RampOn}$, the ramping-off problem $\textsc{RampOff}$, and the pseudo-cost function $\textsc{PCost}$ to obtain a %
\RORO instantiation, which we show to be equivalent to the double threshold algorithm.

\begin{claim}\label{claim:equiv}
The \RORO framework is equivalent to the double threshold algorithm for \OPRmin in~\cite{Lechowicz:23} when the framework is instantiated as follows.
Let $\phi$ be a family of $k$ thresholds defined as: $\phi_i = u_i - \beta = \ell_i + \beta, \forall i\in[k]$, where $\{u_i\}_{i\in[k]}$ and $\{\ell_i\}_{i\in[k]}$ are the double thresholds for \OPRmin in~\cite{Lechowicz:23}.
Then, define the ramping-on problem, ramping-off problem, and the pseudo-cost function as:
\begin{align}
    \emph{\textsc{RampOn}}(c_t, x_{t-1}) &= \argmin_{x_t \in \{0,1\}, \ x_t\ge x_{t-1}} c_t x_t + \beta (x_t - x_{t-1}) - \phi_{i} x_t,\\
    \emph{\textsc{RampOff}}(c_t, x_{t-1}) &= \argmin_{x_t \in \{0,1\}, \ x_t \le x_{t-1}} c_t x_t + \beta (x_{t-1} - x_{t}) - \phi_{i} x_t,\\
    \emph{\textsc{PCost}}(c_t, x_t, x_{t-1}) &= c_t x_t + \beta \lvert x_t - x_{t-1} \rvert - \phi_{i}x_t.
\end{align}
\end{claim}

We summarize this instantiation in pseudocode form in the appendix, in \autoref{alg:warmup}. We proceed to justify the equivalence of the \RORO instantiation described in Claim~\ref{claim:equiv} against the double threshold algorithm shown for \OPRmin.

If $x_{t-1} = 0$, the ramping-off solution is restricted to $x_t^- = x_{t-1} = 0$, and thus the online solution is dominated by the ramping-on solution, i.e., $x_t = x_t^+$. 
The ramping-on solution is exactly
\begin{align}
    {\small
    x_t^+ =
    \begin{cases}
     1 & \textbf{if}\ c_t + \beta \le \phi_i \\
     0 & \textbf{if}\ c_t + \beta > \phi_i
    \end{cases},}
\end{align}
which is the same as the solution obtained by the double threshold algorithm, i.e., $x_t = 1$ if $c_t \le \phi_i - \beta = \ell_i$; otherwise, $x_t = 0$.
If $x_{t-1} = 1$, the ramping-on is restricted to $x_t^+ = 1 = x_{t-1}$, and thus the online solution is dominated by the ramping-off, i.e., $x_t = x_t^-$.
The ramping-off solution is
\begin{align}
    {\small
    x_t^- =
    \begin{cases}
     1 & \textbf{if}\ c_t - \phi_i \le \beta \\
     0 & \textbf{if}\ c_t - \phi_i > \beta
    \end{cases},}
\end{align}
which is the same as the solution obtained by the double threshold algorithm, i.e., $x_t = 1$ if $c_t \le \phi_i + \beta = u_i$ and otherwise $x_t = 0$.
In what follows, we describe how \RORO solves \OCS.  

\subsection{Solving \OCSmin using the \RORO framework} \label{sec:ocsmin-roro}

We now move to apply the \RORO framework and propose an algorithm for solving \OCSmin.  Prior work~\cite{SunLee:21} has observed that the limiting case of $k$-search (i.e.~$k \to \infty$) can be recast as the one-way trading problem if the formulation is generalized in a few key areas.  Leveraging the \RORO framework, we apply a similar technique to recast the limiting case of \OPRmin as \OCSmin and then generalize the threshold function in Claim~\ref{claim:equiv}.

Recall that in \OPR, the player must accept $k$ prices total, and at the time that they consider accepting the $i$th price, they have already accepted $i-1$ prices previously (where $i \in [k]$).  In the continuous setting, we let $w^{(t)} \in [0, 1]$ denote the \textit{utilization} at time $t$, i.e., $w^{(t)}$ gives the fraction of the total asset which has been purchased up to $t$. Equivalently, in the binary setting, $w^{(t)}$ is defined as $\nicefrac{i-1}{k}$.
In the binary setting of \OPR, the player can only accept a single price at a time, thus satisfying $\nicefrac{1}{k}$ of their task. This is insufficient for \OCS, so we relax this constraint and allow the player to purchase more than $\nicefrac{1}{k}$ amount in a single time step. In \OPR, this is analogous to allowing to purchase more than 1 ``item'' in a single time step. Finally, instead of a single price $c_t \in [L, U]$ arriving online, a convex cost function $g_t(\cdot)$ arrives online and the decision $x_t$ is in $[0, d_t]$.  

With these transformations, we can initialize a \textit{dynamic threshold} $\phi(w)$ for \OCSmin that will inform the design of our \RORO instantiation for \OCSmin.
In the following, we slightly abuse notation and use $w$ instead of $w^{(t)}$, wherever dropping $t$ is appropriate for the context.

\begin{definition}[Threshold function $\phi$ for \OCSmin]\label{dfn:phi-min}
For any utilization $w \in [0,1]$, $\phi$ is defined as:
\begin{align}\label{eq:phi-min}
    \phi(w) = U - \beta + (\nicefrac{U}{\alpha} - U + 2 \beta) \exp(\nicefrac{w}{\alpha}),
\end{align}
where $\alpha$ is the competitive ratio and is defined in \eqref{eq:alpha}.
\end{definition}

This threshold function follows by taking the limit (as $k \to \infty$) of the closed form threshold $\{\phi_i\}_{i \in [k]}$ defined in Claim~\ref{claim:equiv}, and observing that $w = \nicefrac{i-1}{k}$.  We note that $\phi(1) - \beta = L$, which corresponds to the case of the final threshold value in the \OPR instantiation, i.e., $\phi_{k+1} = L + \beta$.
Given $\phi (\cdot )$, we fully describe our instantiation of \RORO for \OCSmin (\ROROmin) by defining the ramping-on problem, the ramping-off problem, and the pseudo-cost function in the following.

\begin{definition}[\RORO instantiation for \OCSmin (\ROROmin)]
The \RORO framework solves \OCSmin when instantiated as follows.
Let $\phi (\cdot) : [0,1] \rightarrow [L, U]$ be the dynamic threshold defined in~\ref{dfn:phi-min}.
Then, define the ramping-on problem, ramping-off problem, and the pseudo-cost function as:
\begin{align}
    \textsc{RampOn}(g_t( \cdot ), w^{(t-1)}, x_{t-1}) &= \quad \ \ \argmin_{\mathclap{x \in [x_{t-1}, \min (1- w^{(t-1)}, d_t) ]}} \quad \ \ g_t( x ) + \beta (x - x_{t-1}) - \int_{w^{(t-1)}}^{w^{(t-1)} + x}\phi(u) du,\\ \label{eq:rampon}
    \textsc{RampOff}(g_t( \cdot ), w^{(t-1)}, x_{t-1}) &= \argmin_{x \in [0,  \min (x_{t-1}, d_t)]} g_t ( x ) + \beta (x_{t-1} - x) - \int_{w^{(t-1)}}^{w^{(t-1)} + x}\phi(u) du,\\ \label{eq:rampoff}
    \textsc{PCost}(g_t( \cdot ), w^{(t-1)}, x_t, x_{t-1}) &= g_t ( x_t ) + \beta |x_t - x_{t-1}| - \int_{w^{(t-1)}}^{w^{(t-1)} + x_t}\phi(u) du.
\end{align}
\end{definition}\label{dfn:roromin}

The optimizations defined above are inserted into the pseudocode defined in \autoref{alg:roro} to create an instance that solves \OCSmin (\ROROmin).  {\color{blue} At a glance, it is not obvious that the ramping-on and ramping-off problems given by \eqref{eq:rampon} and \eqref{eq:rampoff} can be solved efficiently.  In \autoref{apx:convexproofmin}, we show that both are convex minimization problems, which can be solved efficiently (e.g., using iterative methods). }
In the following theorem, we state the competitive result for $\ROROmin$.  We give a proof sketch of the result and discuss its significance here, and relegate the full proof to \autoref{apx:comp-proofs-min}.

\begin{thm}\label{thm:roromin}
\autoref{alg:roro} for \OCSmin (\ROROmin) is $\alpha$-competitive when the threshold function is given by $\phi(w)$ (from Def.~\ref{dfn:phi-min}), where $\alpha$ is the solution to $\frac{U - L - 2 \beta}{U/\alpha - U - 2 \beta} = \exp(1/\alpha)$ and is given by
\begin{equation}
    \alpha \defeq \left[ W \left( \left( \frac{2\beta}{U} + \frac{L}{U} -1 \right) e^{\nicefrac{2\beta}{U}-1} \right) - \frac{2\beta}{U} + 1 \right]^{-1} \label{eq:alpha}. 
\end{equation}
In the above, $W(\cdot)$ is the Lambert $W$ function, defined as the inverse of $f(x) = xe^x$.
\end{thm}
\begin{proof}[Proof Sketch for \autoref{thm:roromin}]

To show this result, we give three lemmas characterizing the optimal solution and the solution obtained by \ROROmin.  Note that the solution obtained by \ROROmin is feasible, as the constraints in the ramping-on problem combined with the compulsory trade always enforce that $\sum_{t=1}^{T} x_t = 1$.

{\color{blue}
We first consider the case without rate constraints, and derive a lower bound on the cost incurred by \OPT.  We denote the final utilization on an arbitrary instance ${\mathcal{I} \in \Omega}$ attained by \ROROmin (before the compulsory trade) by $w^{(j)}$.  Leveraging the definition of the pseudo-cost and an analysis of the ramping-on solution, we show by contradiction that the \textit{best marginal cost} (i.e., the smallest cost function which arrives at any point in the sequence) is lower bounded by $\phi(w^{(j)}) - \beta$.  This gives a corresponding lower bound on $\OPT(\mathcal{I})$ (Lemma~\ref{lem:opt-lb}).

We continue by showing an upper bound on the cost incurred by \ROROmin under the same sequence (denoted by $\ROROmin(\mathcal{I})$).  By combining a worst-case analysis of the actual switching cost paid by \ROROmin with the definition of the pseudo-cost and its relationship to the cost functions accepted by \ROROmin, we show that $\ROROmin(\mathcal{I})$ is upper bounded by $\int^{w^{(t)}}_0 \phi(u) du + \beta w^{(t)} + (1-w^{(t)}) U$ (Lemma~\ref{lem:alg-ub}).
}

Finally, we consider the setting with rate constraints.  Through a careful analysis of the different costs paid by $\OPT$ and $\ALG$ in this setting, we show that the rate constraints do not degrade the worst-case competitive ratio (Lemma~\ref{lem:rate-const-min}).
{\color{blue}
A final analysis of the interaction between these three lemmas yields the result, showing that \ROROmin's competitive ratio is bounded by $\alpha$. The full proof can be found in \autoref{apx:comp-proofs-min}.
}
\end{proof}

\autoref{thm:roromin} gives an upper bound on the competitive ratio of \ROROmin, showing how the parameters of the problem (e.g., $L$, $U$, and $\beta$) filter through and appear in the competitive results.
A reasonable question to ask is whether any other algorithm for \OCSmin can achieve a better competitive ratio.  In the following result, we highlight that no improvement is possible, i.e., that \ROROmin achieves the optimal competitive ratio across all deterministic online algorithms for \OCSmin.  
We give a proof sketch of the result and discuss its significance here, and relegate the full proof to \autoref{sec:lb-proof-min}.

\begin{thm} \label{thm:lowerboundmin}
No deterministic online algorithm for \OCSmin can achieve a competitive ratio better than $\alpha$, as defined in \eqref{eq:alpha}.
\end{thm}
\begin{proof}[Proof Sketch for \autoref{thm:lowerboundmin}]

We first define a special class of worst-case $x$-decreasing instances $\mathcal{I}_x : x \in [L, U]$ for which the offline optimal solution is $\OPT(\mathcal{I}_x) = x$.
An $x$-decreasing instance (defined formally in Def.~\ref{dfn:xinstance-min}) is a sequence of linear cost functions of the form $g(y) = \kappa y$, where $\kappa$ is a linear coefficient.  Letting $n, m \in \mathbb{N}$ be sufficiently large, and $\delta := \nicefrac{(U-L)}{n}$, there are $n_x := 2 \cdot \lceil \nicefrac{(x - L)}{\delta} \rceil + 1$ alternating batches of cost functions.  The $i^{\text{th}}$ batch ($i\in [n_x-2]$) contains a single cost function with $\kappa = U-(\lceil i / 2 \rceil)\delta$ if $i$ is even, and these ``decreasing'' batches are interrupted by blocks of $m$ cost functions with $\kappa = U$ (i.e., when $i$ is odd).  See below for an illustration.
\begin{align*}
\{\underbrace{Uy,\dots,Uy }_{m}, (U-\delta)y, \underbrace{Uy,\dots,Uy}_{m}, (U-2\delta)y,\dots, \underbrace{Uy,\dots,Uy }_{m},\underbrace{(x+\varepsilon)y,\dots,(x+\varepsilon)y }_{m}, \underbrace{Uy,\dots,Uy }_{m}\}
\end{align*}

As $n \to \infty$, the alternating single cost functions in an $x$-decreasing sequence continuously decrease down to $x$, and each of these ``good cost functions'' is interrupted by a section of worst-case $U$ cost functions.  Note that the last few cost functions in an $x$-decreasing instance are always $U$.  We also note that $\mathcal{I}_{U}$ is simply a stream of $m$ cost functions $U$.

By defining the special class of instances, we solve a key technical challenge -- the actions of any deterministic online algorithm $\ALG$ on these instances can then be fully described by an arbitrary conversion function $h(x) : [L, U] \rightarrow [0,1]$, where $h(x)$ denotes the fractional amount purchased (before the final compulsory trade) on instance $\mathcal{I}_x$.  

{\color{blue}
Using this conversion function, we can construct several necessary conditions that must be satisfied by any $\alpha^\star$-competitive algorithm, where $\alpha^\star$ is the target (best) competitive ratio.  First, we can observe that the cost of the arbitrary algorithm is $\ALG(\mathcal{I}_x) = h(\nicefrac{U}{\alpha^\star}) \nicefrac{U}{\alpha^\star} - \int^x_{\nicefrac{U}{\alpha^\star}} u d h(u) + 2\beta h(x) + (1 - h(x))U$,
where $u dh (u)$ is the cost of purchasing $dh(u)$ utilization at price $u$.  Based on this, since $\ALG$'s cost must be $\alpha^\star$-competitive against $\OPT$, we can derive a differential inequality that lower bounds the conversion function $h(x)$.  

We can then use Gr\"{o}nwall's Inequality~\cite[Theorem 1, p. 356]{Mitrinovic:91} to obtain a \textit{necessary condition}, where $\alpha^\star$ and $h(x)$ must jointly satisfy the following: 
$\alpha^\star \ln \left( L + 2\beta -U \right) - \alpha^\star \ln \left( \nicefrac{U}{\alpha^\star} + 2\beta -U \right) \leq h(L) = 1$.  The optimal $\alpha^\star$ is obtained when the inequality is binding, and the result follows.  %
}
\end{proof}

By combining Theorems~\ref{thm:roromin} and \ref{thm:lowerboundmin}, we conclude that \ROROmin is optimal for \OCSmin.  We note that the definition of $\alpha$ gives that when the switching cost $\beta = 0$, our \RORO algorithm recovers the optimal competitive ratio for the minimization variant of one-way trading, which is $\alpha \thicksim \left[ W \left( \left( \nicefrac{L}{U} - 1 \right) e^{-1} \right) + 1 \right]^{-1}$, as given by~\cite{Lorenz:08, SunLee:21}.  It is known that $W(x) \thicksim \ln (x)$ as $x \to \infty$ ~\cite{HoorfarHassani:08, Stewart:09}.  As $\beta$ grows, the competitive ratio $\alpha$ intuitively degrades linearly in the magnitude of $\beta$, although the assumed upper bound on $\beta$ prevents the competitive ratio from becoming unbounded.

%% file: 4-advice.tex
In this section, we explore how ML advice can help break past the pessimistic worst-case bounds of competitive algorithms for \OCS.  We propose a meta-algorithm, \ROAdv, which integrates advice to significantly improve empirical performance and theoretical bounds when the advice is of high quality. Along the way, we discuss and consider two distinct models of advice employed in prior work on learning-augmented algorithms.  We present lower bounds that show that ``simple predictions'' are not powerful enough to allow learning-augmented algorithms to achieve $1$-consistency in \PrOb (\autoref{thm:bestPriceConsistency}). In addition, we show that the amount of advice necessary to break through this barrier is at least linear in the length of the instance (\autoref{thm:advicecomplexity}).  These results, combined with the availability of applicable machine learning techniques, heavily influence our choice of advice model, which is detailed below.

\subsection{Advice Model} \label{sec:advice-model}

In~\cite{SunLee:21, Lee:22}, the authors present learning-augmented algorithms for one-way trading and $k$-search, respectively.  Both works consider a ``best price'' prediction model, where the algorithm receives a single prediction $\hat{P} \in [L, U]$ of the best-quality input in the upcoming sequence.  For instance, $\hat{P}$ in the \OCS problem would predict the lowest cost function in the minimization setting or the largest price function in the maximization setting.
On the other hand, learning-augmented work for convex function chasing (\CFC)~\cite{Christianson:22}, smoothed online convex optimization (\texttt{SOCO})~\cite{Rutten:23}, and metrical task systems (\texttt{MTS})~\cite{Antoniadis:20MTS, Christianson:23MTS} use \textit{black-box advice}, which directly predicts the optimal online decision at each time step.  In the \OCS problem, this advice would involve predicting $x_t$ for each time step $t \in [T]$.
We first consider leveraging ``best price'' predictions since our \RORO algorithm design is closer in lineage to the robust one-way trading algorithm used as a baseline in~\cite{SunLee:21}.  Notably, however, as we show in the following theorem, 
these relatively simple predictions are not expressive enough in the \OCS setting for learning-augmented algorithms to obtain performance comparable to the offline optimal.
The full proof of \autoref{thm:bestPriceConsistency} is deferred to Appendix~\ref{apx:advice}.

\begin{thm}\label{thm:bestPriceConsistency}
    Any learning-augmented algorithm for \PrOb which uses a prediction of the \textit{``best price''} in a sequence cannot achieve $1$-consistency, regardless of robustness.
\end{thm}

\begin{proof}[Proof Sketch for \autoref{thm:bestPriceConsistency}]
Assume that an algorithm for \OCS receives a prediction of the best price $\hat{P}$, and that this prediction is correct. While this information is useful because it enables the algorithm to wait to accept cost/price functions matching $\hat{P}$, it does not convey any information about the switching cost, which is a crucial part of making optimal decisions.

We construct two toy sequences to illustrate how the best price prediction fails in the presence of switching costs.  The first sequence consists of one cost/price function $g(y) = \hat{P}y$, followed by any number of significantly worse cost/price functions.
The second sequence consists of one cost/price function $g(y) = \hat{P}y$, followed by some number of significantly worse cost/price functions, and ending with a block of $m$ cost/price functions $\{ g(y) = \hat{P}y \} \times m$.

From the viewpoint of the first time step, both sequences look identical to the learning-augmented algorithm, but the switching cost incurred by the optimal solution is significantly different ($2\beta$ for the first sequence, $\nicefrac{2\beta}{m}$ for the second sequence). %
Therefore, the online algorithm must incur a switching cost of $2\beta$ (otherwise, it would not be $1$-consistent for the first sequence), while the optimal solution in the second sequence incurs a switching cost of $\nicefrac{2\beta}{m}$, which approaches $0$ as $m \to \infty$.
Thus, a learning-augmented algorithm with simple predictions cannot achieve a consistency better than $\nicefrac{\ALG}{\OPT} \geq \nicefrac{L + 2\beta}{L}$ in the minimization or $\nicefrac{\OPT}{\ALG} \geq \nicefrac{U}{U - 2\beta}$ in the maximization setting.
\end{proof}

The result in \autoref{thm:bestPriceConsistency} motivates considering the question: \textit{how much advice is needed to achieve $1$-consistency?}  Answering this question provides insight about which advice model is the ``right choice'' for \OCS \ -- although we may not require $1$-consistency in all cases, it is valuable to have a parameterized learning-augmented algorithm which \textit{can} achieve $1$ consistency if predictions are consistently accurate.  In what follows, we leverage techniques from the advice complexity literature~\cite{Bockenhauer:09, Emek:11, Bockenhauer:14, Gupta:13} to prove a lower bound on the amount of advice necessary for $1$-consistency in \OCS.  We give a proof sketch of the result and discuss its significance here, relegating the full proof to \autoref{sec:advice-comp}.

\begin{thm}\label{thm:advicecomplexity}
    For any instance of \OCS with sequence length $T$, at least $\Omega(T)$ bits of advice are necessary to achieve $1$-consistency.
\end{thm}
\begin{proof}[Proof Sketch for \autoref{thm:advicecomplexity}]
We first characterize the \textit{number of unique solutions} for \OCS sequences with length $T$.  
We leverage the structure of the problem, observing first that any feasible solution satisfies $\sum_{i=1}^{T} x_t = 1, 0 \leq x_t \leq 1 \ \forall t \in [T]$.  
We use a careful discretization of the decision space $[0,1]$ to the set $\{ 1/T, 2/T \dots, 1\}$ to capture two extreme cases for the possible solutions.  One feasible solution (rate constraints permitting) is to distribute purchasing decisions across all $T$ time steps, setting $x_t = 1/T$ at each step.  The other extreme is to group all purchasing decisions into a single time step $m$, such that $x_m = 1$ and $x_t = 0 : \forall t \not = m$.  This careful discretization allows us to derive a lower bound on the number of unique solutions using combinatorial principles.

Using this lower bound, we can then show that the growth in the number of unique solutions as $T$ grows is fast enough to overwhelm an advice model that is sublinear in $T$.  The intuition behind the proof hinges on the fact that as $T$ grows, any general advice model which uses $A_T$ bits of advice (where $A_T = o(T)$) can distinguish between $2^{A_T}$ inputs based on the advice.  As $T$ grows, the number of unique solutions will quickly exceed $2^{A_T}$, and by the pigeonhole principle, there must be at least two unique solutions that map to the same advice string.  The full proof is given in \autoref{sec:advice-comp}.
\end{proof}

The above result concludes that any algorithm with advice for \PrOb must read $\Omega(T)$ advice bits to achieve a competitive ratio of $1$; in the \textit{untrusted advice} setting, this corresponds to the advice necessary for $1$-consistency. Consequently, it naturally motivates the usage of black-box advice (as in~\cite{Antoniadis:20MTS, Christianson:22, Rutten:23, Christianson:23MTS}) which predicts the online decision at each time step.  
Below we formally define the black-box advice model we use for \OCS, and some assumptions.

\begin{definition}[Black-box advice model for \OCS]
     A learning-augmented algorithm $\ALG$ for \OCS receives advice of the form $\{ \hat{x}_t \}_{t \in [T]}$.  If the advice is correct, a naïve algorithm $\ADV$ which chooses $\hat{x}_t$ at each time step obtains the optimal cost/profit, i.e., for any valid \OCS instance $\mathcal{I}$, $\ADV(\mathcal{I}) = \OPT(\mathcal{I})$.

     We assume that regardless of whether the black-box advice is correct, it is always \textit{feasible}, i.e., it satisfies $\sum_{t = 1}^{T} \hat{x}_t = 1$ and $\hat{x}_t \in [0, d_t] \ \forall t \in [T]$. 
\end{definition}

We note here that it is not obvious that this black-box advice model is something that can be directly learned by a machine learning model in practice, and the assumption that the black-box advice is always feasible can similarly be questionable.  In practice, however, the model's generality means that it can accommodate any ``black-box predictor'', which can actually be a layer of abstraction over several components, such as an ML model plus a post-processing pipeline to enforce feasibility.  In \autoref{sec:eval}, we show that this ``black-box'' advice model is still useful in the real world by leveraging an off-the-shelf ML model for our empirical experiments.

Before we describe our main learning-augmented results, we note that the prediction lower bound shown in \autoref{thm:advicecomplexity} can be improved if aspects of the \OCS problem are relaxed or simplified.  Specifically, if rate constraints are removed (i.e., $d_t = 1 \ \forall t \in [T]$) and the cost/price functions are all linear, the problem's advice complexity seems to decrease substantially.  We provide some more discussion on this point in Appendix~\ref{apx:specialcase-advice}, although we do not explore this dynamic in detail since these simplifications are not relevant to the real-world applications we consider.  It would be very interesting to explore this trade-off between the problem's complexity and the amount of advice required in future work, particularly since \OCS seems to lie on a boundary between the simple advice used in the online search and the decision advice used in the online metric literature.

\subsection{\ROAdvmin: robustly incorporating black-box decision advice} \label{sec:ocsmin-adv}

We present a learning-augmented algorithm for \OCS called Ramp-On-Advice (\ROAdv, pseudocode in Algorithm~\ref{alg:ro-advice}). This algorithm combines the robust decision of \RORO ($\tilde{x}_t$) at each time step with the predicted optimal solution ($\hat{x}_t$) obtained from the black-box advice. \ROAdv takes a hyperparameter $\epsilon$, which parameterizes a trade-off between consistency and robustness (i.e., as $\epsilon$ approaches $0$, consistency improves and robustness degrades).  Based on $\epsilon$, \ROAdv sets a combination factor $\lambda \in [0,1]$, which determines the decision portion that it should take from each subroutine (i.e., $\lambda$ from the black-box advice and $(1- \lambda)$ from the robust solution).
We note that the exact definitions of $\epsilon$ and $\lambda$ are different in the minimization and maximization settings{\color{blue}, and we defer discussion of the maximization setting to \autoref{apx:ocsmax-adv}.}

\begin{algorithm}[h]
	\caption{Ramp-On-Advice meta-algorithm for \textit{learning-augmented} \OCS (\texttt{RO-Advice})}
	\label{alg:ro-advice}
	\begin{algorithmic}[1]
		\State \textbf{input:} The black-box advice $\{ \hat{x}_t \}_{\forall t \in [T]}$, competitive robust advice $\{ \tilde{x}_t \}_{\forall t \in [T]}$ (e.g., given by $\RORO$), combination factor $\lambda$ (see Definitions~\ref{dfn:roadvmin} and \ref{dfn:roadvmax}).
		\While{cost/price function $g_t(\cdot)$ is revealed and $w^{(t-1)} < 1$}
		\State obtain untrusted advice $\hat{x}_t$;
        \State obtain robust advice $\tilde{x}_t$;
        \State set the online decision as $x_t = \lambda \hat{x}_t + (1 - \lambda) \tilde{x}_t$;
	\State update the utilization $w^{(t)} = w^{(t-1)} + x_t$;
	\EndWhile
	\end{algorithmic}
\end{algorithm}

We start by describing our instantiation of \ROAdv for \OCSmin (\ROAdvmin), which includes definitions of $\epsilon$, $\lambda$, and the robust instantiation of \ROROmin as a subroutine.

\begin{definition}[\ROAdv instantiation for \OCSmin (\ROAdvmin)]\label{dfn:roadvmin}
    Let $\epsilon \in [0, \alpha - 1]$, where $\alpha$ is the robust competitive ratio defined in \eqref{eq:alpha}.  \ROAdvmin then sets a combination factor $\lambda~\vcentcolon=~\left( \alpha - 1 - \epsilon \right) \cdot \frac{1}{\alpha - 1}$, which is in $[0, 1]$.  The robust advice $\{ \tilde{x}_t \}_{\forall t \in [T]}$ is given by the \RORO instantiation for \OCSmin (\ROROmin), given in Definition~\ref{dfn:roromin}.
\end{definition}

In the following theorem, we state the consistency and robustness bounds for the \ROAdvmin meta-algorithm.  For brevity, we relegate the full proof to \autoref{sec:advice-proofs-min}.

\begin{thm}\label{thm:advice-consist-robust-min}
Given a parameter $\epsilon \in [0, \alpha - 1]$, where $\alpha$ is defined as in~\eqref{eq:alpha},\\ \texttt{RO-Advice-min} is $(1 + \epsilon)$-consistent and $\left( \frac{\frac{(U + 2 \beta)}{L} (\alpha - 1 - \epsilon) + \alpha \epsilon}{( \alpha - 1) } \right)$-robust for \OCSmin. 
\end{thm}

Note that as $\epsilon \rightarrow \alpha - 1$, both the consistency and robustness bounds approach $\alpha$, which is the optimal competitive bound shown for the setting without advice.  When $\epsilon \to 0$, the robustness bound degrades to $\nicefrac{U + 2\beta}{L}$, but the consistency bound implies that perfect advice allows \ROAdvmin to obtain the optimal solution. {\color{blue} Although the result in \autoref{thm:lowerboundmin} implies that the consistency-robustness trade off of \ROAdvmin is optimal as $\epsilon \rightarrow \alpha - 1$, it is not clear whether the trade off is \textit{Pareto-optimal} for any setting of $\epsilon$.  It would be very interesting to consider a corresponding lower bound on consistency-robustness in this setting with predictions.

It is notable that a relatively simple linear combination of the untrusted advice and the robust online decision simultaneously achieves bounded consistency and robustness, particularly since prior works have shown that such a technique fails to achieve bounded robustness for other online problems.  In \OCS, our learning-augmented technique can take advantage of structure in the problem provided by the deadline constraint -- namely, since the deadline constraint \textit{must} be satisfied by any valid solution, we can derive upper and lower bounds on the cost of any valid solution.  Since any cost function must have derivative $\ge L$, the cost of any solution is lower bounded by $L$.  Conversely, since any cost function must have derivative $\le U$ and the worst-case switching cost is exactly $2\beta$, the cost of any solution is upper bounded by $U+2\beta$ (see \autoref{sec:probform}).  These bounds are exactly reflected in the worst-case robustness ratio.

}

%% file: 5-experiments.tex
This section uses carbon-aware electric vehicle (EV) charging as a case study and compares our proposed algorithms for  \PrObmin against existing algorithms from the literature.

\subsection{Experimental Setup} \label{sec:expsetup}

We consider a carbon-aware EV charging system that supplies a requested amount of electricity to a customer's EV while minimizing the total CO$_2$ emissions of the energy used.
For simplicity, we consider a single EV charger connected to the electric grid and a local solar installation, such as a solar parking canopy~\cite{EPA:17}.  
The simulated charger is an \textit{adaptive} Level 2 AC charger with a peak charging rate of 19kW \cite{DOT:chargers} that can modulate its charging rate between 0kW and 19kW.
The charger first uses the electricity from the local solar and resorts to the electric grid to meet the deficit. We assume that local solar electricity has zero operational carbon intensity, while grid-supplied electricity has a \textit{time-varying carbon intensity} based on the mix of generation sources for the region~\cite{sukprasert2023quantifying}.

We construct instances of \OCSmin as follows: Each EV arrives with a charging demand~(in~kWh), a connection time, and a disconnection time.  The objective of the charging system is to selectively ramp the charging rate up or down to meet the full demand before the deadline (disconnection time) while minimizing total carbon emissions. 
The switching cost smooths the charging rate and penalizes large fluctuations, which accelerate lithium-ion battery degradation~\cite{Zhang:06} (see \autoref{fig:in-action} for ``smoothing'' effect illustration).
The rate constraints $\{d_t\}_{t \in [T]}$ are set to $1$ if the total charging demand is less than 19 kWh, and otherwise it is set to reflect the fraction of the total demand that our simulated Level 2 AC charger can meet during one hour (one time slot) of charging.\footnote{Our setup assumes a slow AC charger, and the EV can always accept its full charging rate. Considering a DC fast charger, which limits the charging rate when the battery is above a certain charge level~\cite{ChargePoint}, is outside the scope of this work.}

\paragraph{EV charging data.} We use the Adaptive Charging Network Dataset (ACN-Data)~\cite{Lee:19:ACNData}, which provides real EV charging data for $\sim$16,000 EV charging sessions, including arrival time, departure time, and energy demand (in kWh) for each session.  We focus on charging sessions that are $5$ hours or longer, since the carbon intensity of the electric grid does not change over shorter time periods, and carbon savings can only be realized when charging has temporal flexibility.

\paragraph{Carbon data traces.}  We use historical grid carbon intensity data for the California ISO (CAISO), obtained via Electricity Maps~\cite{electricity-map}, which provides hourly \textit{average carbon emissions} of the CAISO grid, expressed in grams of CO$_2$ equivalent per kilowatt-hour (gCO$_2$eq/kWh).%

\paragraph{Carbon forecasts.}  To evaluate our learning-augmented \texttt{RO-Advice} algorithm, we use an \textit{``off-the-shelf''} predictive model called \texttt{CarbonCast}~\cite{Maji:22:CC}.  \texttt{CarbonCast} is a CNN-LSTM-based ML architecture that accurately forecasts the hourly carbon intensity of the grid in up to 96-hour windows. \texttt{CarbonCast} provides pre-trained models and training data for 13 regions, including CAISO.

\paragraph{Solar data traces.} To simulate the impact of solar installations of varying sizes, we obtain historical solar data from the National Solar Radiation Database (NSRDB)~\cite{Sengupta:18:NSRDB}.  The NSRDB provides satellite measurements of hourly solar irradiation, including diffuse horizontal irradiance (DHI) and direct normal irradiance (DNI), and solar zenith angle, at a 2 km spatial resolution.  
Using this data, we can approximate the historical generation of a simulated solar system in a particular location using the following method~\cite{dobos2014pvwatts}. 
We use the optimal panel tilt angle for the location to compute the \textit{global tilted irradiance} (GTI) on the panel's surface based on the DHI, DNI, and solar elevation angle as
$\text{GTI} = \text{DNI} \cdot \sin(\text{solar elevation angle}) + \text{DHI}$.
The DC rating of a solar cell is customarily given under ``ideal'' irradiance conditions of 1000 W/m$^2$.  Given the DC rating of an entire solar system, we can compute the estimated generation as follows:
\[
\text{Generation (in kW)} = \text{DC rating} \cdot \frac{\text{GTI}}{1000 \text{W/m}^2} \cdot \textit{AC Inverter Efficiency} \cdot \left( 1 - \textit{System Losses} \right).
\]
We assume that the AC inverter is 95\% efficient and the panel-to-output system losses are $\sim$14\% -- default values from the PVWatts solar calculator~\cite{NREL:PVWatts}.
In Appendix \autoref{fig:solar-vis}, we visualize a two-week sample of the grid's carbon intensity and simulated solar generation for a 15 kW solar system. 

\paragraph{Parameter settings} Most algorithms we test accept $L$ and $U$ as parameters for their threshold functions.  To set these parameters, we examine a month's worth of grid carbon intensity values leading up to the connection time for a given charging session, setting $L$ and $U$ to be the minimum and maximum observed carbon intensity values over the past month, respectively.

First, to gauge the impact of local solar generation, we tested differently sized simulated local solar systems with DC-rated nameplate capacities of $0$, $5$, $10$, and $15$ kW.  A rating of $0$ implies that the charger can only draw current from the grid, while a larger DC rating represents greater impact from local solar (during periods when the sun is shining).  The switching cost is fixed at $\beta = 20$.
Second, to gauge the impact of the switching cost magnitude, we run a set of experiments to test different values for $\beta$ in the range $[0, 40]$.  Local solar generation is fixed at $0$ kW.  Note that when $\beta = 0$, \PrOb reduces exactly to the one-way trading problem.

Finally, to measure the impact of prediction error, we generate a spectrum of simulated black-box advice, ranging from perfect advice to completely adversarial advice, and test $\ROAdv$ for different values of $\epsilon$.  
Unlike other experiments, in these experiments, $\ROAdv$ does not use \texttt{CarbonCast} predictions.  Instead, we define a simulated advice vector $\hat{\mathbf{x}} = (1-\zeta) \mathbf{x}^\star + \zeta \breve{\mathbf{x}}$, where $\mathbf{x}^\star$ represents the decisions made by an optimal solution, $\breve{\mathbf{x}}$ represents the decisions made by a solution which maximizes the objective (rather than minimizing it), and $\zeta$ is an \textit{adversarial factor} in $[0,1]$.  When $\zeta = 0$, the simulated prediction is exactly correct ($\hat{\mathbf{x}} = \mathbf{x}^\star $), and when $\zeta = 1$, the simulated prediction is fully adversarial ($\hat{\mathbf{x}} = \breve{\mathbf{x}}$). We note that $\breve{\mathbf{x}}$ is adversarial from the perspective of the objective, although it is still a feasible solution for the problem.

\paragraph{Benchmark algorithms}
To evaluate the performance of the tested algorithms, we compute the offline optimal solution for each charging session (instance) using a numerical solver~\cite{SciPy}.  This allows us to report the empirical competitive ratio for each algorithm on each charging session.  
We compare \RORO and \ROAdv against three benchmark algorithms, which are not developed for solving \OCSmin. Still, we construct these algorithms to represent simple adaptations of existing ideas for adjacent problems and evaluate how our algorithms measure up.

The first benchmark algorithm is a \textit{carbon-agnostic} approach, which starts charging the EV at the fastest rate as soon as it is plugged in.  This carbon-agnostic approach simulates the behavior of a traditional, non-adaptive charger.  
We also compare \RORO against carbon-aware \textit{switching-cost-agnostic} algorithms adapted for \PrOb. 
We have two algorithms of this type, each drawing from existing online search methods in the literature.  
The first algorithm is a \textit{simple threshold algorithm}, which uses the $\sqrt{UL}$ threshold value first presented for the online search~\cite{ElYaniv:01}.  In our experiments, this algorithm chooses to charge the EV (at the maximum rate) whenever the carbon intensity is less than or equal to $\sqrt{UL}$.  This approach simulates the behavior of the simplest carbon-aware charger, which is not adaptive but can switch on or off in response to a carbon signal.
The other switching-cost-agnostic algorithm tested is the threshold-based one-way trading algorithm, \OWT, shown by~\cite{SunZeynali:20} and described in \autoref{sec:OTA}. This algorithm chooses to ramp to a specific charging rate $x_t$ based on a threshold function $\varphi$.  Note that when $\beta = 0$, \RORO reduces exactly to this one-way trading algorithm. In all experiments, we report the empirical competitive ratio of different algorithms, i.e., closer to 1 is a better competitive ratio, and 1 is the best. 

\begin{figure}[b]
\begin{minipage}{0.38\textwidth}
\centering \vspace{-2em}
\includegraphics[width=\linewidth]{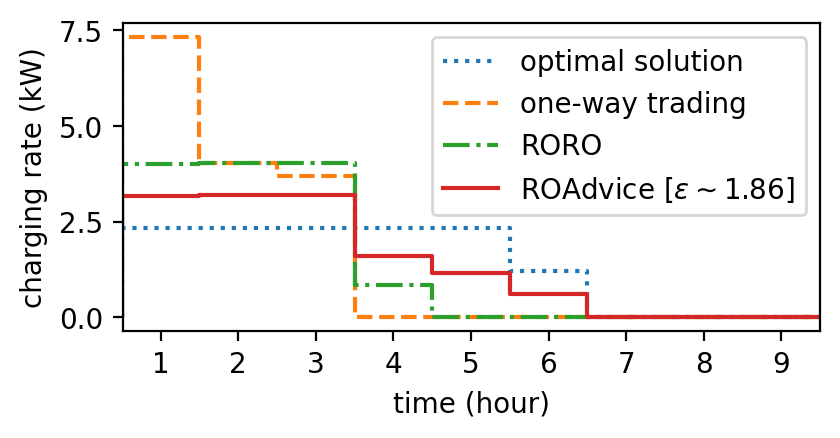}\vspace{-1.2em}
\caption{The decisions made by different algorithms and the optimal solution during a 9-hour EV charging session with 10kW of local solar capacity and $\beta = 30$.  The total energy requested is $\sim$12.9 kWh.}\label{fig:in-action}
\end{minipage}\hfill
\begin{minipage}{0.6\textwidth}
\small
\captionsetup{type=table} %
\begin{tabular}{l|ll|ll|ll}
\hline
\begin{tabular}{@{}c@{}}\footnotesize improvement \\ \footnotesize against...\end{tabular} & \multicolumn{2}{l|}{\begin{tabular}{@{}c@{}}simple \\ threshold\end{tabular}} & \multicolumn{2}{l|}{\begin{tabular}{@{}c@{}}one-way \\ trading\end{tabular}} & \multicolumn{2}{l}{\RORO} \\ \hline
                           & avg.  & 95$^\text{th}_\text{\%ile}$        & avg.        & 95$^\text{th}_\text{\%ile}$ & avg.        & 95$^\text{th}_\text{\%ile}$        \\ \hline
\RORO       &  $52.4$\%  &  $54.1$\%  &  $12.1$\%  &  $3.6$\% & -- & --  \\ \hline
\ROAdv   &  $66.4$\%  &  $73.3$\%  &  $41.4$\%  &  $46.2$\%  &  $33.4$\%  &  $44.9$\%  \\ \hline
\end{tabular} \vspace{1em}
\caption{Overall performance improvement of \ROROmin and \ROAdvmin ($\epsilon\sim 1.86$) over baseline threshold $\sqrt{UL}$ algorithm and one-way trading algorithm.  The empirical competitive ratio measures performance, and we report the average and the 95$^\text{th}$ percentile (i.e., worst-case) improvement for all experiments.}\label{tab:top-line-results}
\end{minipage}
\vspace{-0.5cm}
\end{figure}

\subsection{Experimental Results} \label{sec:eval-results}
Table~\ref{tab:top-line-results} summarizes the most notable overall results, aggregating over the first two experiments. We observe that the average empirical competitive ratio of \ROROmin is a $57.3$\% improvement on the carbon agnostic method, a $52.4$\% improvement on the simple threshold algorithm, and a $12.1$\% improvement on the one-way trading algorithm.
Across the first two experiments, \ROAdv with $\epsilon \thicksim 1.86$ and predictions from \texttt{CarbonCast} shows a $33.4$\% improvement on \ROROmin, a $69.2$\% improvement on the carbon-agnostic method, a $66.4$\% improvement on the simple threshold algorithm, and a $41.4$\% improvement on the one-way trading algorithm. 

To further motivate the inclusion of the switching cost $\beta$ in this EV charging application, in \autoref{fig:in-action}, we plot a step function of the decisions made by several algorithms on a representative 9-hour charging session from Feb 2020.  We note that the switching cost $\beta$ has a ``smoothing'' effect on the charging schedule generated by our algorithms, which are \textit{switching-aware}.  This is in contrast to the schedule generated by the one-way trading algorithm, which front loads most of the charging into the first hour of the charging session.

\begin{figure*}[t]
    \vspace{-0.5em}
    \minipage{0.8\textwidth}
    \includegraphics[width=\linewidth]{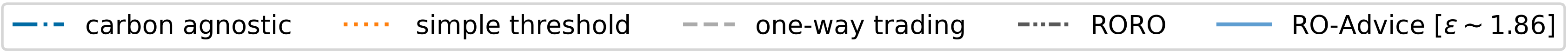}
    \endminipage\hfill\\
	\minipage{0.24\textwidth}
	\includegraphics[width=\linewidth]{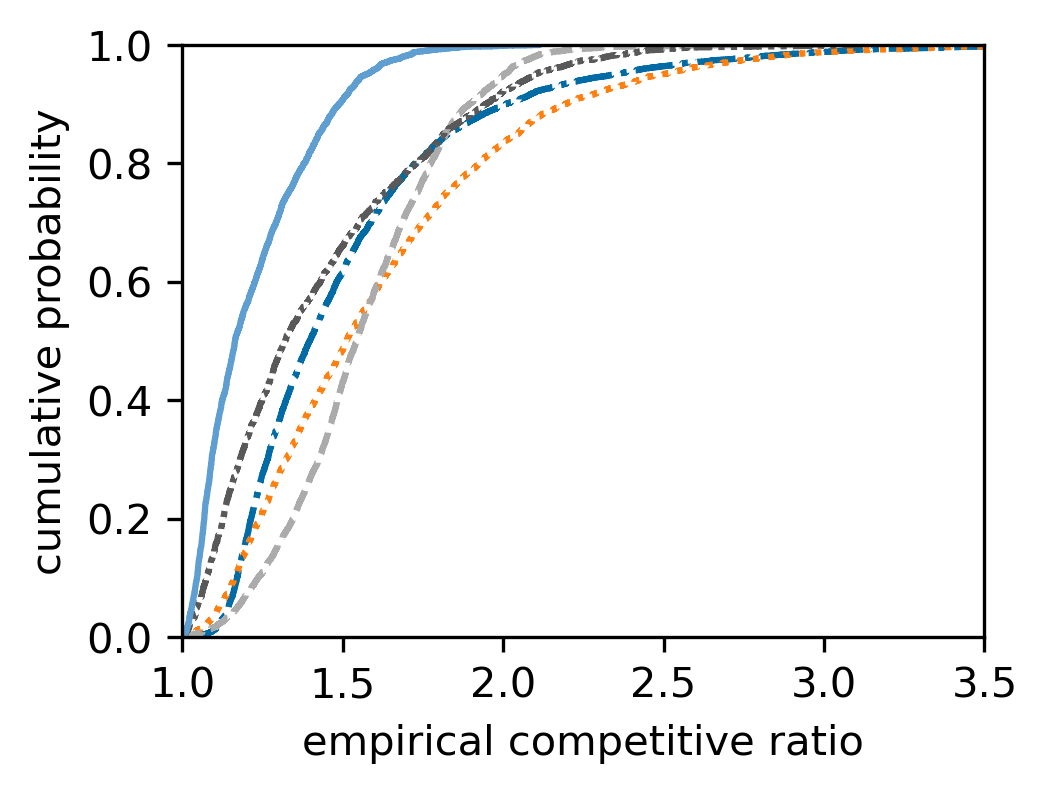}\vspace{-1.2em}
    \caption*{(a) 0 kW solar}
	\endminipage\hfill
	\minipage{0.24\textwidth}
	\includegraphics[width=\linewidth]{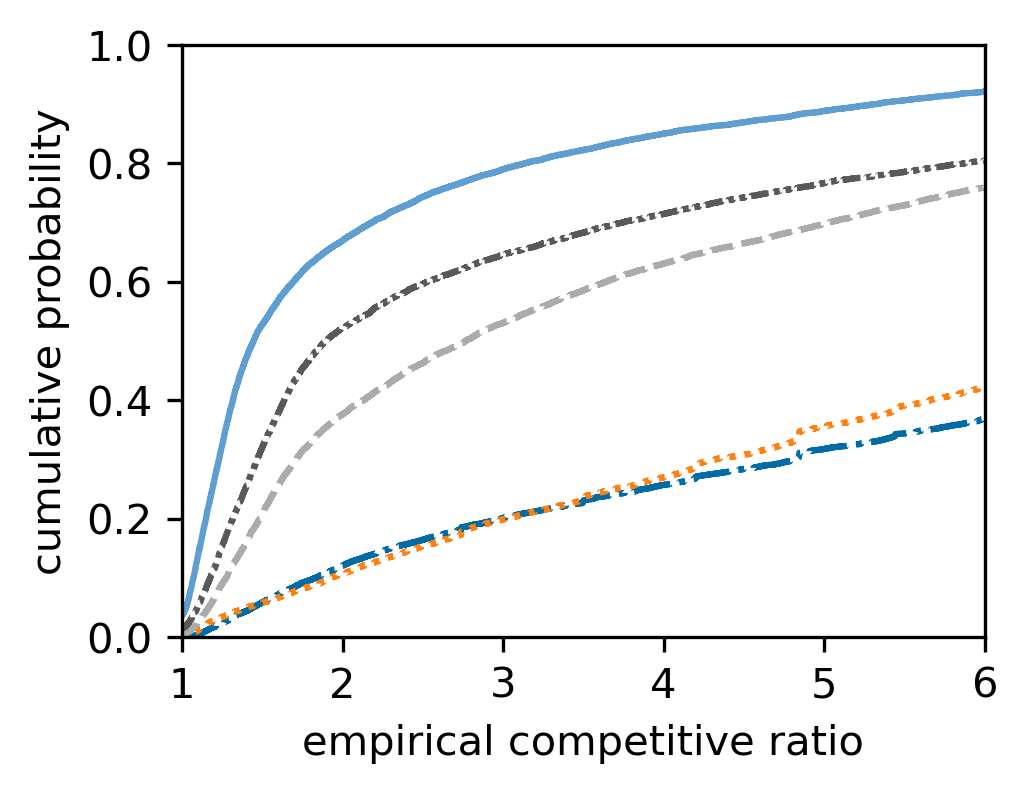}\vspace{-1.2em}
    \caption*{(b) 5 kW solar}
	\endminipage\hfill
    \minipage{0.24\textwidth}
	\includegraphics[width=\linewidth]{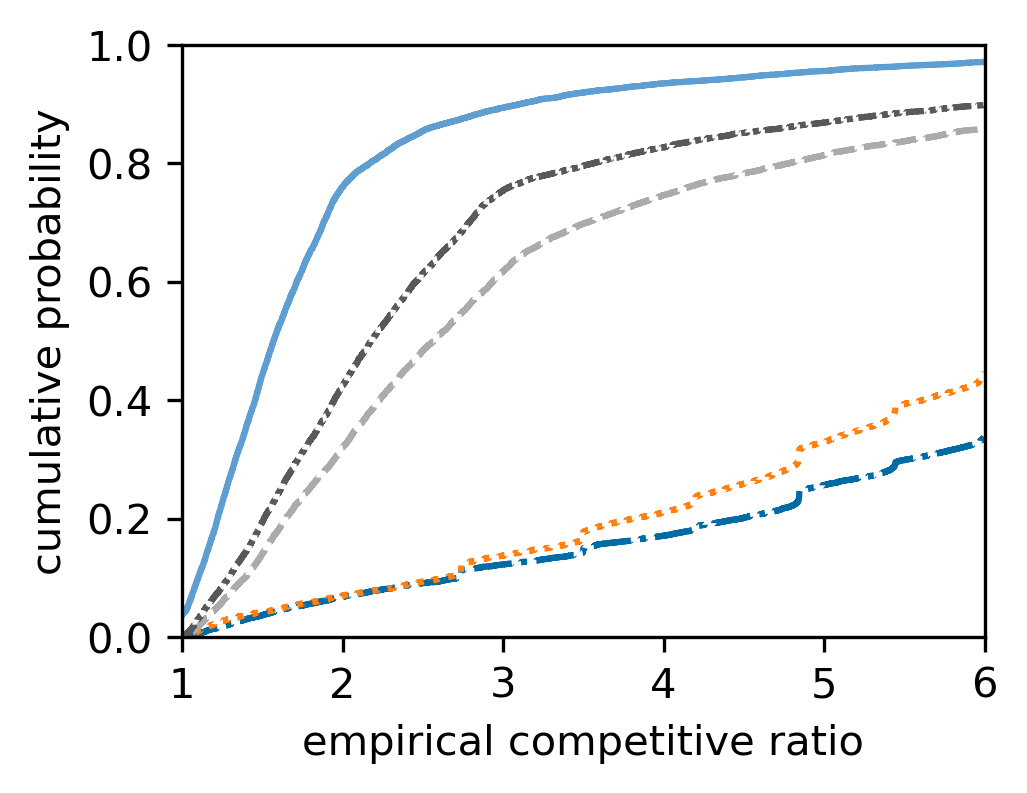}\vspace{-1.2em}
    \caption*{(c) 10 kW solar}
	\endminipage\hfill
    \minipage{0.24\textwidth}
	\includegraphics[width=\linewidth]{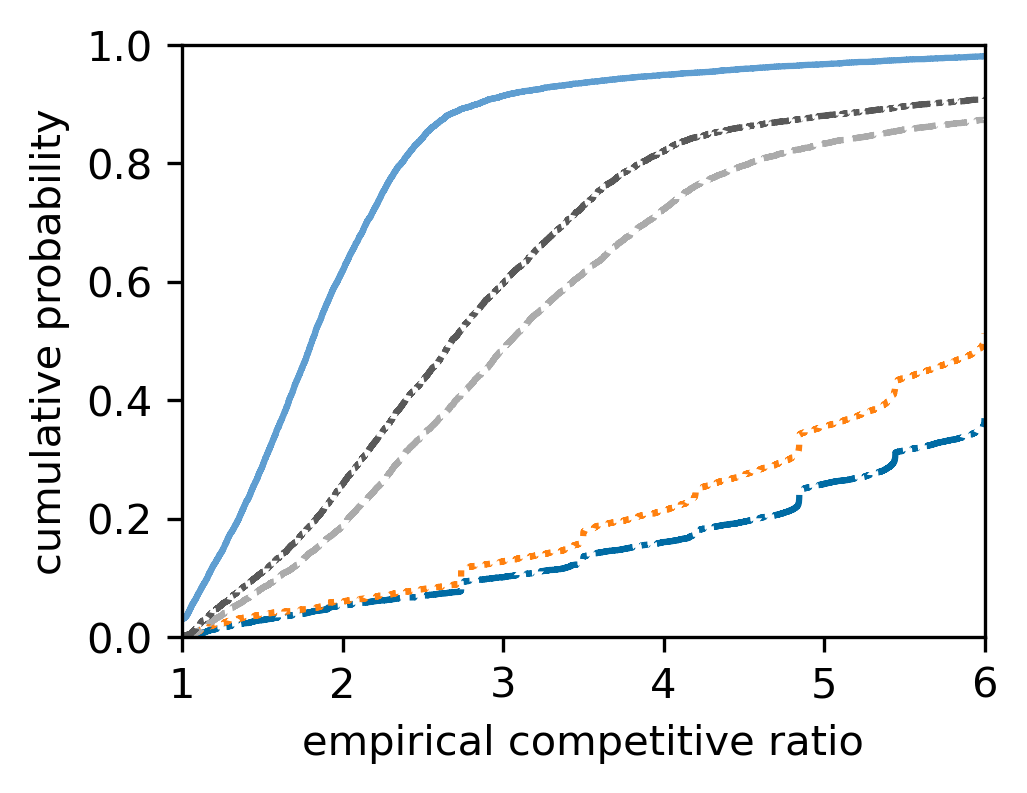}\vspace{-1.2em}
    \caption*{(d) 15 kW solar} %
	\endminipage\hfill
    \vspace{-1em}
    \caption{Cumulative distribution functions (CDFs) of empirical competitive ratios for all tested algorithms, in experiments testing the impact of \textit{differently-sized local solar generation}.  Each simulated solar system is described using the DC rating.  Switching cost is fixed $\beta = 20$, and \ROAdv uses \texttt{CarbonCast} predictions.}\label{fig:solarGen}
    \vspace{-1em}
\end{figure*}

We now proceed to explain the details of our experiments. \autoref{fig:solarGen} shows a series of cumulative distribution function (CDF) plots. In this experiment, we test all algorithms for differently sized local solar installations, ranging from $0$ kW (no solar) to $15$ kW.  The switching cost $\beta$ is fixed to $\beta = 20$, and \ROAdv uses \texttt{CarbonCast} predictions with $\epsilon \thicksim 1.86$ (this value of $\epsilon$ gives $\lambda = 0.5$).  By testing different simulated solar system ratings, this experiment gauges how the makeup of the grid can affect performance (i.e., more or less renewable volatility).  Unsurprisingly, \ROAdv performs well across the board, improving on the next closest algorithm \ROROmin by $33.4$\% in the average case and $52.4$\% in the 95th percentile, respectively.  \ROROmin (our robust algorithm without predictions) compares favorably to the other algorithms, achieving an average empirical competitive ratio which is a $57.4$\% improvement on the carbon-agnostic method, a $52.4$\% improvement on the simple threshold algorithm, and a $12.1$\% improvement on the one-way trading algorithm.

We further note that in \autoref{fig:solarGen}(a) (experiment with no local solar generation), while \ROROmin outperforms the one-way trading algorithm on average, the tail of the CDF implies that \ROROmin's empirical competitive ratio is worse on roughly ${\thicksim} 20$\% of the instances.  Since the one-way trading algorithm is \textit{switching-agnostic}, there are cases where \ROROmin is too conservative for a given sequence -- in other words, \ROROmin avoids ramping up to avoid a large switching penalty and eventually is forced to fulfill the charge at the end of the sequence (a compulsory trade). This dynamic disappears when simulated local solar generation is available because \ROROmin's conservative design matches well with the intermittent presence of local solar.

\begin{figure*}[t]
	\minipage{0.48\textwidth}
	\includegraphics[width=\linewidth]{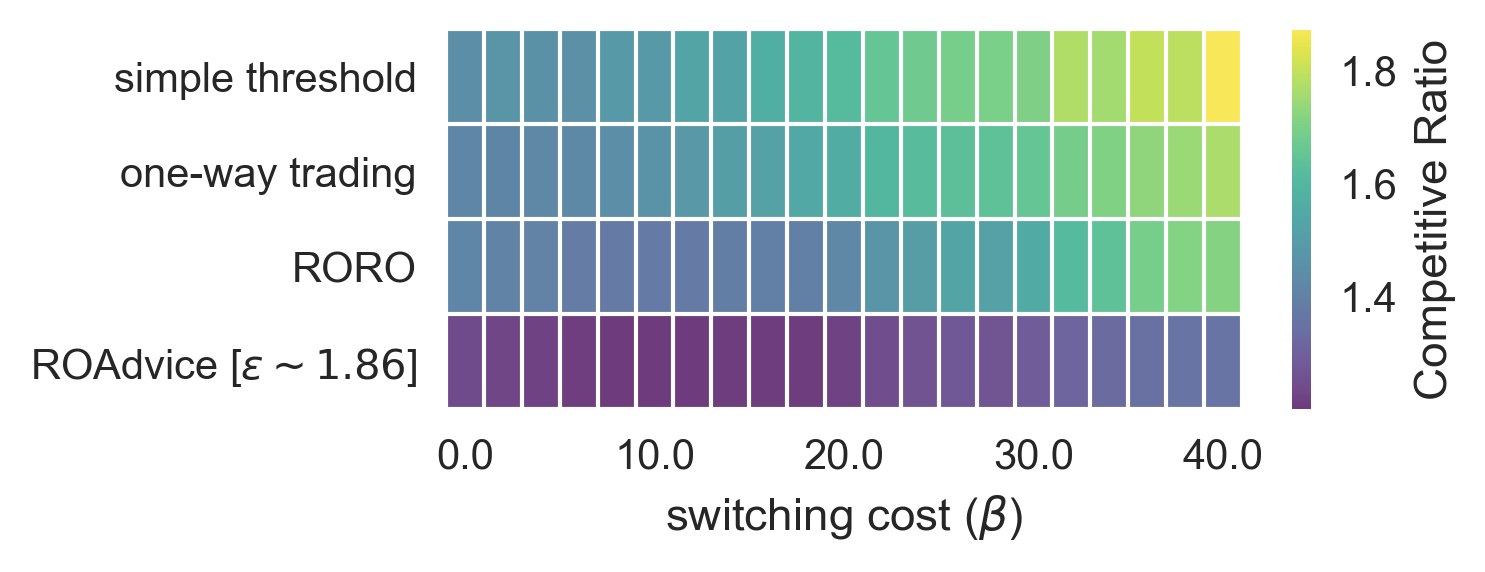}\vspace{-1.2em}
    \caption{Heat map of empirical competitive ratios for algorithms in a set of experiments testing the impact of different \textit{switching costs}.  Local solar generation is fixed to $0$, and \ROAdv uses \texttt{CarbonCast} predictions with $\epsilon \thicksim 1.86$.  Note that the simple threshold and one-way trading algorithms are \textit{switching-oblivious}, while our proposed \RORO and \ROAdv algorithms are aware of the switching cost.} \label{fig:switching}
	\endminipage\hfill
	\minipage{0.48\textwidth}
	\includegraphics[width=\linewidth]{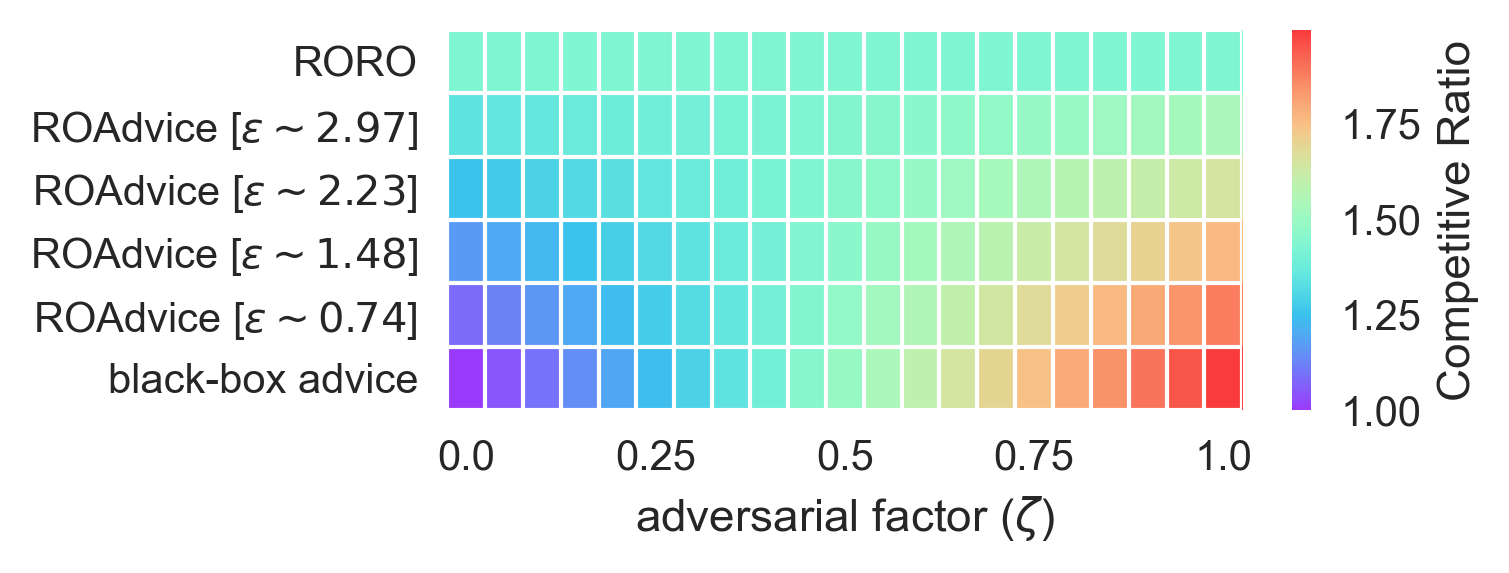}\vspace{-1.2em}
    \caption{Heat map of empirical competitive ratios for algorithms in a set of experiments testing the consistency and robustness of \ROAdvmin by simulating advice from perfect ($\zeta = 0$) to adversarial ($\zeta = 1$).  No local solar, and four version of \ROAdvmin, with $\epsilon \in [{\thicksim} 2.97, {\thicksim} 2.23, {\thicksim} 1.48, {\thicksim} 0.74]$.  The black-box advice algorithm naïvely plays the simulated advice at each time step.} \label{fig:advice}
	\endminipage\hfill
    \vspace{-1em}
\end{figure*}

In the second experiment, reported as a heat map of empirical competitive ratios in \autoref{fig:switching}, we test our proposed algorithms against the \textit{switching-agnostic} algorithms for different switching costs $\beta$ in the range from $0$ to $40$ (see Appendix \autoref{fig:switching-stddev} for a more detailed version of this plot). By testing different values for $\beta$, this experiment tests how an increasing switching cost impacts the performance of \ROROmin and \ROAdvmin with respect to other algorithms that do not consider the switching cost in their designs. We observe that as the switching cost $\beta$ increases, the empirical performance of all algorithms degrades (indicated by the luminance of the plot). Notably, comparing \ROROmin and \ROAdvmin against the simple threshold algorithm and the one-way trading algorithm, our proposed algorithms exhibit competitive ratios that degrade \textit{slower} as a function of the increasing competitive ratio, which validates our theoretical results. 
\ROAdv outperforms across the board, improving on the next closest algorithm \ROROmin by $14.3$\% in the average case and $27.9$\% in the 95th percentile, respectively.  \ROROmin also outperforms the other tested algorithms, achieving an average improvement of $11.4$\% on the simple threshold algorithm and a $8.5$\% improvement on the one-way trading algorithm.  Note that when $\beta = 0$, \ROROmin is equivalent to the one-way trading algorithm, which the experimental results confirm.

Finally, in the third experiment, we rigorously test the consistency and robustness of \ROAdvmin and include \RORO and the naïve black-box advice as baselines.  We test four different variants of \ROAdvmin, parameterized by $\epsilon \in [{\thicksim} 2.97, {\thicksim} 2.23, {\thicksim} 1.48, {\thicksim} 0.74]$ (these values of $\epsilon$ correspond to $\lambda$ parameters $[0.2, 0.4, 0.6, 0.8]$).  As described in \autoref{sec:expsetup}, \ROAdvmin and the naïve black-box advice receive simulated advice as input, which is manipulated according to an adversarial factor $\zeta$ (where $\zeta = 0$ is perfect advice and $\zeta = 1$ is fully adversarial advice).
We test different values of $\zeta$ in the range $[0,1]$ and visualize these results as a heat map in \autoref{fig:advice} (see Appendix \autoref{fig:advice-stddev} for a more detailed version).  As discussed in the theoretical results, we observe opposing dynamics on either side of the heat map.  When $\zeta = 0$, setting $\epsilon \to 0$ yields increasingly better performance as the (accurate) advice is ``trusted'' more.  On the other hand, when $\zeta = 1$, setting $\epsilon \to 1$ can quickly increase the competitive ratio well beyond what is achieved by the robust \RORO algorithm since trusting a prediction that is actually incorrect or even adversarial can significantly degrade performance.  Importantly, \autoref{fig:advice} shows that, depending on the desired $\epsilon$, \ROAdvmin maintains a competitive ratio bounded away from the naïve worst-case regardless of the prediction quality, while significantly improving the average-case performance when predictions are correct or nearly correct, showing that it can achieve the ``best of both worlds''.

%% file: 6-conclusion.tex
Motivated by sustainability applications such as carbon-aware EV charging, we introduce and study online conversion with switching costs (\OCS), which bridges gaps between several existing online problems. 
It is the first online optimization problem to simultaneously capture a continuous decision space, long-term constraints, and switching costs.  Our main results introduce and analyze the \RORO framework, which achieves the optimal deterministic competitive ratio for the problem.  We also consider designing learning-augmented algorithms in this setting, first showing that existing advice models for related online search problems are insufficient to achieve improved consistency.  We prove an advice complexity lower bound, which shows that advice growing linearly in the length of the sequence is necessary. This result motivates our design of learning-augmented algorithms that take advantage of untrusted black-box advice (such as ML-based predictions of decisions), which achieve bounded consistency and robustness.  Finally, we empirically evaluate our proposed algorithms using a carbon-aware EV charging case study.

There are a number of interesting directions in which to continue the study of \OCS.  In particular, there are several open theoretical questions.  Considering a target application of carbon-aware load shifting, some workloads can be shifted both spatially and temporally~\cite{Sukprasert:23, radovanovic2022carbon, bashir2021enabling}, which would require a generalization to a multidimensional setting where cost functions and decisions are in $\mathbb{R}^d$ (e.g., each dimension corresponds to a location).  Such a generalization is highly non-trivial, as there are no existing works that consider online search-type deadline constraints in a multidimensional setting, and the threshold-based algorithm design we generalize in this work has not historically been applied to multidimensional problems.  As this paper considers a problem that bridges the gap between online search problems and a special case of an online metric problem (convex function chasing in $\mathbb{R}$), any multidimensional generalization naturally prompts further questions about possible extensions of other online metric problems, such as metrical task systems (\texttt{MTS}).  These questions are substantial additions that would be very interesting to explore in future work.

%% file: Z-appendix.tex
\section*{Appendix}

\renewcommand{\thefigure}{A\arabic{figure}}
\renewcommand{\thetable}{A\arabic{table}}
\renewcommand{\thealgorithm}{A\arabic{algorithm}}
\setcounter{figure}{0}
\setcounter{table}{0}
\setcounter{algorithm}{0}

\begin{table}[h]
	\caption{A summary of key notations }
 \vspace{-3mm}
	\label{tab:notations}
 	\small
	\begin{center}
		\begin{tabular}[P]{|c|l|}
            \hline
            \textbf{Notation} & \textbf{Description} \\
            \hline
            $C = 1$ & Total item to be bought (or sold) \\
			\hline
            $T$ & Deadline constraint; the player must finish purchasing (or selling) before $T$ \\
            \hline
            $t \in [T]$ & Current time step \\
            \hline
            $d_t : t \in [T]$ & Rate constraint at time step $t$ \\
            \hline
            $x_t \in [0, d_t]$ & Decision at time $t$. \\
			\hline
            $\beta$ & Switching cost coefficient of the cost incurred from decision change $\lvert x_t - x_{t-1} \rvert$ \\
			\hline
            $U$ & Upper bound on any partial derivative of $g_t(\cdot)$ that will be encountered  \\
            \hline
            $L$ & Lower bound on any partial derivative of $g_t(\cdot)$ that will be encountered \\
            \hline
            \hline
			$g_t(\cdot)$ & (\textit{Online input}) Cost/price function revealed to the player at time $t$ \\
			\hline
		\end{tabular}
	\end{center}
\end{table}

\begin{figure*}[h]
    \minipage{\textwidth}
     \vspace{-3mm}
    \includegraphics[width=\linewidth]{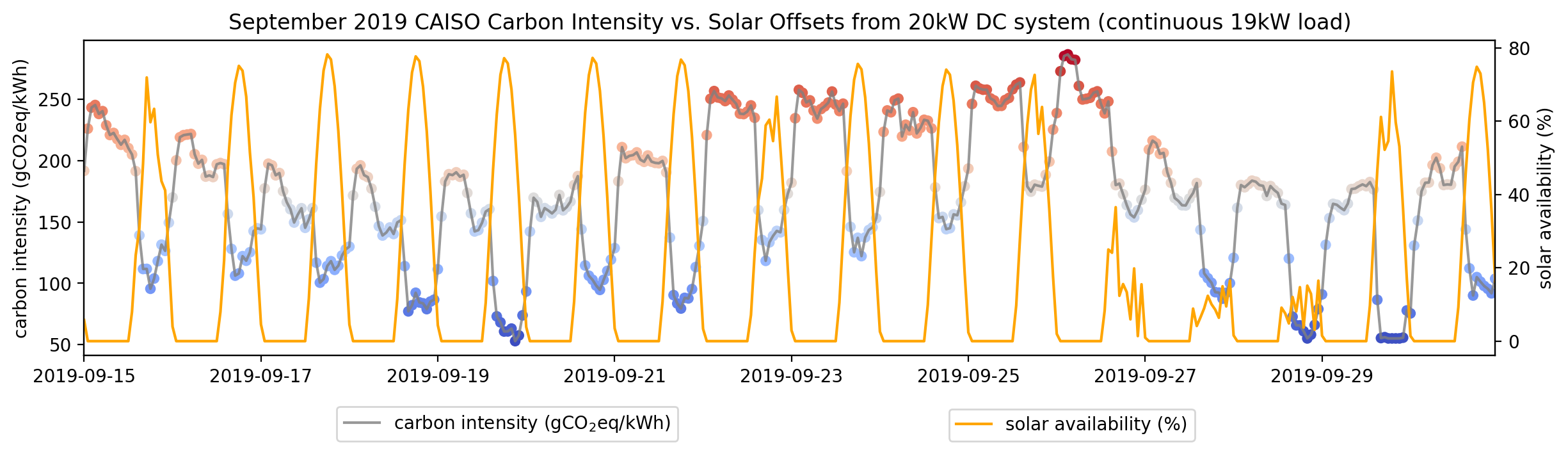}
    \endminipage\hfill\\
    \vspace{-1em}
    \caption{ (see \autoref{sec:eval}) Grid carbon intensity (in gCO$_2$eq/kWh) and simulated solar generation for the Caltech campus plotted over a two-week period in Sep. 2019, with one-hour granularity.  The left-hand $y$-axis, scatter dots, and gray line correspond to the carbon intensity on the CAISO grid provided by Electricity Maps~\cite{electricity-map}.  The right-hand $y$ axis and the orange line give the simulated generation of a 15 kWh (DC nameplate rating) solar system, with irradiation data provided by the NSRDB data set~\cite{Sengupta:18:NSRDB}.  Solar generation is expressed as a percentage with respect to a 19 kW continuous load, corresponding to the Level 2 AC charger simulated in our case study experiments (see \autoref{sec:expsetup} for more setup details).}
    \label{fig:solar-vis}
\end{figure*}

\clearpage

\begin{figure*}[t]
	\minipage{0.45\textwidth}
    \centering
	\includegraphics[width=\linewidth]{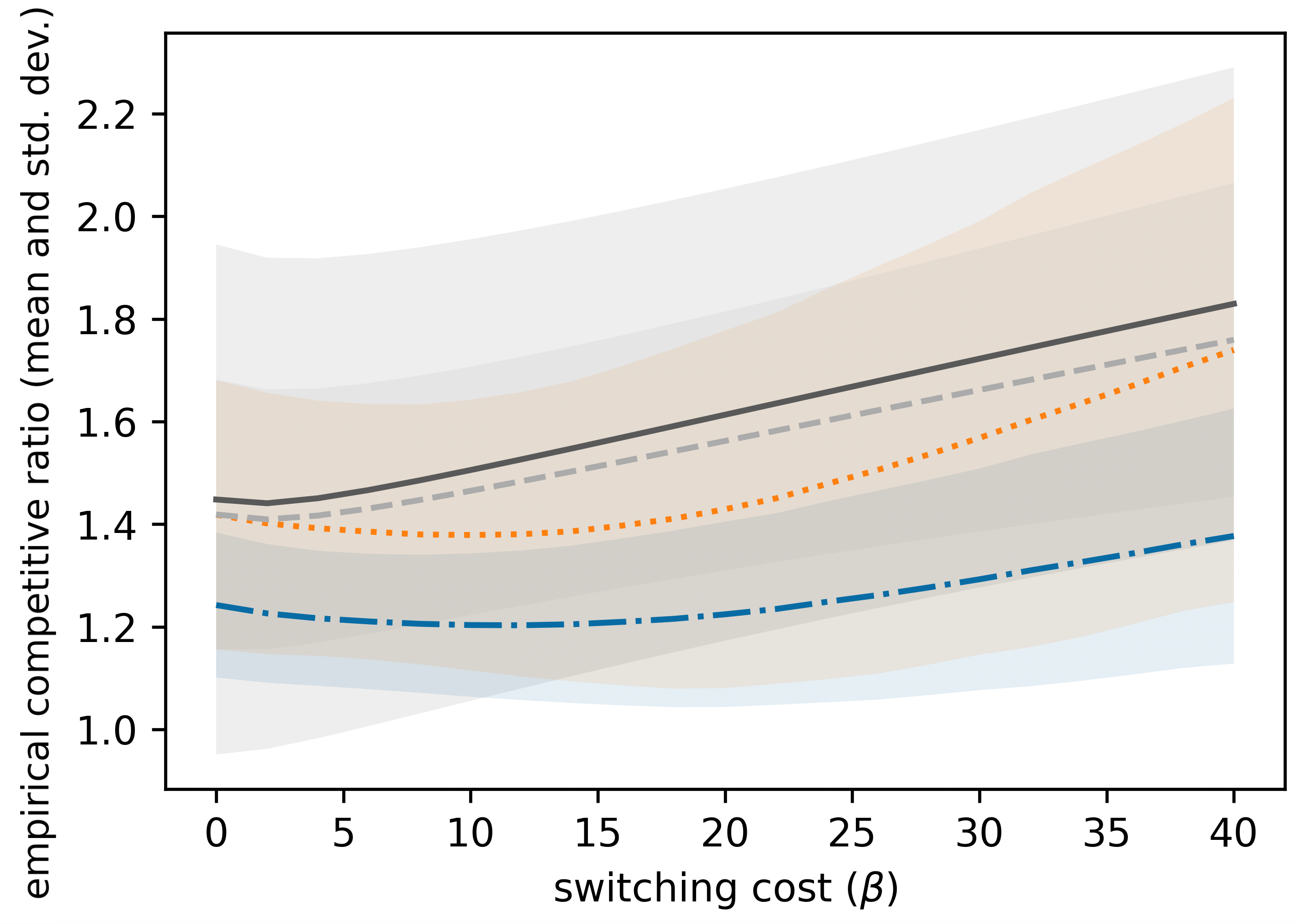}
    \includegraphics[width=0.8\linewidth]{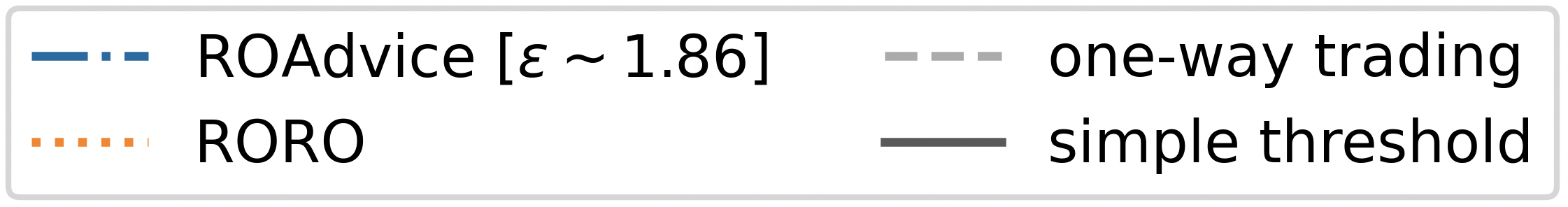}\vspace{-0.3em}
    \caption{ (see \autoref{sec:eval-results}) Average empirical competitive ratios and their standard deviations for experiments testing the impact of \textit{switching cost magnitude}.  Local solar generation is fixed to $0$, and \ROAdv uses \texttt{CarbonCast} predictions with $\epsilon \thicksim 1.86$.  Note that the simple threshold and one-way trading algorithms are \textit{switching-oblivious}, while our proposed \RORO and \ROAdv algorithms are aware of the switching cost.} \label{fig:switching-stddev}
	\endminipage\hfill
	\minipage{0.48\textwidth}
    \centering
	\includegraphics[width=\linewidth]{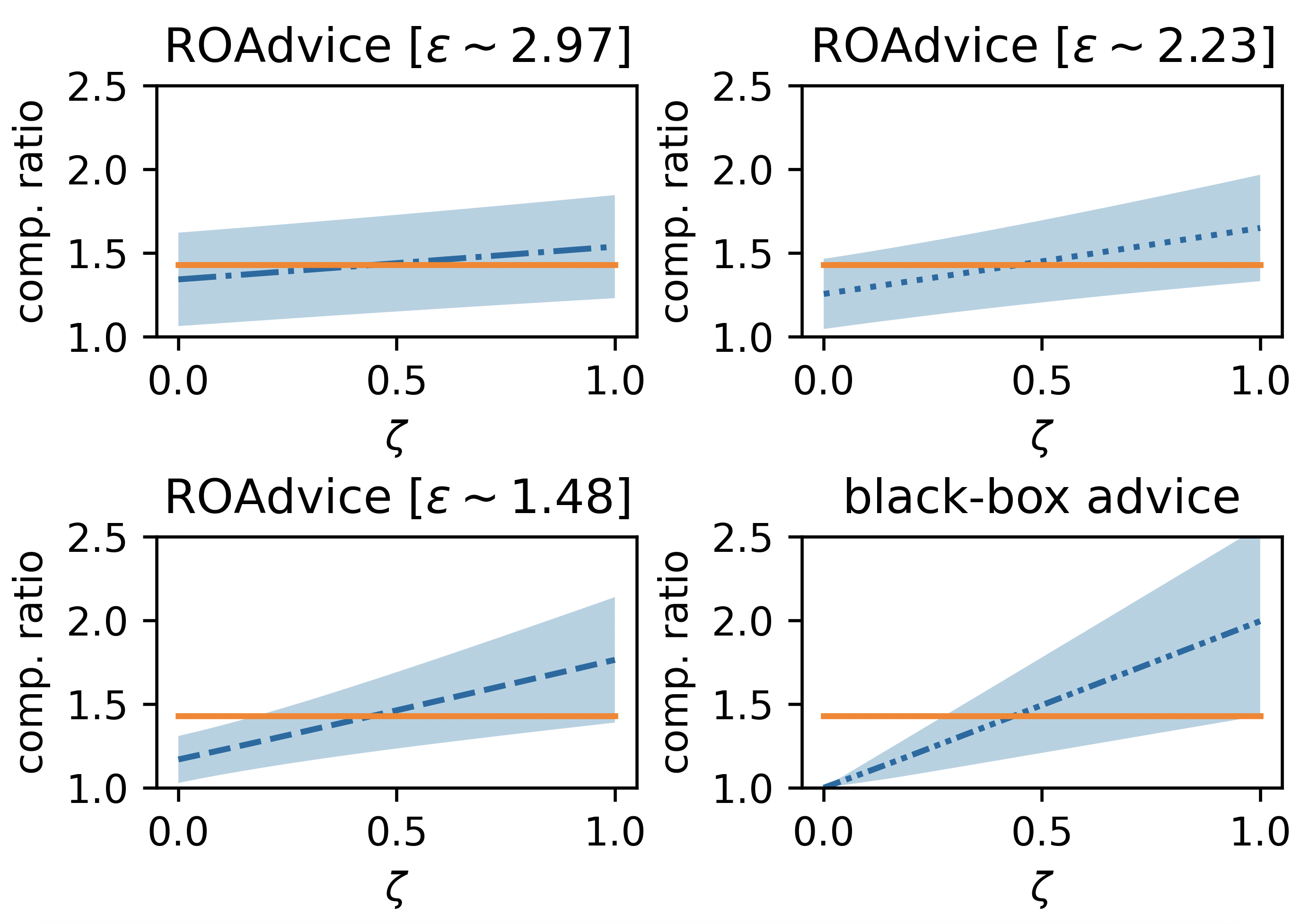}
    \includegraphics[width=0.6\linewidth]{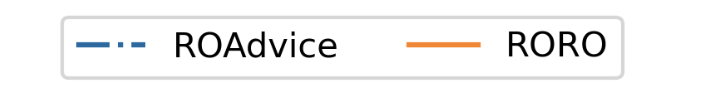}\vspace{-0.5em}
    \caption{ (see \autoref{sec:eval-results}) Average empirical competitive ratios and their standard deviations for experiments testing the consistency and robustness of \ROAdvmin by simulating advice from perfect ($\zeta = 0$) to adversarial ($\zeta = 1$).  No local solar, and three versions of \ROAdvmin, with $\epsilon \in [{\thicksim} 2.97, {\thicksim} 2.23, {\thicksim} 1.48]$.  The black-box advice algorithm naïvely plays the simulated advice at each time step.  \RORO (orange solid line) is a benchmark.} \label{fig:advice-stddev}
	\endminipage\hfill
\end{figure*}

\section{Deferred Pseudocode from Section~\ref{sec:roro}} 

\begin{algorithm}[h]
	\caption{Online Ramp-On, Ramp-Off (\RORO) framework instantiated for the online pause \& resume problem (\OPRmin)}
	\label{alg:warmup}
	\begin{algorithmic}[1]
		\State \textbf{input:} threshold $\{\phi_i\}_{i\in[k]}$, switching cost $\beta$;
        \State \textbf{initialization:} initial decision $x_0 = 0$, $i = 1$;
		\While{price $c_t$ is revealed and $i < k$}
		\State solve a \textbf{(ramping-on problem)} and obtain decision $x_t^+$ and its pseudo cost $r_t^+$, 
        \begin{align}
            {x}_t^+ &= \argmin_{x_t \in \{0,1\}, \ x_t\ge x_{t-1}} c_t x_t + \beta (x_t - x_{t-1}) - \phi_{i} x_t;\\
            r_t^+ &= c_t x_t^+ + \beta (x_t^+ - x_{t-1}) - \phi_{i}x_t^+;
        \end{align}
        \State solve a \textbf{(ramping-off problem)} and obtain decision $x_t^-$ and its pseudo cost $r_t^-$, 
        \begin{align}
            {x}_t^- &= \argmin_{x_t \in \{0,1\}, \ x_t \le x_{t-1}} c_t x_t + \beta (x_{t-1} - x_{n}) - \phi_{i} x_t;\\
            r_t^- &= c_t x_t^- + \beta (x_{t-1} - x_t^-) - \phi_{i}x_t^-;
        \end{align}
        \State \textbf{if } $r^+_t \le r^-_t$ \textbf{ then } set $x_t = x_t^+$ \textbf{ else } set $x_t = x_t^-$;
	   \State update $i \leftarrow i + x_t$;
		\EndWhile
	\end{algorithmic}
\end{algorithm}

\clearpage

\section{Proofs}

We now prove several results described in the main body.  

{\color{blue}
In \autoref{apx:convexproofmin}, we show that the ramping-on and ramping-off optimization problems used in the \RORO instantiations for \OCSmin can be efficiently solved using iterative convex optimization techniques.

\smallskip

In \autoref{apx:comp-proofs-min}, we prove the competitive upper bound for \ROROmin (\autoref{thm:roromin}).  
In \autoref{sec:lb-proof-min}, we provide the proof of the information-theoretic lower bound for \OCSmin (\autoref{thm:lowerboundmin}), which subsequently proves that \ROROmin is optimal.  

\smallskip

In \autoref{sec:advice-comp}, we prove the $\Omega(T)$ advice complexity lower bound for \OCS presented in \autoref{thm:advicecomplexity}.
In \autoref{sec:advice-proofs-min}, we prove the consistency and robustness bounds for \ROAdvmin (\autoref{thm:advice-consist-robust-min}).
}

{\color{blue}

\subsection{Efficiently solving the ramping-on and ramping-off problems for \ROROmin} \label{apx:convexproofmin}

Recall the ramping-on and ramping-off problems given by \eqref{eq:rampon} and \eqref{eq:rampoff}, respectively.  First, we note that because the primary difference between the two problems is the restriction on the decision space, if we are only interested in the actual online decision $x_t$, it is valid to merge these problems and consider a single optimization problem as follows:
\begin{align*}
    \textsc{RampOnRampOff}(g_t( \cdot ), w^{(t-1)}, x_{t-1}) &= \quad \ \ \argmin_{\mathclap{x \in [0, \min (1- w^{(t-1)}, d_t) ]}} \quad \ \ g_t( x ) + \beta \lvert x - x_{t-1} \rvert - \int_{w^{(t-1)}}^{w^{(t-1)} + x}\phi(u) du.
\end{align*}

\noindent Let us define $f_t(x) : t \in [T]$ as the right-hand side of the combined ramping-on and ramping-off minimization problem defined above:
\begin{align}
    f_t(x) = g_t( x ) + \beta \lvert x - x_{t-1} \rvert - \int_{w^{(t-1)}}^{w^{(t-1)} + x}\phi(u) du.
\end{align}

\begin{thm}
    Under the assumptions of \OCSmin, $f_t(x)$ is convex on the entire interval $x\in [0, 1]$ for any $t \in [T]$.
\end{thm}
\begin{proof}

We prove the above statement by contradiction.

By definition, the sum of two convex functions gives a convex function.  
First, note that the switching cost term can be equivalently defined in terms of the $\ell$1 norm as follows: $\beta \lVert x - x_{t-1} \rVert_1$.  By definition and by observing that $x_{t-1}$ is fixed, $\beta \lVert x - x_{t-1} \rVert_1$ is convex.  We have also assumed as part of the \OCSmin problem setting that each $g_t( x )$ is convex.  Thus, $g_t(x) + \beta \lVert x - x_{t-1} \rVert_1$ must be convex.

We turn our attention to the term $- \int_{w^{(t-1)}}^{w^{(t-1)} + x}\phi(u) du$.  Let $k(x) = \int_{w^{(t-1)}}^{w^{(t-1)} + x }\phi(u) du$.  

By the fundamental theorem of calculus, $\nicefrac{d}{dx} \ k(x) = \phi( z^{(t-1)} + x ) \cdot \nicefrac{d}{dx} \ x = \phi( z^{(t-1)} + x )$

Let $p(x) = \phi( z^{(t-1)} + x )$.  Then $\nicefrac{d^2}{d^2 x} \ k(x) = p'(x) \cdot \nicefrac{d}{dx} x  $.  Since $\phi$ is monotonically decreasing on the interval $[0,1]$, we know that $ p'(x) < 0$, and thus $\nicefrac{d}{dx} x \cdot p'(x)$ is negative.  This gives that $k(x)$ is concave in $x$.

Since the negation of a concave function is convex, this causes a contradiction, because the sum of two convex functions is a convex function.  Note that $\left( g_t(x) + \beta \lVert x - x_{t-1} \rVert_1\right)$ and $\left( - \int_{w^{(t-1)}}^{w^{(t-1)} + x}\phi(u) du \right)$ are both convex.

Thus, $f_t(\cdot) = g_t( x ) + \beta \lvert x - x_{t-1} \rvert - \int_{w^{(t-1)}}^{w^{(t-1)} + x}\phi(u) du$ is always convex under the assumptions of the problem setting.
\end{proof}

Since the right-hand side of the combined ramping-on and ramping-off minimization problem is convex, \eqref{eq:rampon} and \eqref{eq:rampoff} can be solved efficiently using an iterative convex optimization method.
}

\subsection{Competitive results for \texttt{RORO-min}} \label{apx:comp-proofs-min}

In the following, we prove \autoref{thm:roromin}, which states that the instantiation of \RORO for \OCSmin is $\alpha$-competitive, where $\alpha$ is as defined in \eqref{eq:alpha}.

\begin{proof}[Proof of \autoref{thm:roromin}]

Let $\mathcal{I} \in \Omega$ denote any valid \OCSmin sequence, and
let $w^{(j)}$ denote \ROROmin's final utilization before the compulsory trading, which begins at time step $j \leq T$. Note that $w^{(t)} = \sum_{m\in[t]} x_m$ is non-decreasing over $t$.

\begin{lem}
\label{lem:opt-lb}
The offline optimum is lower bounded by $\emph{\OPT}(\mathcal{I}) \ge \phi(w^{(j)}) - \beta$.
\end{lem}
\paragraph{Proof of Lemma~\ref{lem:opt-lb}.} We prove this lemma by contradiction. Note that the offline optimum (ignoring additional switching costs) is to trade all items at the best cost function over the sequence $\{g_t ( \cdot ) \}_{t\in[T]}$. 

{\color{blue}
Suppose this best cost function is at an arbitrary step $m$ ($m\in[T], \ m \leq j$), denoted by $g_m (\cdot)$.  Since cost functions are convex and additionally satisfy the conditions $g_m(0) = 0$ and $g_m(x) \ge 0 \ \forall x \in [0,1]$, the derivative of $g_m$ at $x = 0$ is a lower bound on the best marginal cost (i.e. per unit of capacity) that the optimal solution can obtain.  We henceforth denote this derivative by $\frac{d g_m}{d x}$.  Assume for the sake of contradiction that $\frac{d g_m}{d x} = \OPT(\mathcal{I}) < \phi(w^{(j)}) - \beta$.  

Next, we consider $\frac{d g_m}{d x} - \beta < \phi(w^{(j)}) - 2\beta$.  Since $\phi(z)$ is strictly decreasing on $z \in [0,1]$, it follows that $\phi(w^{(j)}) - 2\beta \le \phi(w^{(m)})$, as $m \leq j$.  By solving the ramping-off problem, we have $x_m^- = x_{m-1}$. Thus, the online solution of step $m$ is dominated by $x_m^+$, i.e., $x_m = x_m^+$.

Furthermore, since $\frac{d g_m}{d x} < \phi(w^{(j)}) - \beta$ and $\phi(w^{(j)}) - \beta < \phi(w^{(m)}) - \beta$ as previously, by solving the ramping-on problem, we must have that the resulting decision satisfies $x_m^+ > w^{(j)} - w^{(m-1)}$.  This implies that $w^{(m)} > w^{(j)}$, which, given that $w^{(t)}$ is non-decreasing in $t \in [T]$, contradicts with the assumption that the final utilization before the compulsory trade is $w^{(j)}$. 

Thus, we conclude that $\OPT(\mathcal{I}) \geq \phi(w^{(j)}) - \beta$.

}

\begin{lem}
\label{lem:alg-ub}
The cost of $\ROROmin(\mathcal{I})$ is upper bounded by
\begin{align}
    \emph{\ROROmin}(\mathcal{I}) \le \int_{0}^{w^{(j)}} \phi(u) du + \beta w^{(j)} + (1- w^{(j)}) U.
\end{align}
\end{lem}
\paragraph{Proof of Lemma~\ref{lem:alg-ub}.}
By solving the ramping-off problem for any arbitrary time step $t \in [j]$, we can observe that $g_t (x_t^-) - \beta x_t^- \le \int_{w^{(t-1)}}^{w^{(t-1)}+x_t^-} \phi(u) du, \forall t\in[j]$. Therefore, we have the following inequality
\begin{align}
   \min\{r_t^+, r_t^-\} \le r_t^- \le \beta x_{t-1}, \forall t\in[j].
\end{align}

Thus, we have
\begin{align}
   \beta w^{(j)} \ge \sum_{t\in[j]} (\beta x_{t-1}) &\ge \sum_{t\in[j]}\min\{r_t^+, r_t^-\}\\
    &= \sum_{t\in[j]} \left[ g_t (x_t) + \beta |x_t - x_{t-1}| - \int_{w^{(t-1)}}^{w^{(t-1)} + x_t}\phi(u) du \right]\\
    &=  \sum_{t\in[j]} \left[ g_t (x_t) + \beta |x_t - x_{t-1}|\right] - \int_{0}^{w^{(j)}}\phi(u) du\\
    &= \ROROmin(\mathcal{I}) - (1-w^{(j)})U - \int_{0}^{w^{(j)}}\phi(u) du.
\end{align}
 
Combining Lemma~\ref{lem:opt-lb} and Lemma~\ref{lem:alg-ub} gives
\begin{align}
    \texttt{CR} \le \frac{\ROROmin(\mathcal{I})}{\OPT(\mathcal{I})} \le \frac{\int_{0}^{w^{(j)}} \phi(u) du + \beta w^{(j)} + (1- w^{(j)}) U}{\phi(w^{(j)}) - \beta} \leq \alpha, \label{eq:cr-alpha}
\end{align}
where the last inequality holds since for any $w\in[0,1]$ 
\begin{align}
  \int_{0}^{w} \phi(u) du + \beta w + (1- w) U &=  \int_{0}^{w}\left[U - \beta + (U/\alpha - U + 2 \beta) \exp(u/\alpha)\right] du + \beta w + (1 - w)U\\
  & = (U - \beta)w + \alpha (U/\alpha - U + 2 \beta) [\exp(w/\alpha) - 1] + \beta w + (1 - w)U\\
  & = \alpha (U/\alpha - U + 2 \beta) [\exp(w/\alpha) - 1] + U\\
  &= \alpha \left[U - 2 \beta + (U/\alpha - U + 2 \beta) \exp(w/\alpha) \right]\\
  &= \alpha [\phi(w) - \beta].
\end{align}

We note that the rate constraints $\{ d_t \}_{t \in [T]}$ surprisingly do not appear in this worst-case analysis.  For completeness, we state and prove Lemma~\ref{lem:rate-const-min}.

\begin{lem} \label{lem:rate-const-min}
    If $d_t < 1 \ \forall t \in [T]$, the competitive ratio of \ROROmin is still upper bounded by $\alpha$.
\end{lem}
\begin{proof}[Proof of Lemma~\ref{lem:rate-const-min}]
Suppose that the presence of a rate constraint $d_t$ causes \ROROmin to make a decision which violates $\alpha$-competitiveness.  At time $t$, the only difference between the setting where $x_t \in [0, 1]$ and the setting with rate constraints $< 1$, where $x_t \in [0, d_t]$, is that $x_t$ cannot be~$> d_t$.  

This implies that a challenging situation for \ROROmin under a rate constraint is the case where \ROROmin would otherwise accept $> d_t$ of a good cost function, but it cannot due to the rate constraint.

We can now show that such a situation implies that \ROROmin achieves a competitive ratio which is strictly better than $\alpha$ (in the minimization setting).

From \eqref{eq:cr-alpha}, we know that the following holds for any value of $w \in [0, 1]$:
\begin{align*}
\int^w_0 \phi(u) du + \beta w + (1-w) U \leq \alpha [ \phi(w) - \beta ].
\end{align*}

For an arbitrary instance $\mathcal{I} \in \Omega$ and an arbitrary time step $t$, let $w^{(t)} = w^{(t-1)} + d_t$, implying that $x_t = d_t$.  For the sake of comparison, we first consider this time step with a cost function $g_t(\cdot)$, such that $g_t(x_t) = \phi(w^{(t)}) \cdot x_t$, implying that without the presence of a rate constraint, \ROROmin would set $x_t = d_t$.  If no more cost functions are accepted by \ROROmin after time step $t$ (excepting the compulsory trade), we have the following:
\begin{align}\label{eq:rate-const-alpha}
\frac{\ROROmin(\mathcal{I})}{\OPT(\mathcal{I})} \leq \frac{\int^{w^{(t)}}_0 \phi(u) du + \beta w^{(t)} + (1-w^{(t)}) U}{\phi(w^{(t)}) - \beta} = \alpha.
\end{align}

Now, consider the exact same setting as above, except with a new cost function $g_t'(\cdot)$, such that $g_t'(x_t) < \phi(w^{(t)}) \cdot x_t$.  This implies that without the presence of a rate constraint, \ROROmin would set $x_t > d_t$.  In other words, $g_t'( \cdot)$ is a good cost function which $\ROROmin$ cannot accept more of due to the rate constraint.  

Since $\OPT$ is subject to the same rate constraint $d_t$, we know that $\OPT(\mathcal{I})$ is lower bounded by $[ \phi(w^{(t)}) - \beta ] (1- d_t) + g_t'(d_t)$ -- the rest of the optimal solution is bounded by the final threshold value, since we assume that no more prices are accepted by \ROROmin after time step $t$.

The cost \ROROmin incurs is upper bounded by $\ROROmin(\mathcal{I}) \le \int^{w^{(t)}}_0 \phi(u) du - \int^{w^{(t)}}_{w^{(t-1)}} \phi(u) du + g_t'(d_t) + \beta w^{(t)} + (1-w^{(t)}) U$.

Observe that compared to the previous setting, the $\OPT$ and $\ROROmin$ solutions have both decreased -- $\OPT(\mathcal{I})$ has decreased by a subtractive factor of $- [ \phi(w^{(t)}) - \beta ] d_t + g_t'(d_t)$, while $\ROROmin(\mathcal{I})$ has decreased by a subtractive factor of $- \int^{w^{(t)}}_{w^{(t-1)}} \phi(u) du + g_t'(d_t)$.

Since $\phi$ is monotonically decreasing on $w \in [0,1]$, we have that by definition, $[ \phi(w^{(t)}) - \beta ] d_t < \int^{w^{(t)}}_{w^{(t-1)}} \phi(u) du$.  Thus, the solution obtained by $\ROROmin$ has decreased \textit{more} than the solution obtained by $\OPT$.  This then implies the following:
\begin{align*}
\frac{\ROROmin(\mathcal{I})}{\OPT(\mathcal{I})} \leq \frac{\int^{w^{(t-1)}}_0 \phi(u) du - \int^{w^{(t)}}_{w^{(t-1)}} \phi(u) du + g_t'(d_t) + \beta w^{(t)} + (1-w^{(t)}) U}{[ \phi(w^{(t)}) - \beta ] (1- d_t) + g_t'(d_t)} <  \alpha,
\end{align*}
where the final inequality follows from \eqref{eq:rate-const-alpha}.
\end{proof}

At a high-level, this result shows that even if there is a rate constraint which prevents \ROROmin from accepting a good cost function, the worst-case competitive ratio does not change.  Combining Lemmas~\ref{lem:opt-lb}, \ref{lem:alg-ub}, and \ref{lem:rate-const-min} completes the proof.
\end{proof}

\subsection{Lower bound for \OCSmin} \label{sec:lb-proof-min}

In the following, we prove \autoref{thm:lowerboundmin}, which states that for any deterministic online algorithm solving \OCSmin, the optimal competitive ratio is $\alpha$, where $\alpha$ is as defined in \eqref{eq:alpha}.

To show this lower bound, we first define a family of special instances and then show that the competitive ratio for any deterministic algorithm is lower bounded under these instances. 
Prior work has shown that difficult instances for online search problems with a minimization objective occur when inputs arrive at the algorithm in an decreasing order of cost~\cite{Lorenz:08, Lechowicz:23, ElYaniv:01, SunZeynali:20}.  For $\OCSmin$, we additionally consider how an adaptive adversary can essentially force an algorithm to incur a large switching cost in the worst-case.
We now formalize such a family of instances $\{ \mathcal{I}_x \}_{x \in [L, U]}$, where $\mathcal{I}_x$ is called an \textit{$x$-decreasing instance}.

\begin{dfn}[$x$-decreasing instance for \OCSmin] Let $n, m \in \mathbb{N}$ be sufficiently large, and $\delta := \nicefrac{(U-L)}{n}$. For $x \in [L, U]$, $\mathcal{I}_x \in \Omega$ is an $x$-instance if it consists of $n_x := 2 \cdot \lceil \nicefrac{(x - L)}{\delta} \rceil + 1$ alternating batches of linear cost functions. The $i^{\text{th}}$ batch ($i\in [n_x-2]$) contains $m$ cost functions with coefficient $U$ if $i$ is odd, and $1$ cost function with coefficient $U-(\lceil i / 2 \rceil)\delta$ if $i$ is even. The last two batches consist of $m$ cost functions with coefficient $x + \varepsilon$, followed by $m$ cost functions with coefficient $U$.
\end{dfn} \label{dfn:xinstance-min}

\newcommand{\cItem}[1]{\boxed{ \{ g(y) = #1y \} \times m }}

\newcommand{\aItem}[1]{ g(y) = #1y }

Note that $\mathcal{I}_{U}$ is simply a stream of $m$ prices $U$. See Fig. \ref{fig:xinstance-min} for an illustration of an $x$-decreasing instance.  Since a competitive algorithm $\ALG$ should not accept the worst-case cost function with coefficient $U$ unless it is forced to, this special $x$-decreasing instance can be equivalently interpreted as an \textit{adaptive adversary}, which gives $\ALG$ $m$ worst-case inputs $U$ whenever any cost function is accepted.  As $n \to \infty$, the alternating single cost functions in an $x$-decreasing sequence continuously decrease down to $x$, and each of these ``good cost functions'' is interrupted by a section of worst-case $U$ cost functions.  Note that the last few cost functions in an $x$-decreasing instance are always $U$.

\begin{figure}[h]
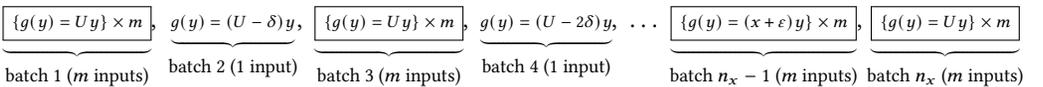

    \begin{align*}
    {\tiny 
        \underbrace{\cItem{U}}_{\text{batch 1 (} m \text{ inputs)}}, \; \underbrace{\aItem{(U - \delta)}}_{\text{batch 2 (} 1 \text{ input)}}, \; \underbrace{\cItem{U}}_{\text{batch 3 (} m \text{ inputs)}}, \; \underbrace{\aItem{(U - 2\delta)}}_{\text{batch 4 (} 1 \text{ input)}}, \; \dots \; \underbrace{\cItem{(x + \varepsilon)}}_{\text{batch $n_x-1$ (} m \text{ inputs)}}, \underbrace{\cItem{U}}_{\text{batch $n_x$ (} m \text{ inputs)}}
    }
    \end{align*}
    \vspace{-1em}
    \caption{$\mathcal{I}_x$ consists of $n_x$ batches of cost functions, where the alternating single functions are continuously decreasing from $U$ down to $x$.}\label{fig:xinstance-min}
\end{figure}

\begin{proof}[Proof of \autoref{thm:lowerboundmin}]
Let $h(x)$ denote a \textit{conversion function} $[L,U] \rightarrow [0,1]$, which fully describes the actions of a deterministic $\ALG$ for \OCSmin on an instance $\mathcal{I}_x$.  Note that for large $n$, processing the instance $\mathcal{I}_{x-\delta}$ is equivalent to first processing $\mathcal{I}_x$ (besides the last two batches), and then processing batches with prices $x - \delta$ and $U$.  Since $\ALG$ is deterministic and the conversion is unidirectional (irrevocable), we must have that $h(x - \delta) \geq h(x)$, i.e. $h(x)$ is non-increasing in $[L, U]$.  Intuitively, the entire capacity should be satisfied if the minimum possible price is observed, i.e $h(L) = 1$.
For instance $\mathcal{I}_x$, the optimal offline solution is $\OPT(\mathcal{I}_x) = x + \nicefrac{2}{m} \beta$.  Note that as $m$ grows large, $\OPT(\mathcal{I}_x) \rightarrow x$.

Due to the adaptive nature of each $x$-instance, any deterministic $\ALG$ incurs a switching cost proportional to $h(x)$, which gives the amount of utilization obtained by $\ALG$ before the end of the sequence on instance $\mathcal{I}_x$.  

Whenever $\ALG$ accepts some price $U-(\lceil i / 2 \rceil)\delta$, the adversary presents prices $U$ starting in the next time step.  Any $\ALG$ which does not switch away immediately obtains a competitive ratio strictly worse than an algorithm which does switch away (if an algorithm accepts $c$ fraction of a good price, the switching cost it will pay is $2\beta c$.  An algorithm may continue accepting $c$ fraction of prices $U$ in the subsequent time steps, but a sequence exists where this decision will take up too much utilization to recover when better prices are presented later.  In the extreme case, if an algorithm continues accepting $c$ fraction of these $U$ prices, it might fill its utilization and then $\OPT$ can accept a cost function which is arbitrarily better).

Since accepting any price by a factor of $c$ incurs a switching cost of $2 \beta c$, the switching cost paid by $\ALG$ on instance $\mathcal{I}_x$ is $2\beta h(x)$.  We assume that $\ALG$ is notified of the compulsory trade, and does not incur a significant switching cost during the final batch.

Then the total cost incurred by an $\alpha^\star$-competitive online algorithm $\ALG$ on instance $\mathcal{I}_x$ is $\ALG(\mathcal{I}_x) = h(\nicefrac{U}{\alpha^\star}) \nicefrac{U}{\alpha^\star} - \int^x_{\nicefrac{U}{\alpha^\star}} u d h(u) + 2\beta h(x) + (1 - h(x))U$, where $udh(u)$ is the cost of buying $dh(u)$ utilization at price $u$, the last term is from the compulsory trade, and the second to last term is the switching cost incurred by $\ALG$.  Note that any deterministic $\ALG$ which makes conversions when the price is larger than $\nicefrac{U}{\alpha^\star}$ can be strictly improved by restricting conversions to prices~$\leq \nicefrac{U}{\alpha^\star}$.

For any $\alpha^\star$-competitive online algorithm, the corresponding conversion function $h(\cdot)$ must satisfy $\ALG(\mathcal{I}_x) \leq \alpha^\star \OPT(\mathcal{I}_x) = \alpha^\star x, \forall x \in [L, U]$.  This gives a necessary condition which the conversion function must satisfy as follows:
\[
\ALG(\mathcal{I}_x) = h(\nicefrac{U}{\alpha^\star}) \nicefrac{U}{\alpha^\star} - \int^x_{\nicefrac{U}{\alpha^\star}} u d h(u) + 2\beta h(x) + (1 - h(x))U \leq \alpha^\star x , \quad \forall x \in [L, U].
\]
By integral by parts, the above implies that the conversion function must satisfy $h(x) \geq \frac{U - \alpha^\star x}{U - x - 2\beta} - \frac{1}{U - x - 2\beta} \int_{\nicefrac{U}{\alpha^\star}}^x h(u) du$.  By Gr\"{o}nwall's Inequality \cite[Theorem 1, p. 356]{Mitrinovic:91}, we have that
\begin{align*}
h(x) & \geq \frac{U - \alpha^\star x}{U - x - 2\beta} - \frac{1}{U - x - 2\beta} \int_{\nicefrac{U}{\alpha^\star}}^x \frac{U - \alpha^\star u}{U - u - 2\beta} \cdot \exp\left( \int_u^x \frac{1}{U - r - 2\beta} dr \right) du \\
& \geq \frac{U - \alpha^\star x}{U - x - 2\beta} - \int_{\nicefrac{U}{\alpha^\star}}^x \frac{U - \alpha^\star u}{(U - u - 2\beta)^2} du \\
& \geq \frac{U - \alpha^\star x}{U - x - 2\beta} - \left[ \frac{U\alpha^\star - U - 2\beta \alpha^\star}{u + 2\beta - U} - \alpha^\star \ln \left( u + 2\beta -U \right) \right]_{\nicefrac{U}{\alpha^\star}}^x \\
& \geq \alpha^\star \ln \left( x + 2\beta -U \right) - \alpha^\star \ln \left( \nicefrac{U}{\alpha^\star} + 2\beta -U \right), \quad \forall x \in [L, U].
\end{align*}

$h(L) = 1$ by the problem definition -- we can combine this with the above constraint to give the following condition for an $\alpha^\star$-competitive online algorithm:
\[
\alpha^\star \ln \left( L + 2\beta -U \right) - \alpha^\star \ln \left( \nicefrac{U}{\alpha^\star} + 2\beta -U \right) \leq h(L) = 1.
\]
The optimal $\alpha^\star$ is obtained when the above inequality is binding, so solving for the value of $\alpha^\star$ which solves $\alpha^\star \ln \left( L + 2\beta -U \right) - \alpha^\star \ln \left( \nicefrac{U}{\alpha^\star} + 2\beta -U \right) = 1$ yields that the best competitive ratio for any $\ALG$ solving \OCSmin is $\alpha^\star \geq \left[ W \left( \frac{e^{\nicefrac{2\beta}{U}} ( \nicefrac{L}{U} + \nicefrac{2 \beta}{U} - 1) }{e} \right) - \frac{2\beta}{U} + 1 \right]^{-1}$.
\end{proof}

\bigskip

\subsection{Lower bound for the ``best-price'' predictions in \OCS}\label{apx:advice}

In the following, we prove \autoref{thm:bestPriceConsistency}, which states that any learning-augmented algorithm for \OCS which uses a prediction of the ``best price'' in a sequence cannot achieve $1$ consistency.

\begin{proof}[Proof of \autoref{thm:bestPriceConsistency}]

To prove the result, we state and prove two lemmas, one for each of the minimization (\OCSmin) and maximization (\OCSmax) settings.  These results show that the achievable consistency is lower bounded by $1 + \nicefrac{2\beta}{L}$ and $\nicefrac{1}{1 - 2\beta / U}$ in \OCSmin and \OCSmax, respectively.

\begin{lem}\label{lem:bestPricePareto-min}
    Given a single prediction of the \textit{``minimum price''} in a sequence, any learning-augmented algorithm for \OCSmin which is $\gamma$-robust is at least $\eta$-consistent, where $\eta$ is
    \begin{align*}
        \eta \geq \gamma - \gamma \left( 1 + \frac{2\beta}{L} - \frac{U}{L} \right) \ln \left[ \frac{L + 2\beta - U}{\frac{U}{\gamma} + 2\beta - U}\right] + 1 - \frac{U}{L} + \frac{2\beta}{L}.
    \end{align*}
    Note that for any value of $\gamma$, $\eta \geq 1 + \nicefrac{2\beta}{L}$.
\end{lem}

\begin{proof}[Proof of Lemma~\ref{lem:bestPricePareto-min}]
    To show this result, we leverage the same special family of $x$-decreasing instances defined in Definition~\ref{dfn:xinstance-min}.    

    We let $h(x)$ denote a \textit{conversion function} $[L,U] \rightarrow [0,1]$, which fully describes the actions of a deterministic $\ALG$ for \OCSmin on an instance $\mathcal{I}_x$, and $h(x)$ gives the total amount traded under the instance $\mathcal{I}_x$ before the compulsory trade.   Note that for large $n$, processing the instance $\mathcal{I}_{x-\delta}$ is equivalent to first processing $\mathcal{I}_x$ (besides the last two batches), and then processing batches with prices $x - \delta$ and $U$.  Since $\ALG$ is deterministic and the conversion is unidirectional (irrevocable), we must have that $h(x - \delta) \geq h(x)$, i.e. $h(x)$ is non-increasing in $[L, U]$.  Intuitively, the entire capacity should be satisfied if the minimum possible price is observed, i.e $h(L) = 1$.

    Recall that for instance $\mathcal{I}_x$, the optimal offline solution is $\OPT(\mathcal{I}_x) = x + \nicefrac{2}{m} \beta$, and that as $m$ grows large, $\OPT(\mathcal{I}_x) \to x$.

    For any $\gamma$-robust online algorithm $\ALG$ given a prediction $\hat{P} = L$, the following must hold:
    \[
    \ALG(\mathcal{I}_x) \leq \gamma \OPT(\mathcal{I}_x) = \gamma x, \ \forall x \in [L, U].
    \]

    The cost of $\ALG$ with conversion function $h$ on an instance $\mathcal{I}_x$ is $\ALG(\mathcal{I}_x) = h(\nicefrac{U}{\gamma}) \nicefrac{U}{\gamma} - \int^x_{\nicefrac{U}{\gamma}} u d h(u) + 2\beta h(x) + (1 - h(x))U$, where $udh(u)$ is the cost of buying $dh(u)$ utilization at price $u$, the last term is from the compulsory conversion, and the second to last term is the switching cost incurred by $\ALG$. 

    This implies that $h(x)$ must satisfy the following:
    \[
    h(\nicefrac{U}{\gamma}) \nicefrac{U}{\gamma} - \int^x_{\nicefrac{U}{\gamma}} u d h(u) + 2\beta h(x) + (1 - h(x))U \leq \gamma x, \ \forall x \in [L, U].
    \]

    By integral by parts, the above implies that the conversion function must satisfy $h(x) \geq \frac{U - \gamma x}{U - x - 2\beta} - \frac{1}{U - x - 2\beta} \int_{\nicefrac{U}{\gamma}}^x h(u) du$.  By Gr\"{o}nwall's Inequality \cite[Theorem 1, p. 356]{Mitrinovic:91}, we have that
    \begin{align}
    h(x) & \geq \frac{U - \gamma x}{U - x - 2\beta} - \frac{1}{U - x - 2\beta} \int_{\nicefrac{U}{\gamma}}^x \frac{U - \gamma u}{U - u - 2\beta} \cdot \exp\left( \int_u^x \frac{1}{U - r - 2\beta} dr \right) du \\
    & \geq \frac{U - \gamma x}{U - x - 2\beta} - \int_{\nicefrac{U}{\gamma}}^x \frac{U - \gamma u}{(U - u - 2\beta)^2} du \\
    & \geq \frac{U - \gamma x}{U - x - 2\beta} - \left[ \frac{U\gamma - U - 2\beta \gamma}{u + 2\beta - U} - \gamma \ln \left( u + 2\beta -U \right) \right]_{\nicefrac{U}{\gamma}}^x \\
    & \geq \gamma \ln \left( x + 2\beta -U \right) - \gamma \ln \left( \nicefrac{U}{\gamma} + 2\beta -U \right), \quad \forall x \in [L, U]. \label{eq:rob-bound-gron-min}
    \end{align}

    In addition, to simultaneously satisfy $\eta$-consistency when the minimum price prediction is $\hat{P} = L$, $\ALG$ must satisfy $\ALG(\mathcal{I}_L) \leq \eta \OPT(\mathcal{I}_L) = \eta L$.  Combining this constraint with $h(L) = 1$ gives:
    \begin{align}
    \int_{\nicefrac{U}{\gamma}}^L h(u) du + 2\beta &\leq \eta L - L. \label{eq:const-bound-gron-min}
    \end{align}

    By combining equations \eqref{eq:rob-bound-gron-min} and \eqref{eq:const-bound-gron-min}, the conversion function $h(x)$ of any $\gamma$-robust and $\eta$-consistent online algorithm given prediction $\hat{P} = L$ must satisfy the following:
    \begin{align}
    \gamma \int_{\nicefrac{U}{\gamma}}^L \ln \left( \frac{u + 2\beta -U}{ \nicefrac{U}{\gamma} + 2\beta -U } \right)  du + 2\beta &\leq \eta L - L.
    \end{align}
    When all inequalities are binding, this equivalently gives that 
    \[
    \eta \geq \gamma - \gamma( 1 + \nicefrac{2\beta}{L} - \nicefrac{U}{L}) \ln \left( \frac{L + 2\beta -U}{ \nicefrac{U}{\gamma} + 2\beta -U } \right) + 1 - \frac{U}{L} + \frac{2\beta}{L} .
    \]
    For any value of $\gamma$, $\eta \geq 1 + 2\beta/L$, and this completes the proof.
\end{proof}

\begin{lem}\label{lem:bestPricePareto-max}
    Given a single prediction of the \textit{``maximum price''} in a sequence, any learning-augmented algorithm for \OCSmax which is $\gamma$-robust is at least $\eta$-consistent, where $\eta$ is 
    \begin{align*}
        \eta \geq \left[ 1 - \frac{U- 2\beta -L}{U \gamma} \ln \left( \frac{ U - 2\beta -L }{ \gamma L - 2\beta -L  } \right) + \frac{1}{\gamma} - \frac{L}{U} - \frac{2\beta}{U} \right]^{-1}.
    \end{align*}
    Note that for any value of $\gamma$, $\eta \geq \nicefrac{1}{1 - 2\beta / U}$.
\end{lem}

\begin{proof}[Proof of Lemma~\ref{lem:bestPricePareto-max}]
    To show this result, we leverage the same special family of $x$-increasing instances defined in Definition~\ref{dfn:xinstance-max}.    

    We let $h(x)$ denote a \textit{conversion function} $[L,U] \rightarrow [0,1]$, which fully describes the actions of a deterministic $\ALG$ for \OCSmax on an instance $\mathcal{I}_x$, and $h(x)$ gives the total amount traded under the instance $\mathcal{I}_x$ before the compulsory trade.   Note that for large $n$, processing the instance $\mathcal{I}_{x+\delta}$ is equivalent to first processing $\mathcal{I}_x$ (besides the last two batches), and then processing batches with prices $x + \delta$ and $L$.  Since $\ALG$ is deterministic and the conversion is unidirectional (irrevocable), we must have that $h(x + \delta) \geq h(x)$, i.e. $h(x)$ is non-decreasing in $[L, U]$.  Intuitively, the entire capacity should be satisfied if the minimum possible price is observed, i.e $h(U) = 1$.

    Recall that for instance $\mathcal{I}_x$, the optimal offline solution is $\OPT(\mathcal{I}_x) = x - \nicefrac{2}{m} \beta$, and that as $m$ grows large, $\OPT(\mathcal{I}_x) \to x$.

    For any $\gamma$-robust online algorithm $\ALG$ given a prediction $\hat{P} = U$, the following must hold:
    \[
    \ALG(\mathcal{I}_x) \leq \nicefrac{\OPT(\mathcal{I}_x)}{\gamma} = \nicefrac{x}{\gamma}, \ \forall x \in [L, U].
    \]

    The profit of $\ALG$ with conversion function $h$ on an instance $\mathcal{I}_x$ is $\ALG(\mathcal{I}_x) = h(\gamma L) \gamma L + \int^x_{\gamma L} u d h(u) - 2\beta h(x) + (1 - h(x))L$, where $udh(u)$ is the profit of selling $dh(u)$ utilization at price $u$, the last term is from the compulsory conversion, and the second to last term is the switching cost incurred by $\ALG$. 

    This implies that $h(x)$ must satisfy the following:
    \[
    h(\gamma L) \gamma L + \int^x_{\gamma L} u d h(u) - 2\beta h(x) + (1 - h(x))L \geq \nicefrac{x}{\gamma}, \ \forall x \in [L, U].
    \]

    By integral by parts, the above implies that the conversion function must satisfy $h(x) \geq \frac{\nicefrac{x}{\gamma} - L}{x - 2\beta - L} + \frac{1}{x - 2\beta - L} \int_{\gamma L}^x h(u) du$.  By Gr\"{o}nwall's Inequality \cite[Theorem 1, p. 356]{Mitrinovic:91}, we have that
    \begin{align}
    h(x) & \geq \frac{\nicefrac{x}{\gamma} - L}{x - 2\beta - L} + \frac{1}{x - 2\beta - L} \int_{\gamma L}^x \frac{\nicefrac{U}{\gamma} - L}{u - 2\beta - L} \cdot \exp\left( \int_u^x \frac{1}{r - 2\beta - L} dr \right) du \\
    & \geq \frac{\nicefrac{x}{\gamma} - L}{x - 2\beta - L} + \int_{\gamma L}^x \frac{\nicefrac{U}{\gamma} - L}{(u - 2\beta - L)^2} du \\
    & \geq \frac{\nicefrac{x}{\gamma} - L}{x - 2\beta - L} + \left[ \frac{L \gamma - 2 \beta - L}{\gamma u - 2\beta \gamma - L \gamma } + \frac{1}{\gamma} \ln \left( u - 2\beta -L \right) \right]_{\gamma L}^x \\
    & \geq \frac{1}{\gamma} \ln \left( x - 2\beta -L \right) - \frac{1}{\gamma} \ln \left( \gamma L - 2\beta -L \right), \quad \forall x \in [L, U] .\label{eq:rob-bound-gron-max}
    \end{align}

    In addition, to simultaneously satisfy $\eta$-consistency when the maximum price prediction is $\hat{P} = U$, $\ALG$ must satisfy $\ALG(\mathcal{I}_U) \geq \nicefrac{\OPT(\mathcal{I}_U)}{\eta} = \nicefrac{U}{\eta}$.  Combining this constraint with $h(U) = 1$ gives:

    \begin{align}
    \int_{\gamma L}^U h(u) du + 2\beta &\leq \frac{(\eta - 1)U}{ \eta}. \label{eq:const-bound-gron-max}
    \end{align}
    
    By combining equations \eqref{eq:rob-bound-gron-max} and \eqref{eq:const-bound-gron-max}, the conversion function $h(x)$ of any $\gamma$-robust and $\eta$-consistent online algorithm given prediction $\hat{P} = L$ must satisfy the following:
    \begin{align}
    \frac{1}{\gamma} \int_{\gamma L }^U \ln \left( \frac{ x - 2\beta -L }{ \gamma L - 2\beta -L  } \right)  du + 2\beta &\leq \frac{(\eta - 1)U}{ \eta}.
    \end{align}
    
    When all inequalities are binding, this equivalently gives that 
    \[
    \eta \geq \left[ 1 - \frac{U- 2\beta -L}{U \gamma} \ln \left( \frac{ U - 2\beta -L }{ \gamma L - 2\beta -L  } \right) + \frac{1}{\gamma} - \frac{L}{U} - \frac{2\beta}{U} \right]^{-1}.
    \]
    For any value of $\gamma$,  $\eta \geq \nicefrac{1}{1 - 2\beta/U}$, and this completes the proof.
\end{proof}

By combining Lemmas~\ref{lem:bestPricePareto-min} and \ref{lem:bestPricePareto-max}, the result follows.
\end{proof}

\bigskip

\subsection{Advice complexity lower bound for \OCS}\label{sec:advice-comp}
In the following, we prove \autoref{thm:advicecomplexity}, which states that for any instance of \OCS with sequence length $T$, at least $\Omega(T)$ bits of advice are necessary to achieve $1$-consistency.

\input{advice-comp}

\bigskip

\subsection{Remarks on the Advice Complexity of \OCS} \label{apx:specialcase-advice}
\input{conjectures}

\bigskip

\subsection{\texttt{RO-Advice} consistency and robustness for \OCSmin} \label{sec:advice-proofs-min}

In the following, we prove \autoref{thm:advice-consist-robust-min}, which states that the instantiation of \ROAdv for \OCSmin is $(1+\epsilon)$-consistent and $\left( \frac{\nicefrac{(U + 2 \beta)}{L} (\alpha - 1 - \epsilon) + \alpha \epsilon}{( \alpha - 1) } \right)$-robust, where $\epsilon$ and $\alpha$ are defined in Definition~\ref{dfn:roadvmin} and \eqref{eq:alpha}, respectively.  We note that consistency describes the performance of the algorithm when predictions are completely correct, while robustness describes the performance when predictions are adversarially wrong~\cite{Lykouris:18, Purohit:18}.

\begin{proof}[Proof of Theorem \ref{thm:advice-consist-robust-min}]
We start by noting that the online solution given by \texttt{RO-Advice-min} is always feasible (under the assumption that $\sum_{t=1}^{T} \hat{x}_t = 1$), since 
\[
\sum_{t=1}^{T} x_t = \sum_{t=1}^{T} \left[ \lambda \hat{x}_t + (1 - \lambda) \tilde{x}_t \right] = \lambda + (1 - \lambda) = 1.
\]
Let $\mathcal{I} \in \Omega$ be an arbitrary valid \OCSmin sequence.
We denote the \textit{purchasing} and \textit{switching} costs of the robust advice by $\ROROmin_{\text{bought}}$ and $\ROROmin_{\text{switch}}$, respectively.  Likewise, the purchasing and switching costs of the black box advice are denoted by $\ADV_{\text{bought}}$ and $\ADV_{\text{switch}}$.

The cost of \texttt{RO-Advice-min} is bounded by
\begin{align*}
    \texttt{RO-Advice-min}(\mathcal{I}) &= \sum_{t=1}^{T} g_t(x_t) + \sum_{t=1}^{T+1} \beta \lvert x_t - x_{t-1} \rvert \\
    &= \sum_{t=1}^{T} g_t \left( \lambda \hat{x}_t + (1 - \lambda) \tilde{x}_t \right) + \sum_{t=1}^{T+1} \beta \lvert \lambda \hat{x}_t + (1 - \lambda) \tilde{x}_t - \lambda \hat{x}_{t-1} - (1 - \lambda) \tilde{x}_{t-1}  \rvert \\
    &\leq \lambda \sum_{t=1}^{T} g_t(\hat{x}_t) + (1 - \lambda) \sum_{t=1}^{T} g_t( \tilde{x}_t ) \\ 
    & \quad \quad \quad \quad + \sum_{t=1}^{T+1} \beta \lvert \lambda \hat{x}_t - \lambda \hat{x}_{t-1} \rvert + \sum_{t=1}^{T+1} \beta \lvert (1 - \lambda) \tilde{x}_t - (1 - \lambda) \tilde{x}_{t-1} \rvert  \\
    &\leq \lambda \ADV_{\text{bought}}(\mathcal{I}) + (1 - \lambda) \ROROmin_{\text{bought}}(\mathcal{I})\\
    & \quad \quad \quad \quad + \lambda \sum_{t=1}^{T+1} \beta \lvert \hat{x}_t - \hat{x}_{t-1} \rvert + (1 - \lambda) \sum_{t=1}^{T+1} \beta \lvert \tilde{x}_t - \tilde{x}_{t-1} \rvert  \\
    &\leq \lambda \ADV_{\text{bought}}(\mathcal{I}) + (1 - \lambda) \ROROmin_{\text{bought}}(\mathcal{I})\\
    & \quad \quad \quad \quad + \lambda \ADV_{\text{switch}}(\mathcal{I}) + (1 - \lambda) \ROROmin_{\text{switch}}(\mathcal{I})  \\
    &\leq \lambda \ADV(\mathcal{I}) + (1 - \lambda) \ROROmin(\mathcal{I})
\end{align*}

Since $\ROROmin \leq \alpha \cdot \OPT \leq \alpha \cdot \ADV$, we have
\begin{align}
\texttt{RO-Advice-min}(\mathcal{I}) & \leq \lambda \ADV(\mathcal{I}) + (1 - \lambda) \alpha \ADV(\mathcal{I})\\
\texttt{RO-Advice-min}(\mathcal{I}) & \leq (\lambda + (1 - \lambda) \alpha) \cdot \ADV(\mathcal{I})\\ \texttt{RO-Advice-min}(\mathcal{I}) &\leq (1 + \epsilon) \cdot \ADV(\mathcal{I}). \label{eq:min-const}
\end{align}

Since $\ADV \leq U + 2\beta \leq \frac{\OPT}{\nicefrac{L}{(U+2\beta)}}$, we have
\begin{align}
\texttt{RO-Advice-min}(\mathcal{I}) & \leq \lambda \frac{\OPT(\mathcal{I})}{\nicefrac{L}{(U+2\beta)}} + (1 - \lambda) \alpha \OPT(\mathcal{I})\\
\texttt{RO-Advice-min}(\mathcal{I}) & \leq \left[ \frac{\lambda (U + 2\beta)}{L} + (1 - \lambda) \alpha \right] \cdot \OPT(\mathcal{I})\\ 
\texttt{RO-Advice-min}(\mathcal{I}) &\leq \left( \frac{\nicefrac{(U + 2 \beta)}{L} (\alpha - 1 - \epsilon) + \alpha \epsilon}{( \alpha - 1) } \right) \cdot \OPT(\mathcal{I}). \label{eq:min-rob}
\end{align}

By combining \eqref{eq:min-const} and \eqref{eq:min-rob}, the result follows.
\end{proof}

\input{Z-appendix-max}

%% file: advice-comp.tex
\begin{proof}[Proof of Theorem~\ref{thm:advicecomplexity}]

Recall that $T \geq 1$ is the integer length of a valid sequence for \PrOb.
We start by deriving a lower bound on the \textit{number of unique solutions} for sequences of length $T$, which we denote as $N_{T}$. Note that any valid solution satisfies the following, assuming that rate constraints are all 1 (i.e., $d_t = 1 \ \forall t \in [T]$):
\[
\sum_{i=1}^{T} x_t = 1, 0 \leq x_t \leq 1 \ \forall t \in [T].
\]
To simplify the problem, we discretize the action space $x_t \in [0,1]$ into $T+1$ bins, allowing each $x_t$ to take on values $\{ 0/T, 1/T, \dots, T/T \}$.  This simplifies the action space, making it countable while still allowing for solutions that either evenly distribute conversion decisions across all $T$ time steps or convert the entire asset in one-time step. By the balls-and-bars formula~\cite[Page~40]{Flajolet:09}, the number of solutions which satisfy the above constraints is $N_T \geq {2T \choose T - 1}$.

\smallskip

We prove the theorem statement by contradiction: suppose there exists an algorithm \ALG which uses $A_T := o(T)$ bits of accurate advice (i.e., sublinear in $T$) and achieves a consistency of $1$ on any input sequence with length $T$.  With $A_T$ bits of advice, $\ALG$ can theoretically distinguish between $2^{A_T}$ different solutions.

Suppose the same $\ALG$ processes an arbitrary sequence with length $T+1$.  Since the number of advice bits is sublinear in $T$, we know that $A_{T+1} < A_T + 1$.
On the other hand, the number of unique solutions $N_{T+1}$ satisfies the following:
\[
N_{T+1} \geq {2T + 2 \choose T} = \frac{(2T+2)!}{(T)! (T+2)!} = \frac{(2T)! (2T+1) (2T+2)}{(T-1)! (T) (T+1)! (T+2)} = \frac{4T^2 + 6T + 2}{T^2 + 2T} \cdot N_T
\]
Note that $\lim_{T \rightarrow \infty} N_{T+1} \geq 4 N_T$.

By assumption we have that $\ALG$ can distinguish between $2^{A_{T+1}}$ different solutions, where $2^{A_{T+1}} < 2 \cdot 2^{A_T}$.  Conversely, as the length of the sequence increases by $1$, the number of unique solutions grows by at least $N_{T+1} \geq 4 \cdot N_T$.  

If we assume that the advice is powerful enough to distinguish between the possible solutions with the sequences of length $T$, then we have $2^{A_T} \approx N_T$. However, this implies that by pigeonhole, because $2^{A_{T+1}} < 4 N_T \leq N_{T+1}$, there must exist at least two unique solutions for sequences of length $T+1$ which map to the same advice string.  Then a contradiction follows by constructing two sequences, where each of these two solutions, respectively, is the unique optimal.  By definition, such a construction forces $\ALG$ to achieve a competitive ratio strictly greater than $1$.
\end{proof}

%% file: conjectures.tex
In Section~\ref{sec:advice-model}, we gave a strong advice complexity lower bound (see Theorem~\ref{thm:advicecomplexity}) showing that advice which grows linearly in the length of the input sequence is necessary to achieve $1$-consistency for the general formulation of \OCS.  

In this supplementary section, we pick up from the brief remarks in Section~\ref{sec:advice-model} and discuss how this lower bound might be improved if certain aspects of the \OCS problem are relaxed and simplified.  Specifically, we will assume that the rate constraints are effectively ``removed'', so they satisfy $\{d_t = C\}_{\forall t \in [T]}$, and the cost/price functions are linear.  For ease of presentation, we focus on the minimization setting (\OCSmin) here.  We start with a useful observation about the optimal solution which holds under the above conditions.

\begin{obs}\label{obs:optBounded}
    For any instance of \OCSmin where the rate constraints satisfy $\{d_t = C\}_{\forall t \in [T]}$ and all cost functions are linear, the optimal solution is upper bounded by $\OPT \leq c_{\min} + 2\beta$, where $c_{\min}$ denotes the linear coefficient of the cost function which satisfies $g_{\min}(1) \leq g_t(1) \ \forall t \in [T]$.
\end{obs}
\begin{proof}[Proof for Observation~\ref{obs:optBounded}]
By definition and without loss of generality, there exists some cost function in the valid sequence which is the lowest, which we denote as $g_{\min}$.  Since we assume all cost functions are linear, it follows that the linear coefficient $c_{\min}$ of this cost function satisfies $c_{\min} \leq c_t \ \forall t \in [T]$ (i.e. it is the smallest linear coefficient of any cost function in the sequence).

Then by definition, the solution which purchases the full item at cost of $g_{\min}(1)$ and incurs a switching cost of $2\beta$ is always an upper bound on the objective value of the optimal solution (i.e., if any solution incurs a cost strictly greater than $c_{\min} + 2\beta$, it cannot be optimal).
\end{proof}

Using this useful observation for the restricted set of \OCS instances, we give two conjectures which jointly suggest that within the space of optimal solutions, there exists a ``simple'' solution which obtains the optimal objective value.  First, we suggest that if an optimal solution chooses non-zero decisions across a contiguous sequence of cost functions, there is an equivalent optimal solution which chooses the same non-zero decision across the same contiguous sequence of cost functions.

\begin{conj}\label{conj:optRampsConstant}
    For any instance of \OCSmin where the rate constraints satisfy $\{d_t = C\}_{\forall t \in [T]}$ and all cost functions are linear, if an optimal solution chooses non-zero actions (i.e. $x_t > 0$) across a contiguous sequence of cost functions with length $m$, there is always an equivalent optimal solution which chooses the same non-zero action $x^*$ for each of the $m$ cost functions.
\end{conj}
\begin{proof}[Intuition for Conjecture~\ref{conj:optRampsConstant}]
Consider a contiguous sequence of cost functions with length $m$, where (without loss of generality) the first $m-1$ cost functions have linear coefficient $c_i$, and the last cost function has linear coefficient $c_j$.  We abuse notation and generally use subscript $i$ to refer to the first $m-1$ time slots and $j$ refers to the final time slot, which is of particular interest.  Assume that $c_j > c_i$.

For the sake of contradiction, assume that there is some optimal solution which assigns non-zero decisions to all of these cost functions, where the decision is denoted by $x_i$ for the first $m-1$ cost functions, and $x_j$ for the last cost function.

Then we ask the following guiding question: \textit{when should an optimal solution \OPT set $x_j < x_i$?}  (i.e. when does it make sense for \OPT to take less than $x_i$ of the cost function w/ coefficient $c_j$?)

Intuitively, if $c_j$ is much larger than $c_i$, \OPT should not accept it at all.  On the other hand, if $c_j$ is only marginally worse than $c_i$, it makes more sense to set $x'_i = x_j$ to minimize the overall switching cost.  Enumerating these cases formally, the following expression describes the case where taking $x_j < x_i$ gives a lower overall cost than uniformly trading across the whole interval at a rate $x'_i$.
\begin{align*}
\sum_{i \not = j} c_i x_i + c_j x_j + 2\beta \max(x_i) &\leq \sum_{i \not = j} c_i x'_i + c_j x'_i + 2\beta x'_i.
\end{align*}
The inequality then implies that if the following inequality holds, setting $x_j < x_i$ is a better or equal outcome:
\begin{align*}
(c_j - c_i) \left[ \lvert \max(x_i) - x_j \rvert - \frac{\lvert \max(x'_i) - x_j \rvert }{m} \right] &\leq 2\beta \left[ \frac{\lvert \max(x'_i) - x_j \rvert }{m} \right]  .
\end{align*}
This subsequently gives the condition that when $c_j = c_i + \frac{2\beta}{m-1}$, $x_j$ can be any value between $0$ and $x'_i$ and the same cost will be attained.  However, also this implies that in this edge case, the optimal solution can either decide to distribute the trading equally (i.e. $x'_i$ for each of the $m$ cost functions), or completely ignore the last cost function (i.e. $x^*_i$ for each of the first $m-1$ cost functions), obtaining an identical objective value.

It follows, then, that if $c_j > c_i + \frac{2\beta}{m-1}$, it should be completely ignored, and if $c_j < c_i + \frac{2\beta}{m-1}$, it should be fully accepted (i.e. $x_j = x_i$).  Thus, in each of these enumerated cases, there is an optimal solution which assigns a single constant decision to each of the $m$ cost functions, and the assumption causes a contradiction.
\end{proof}

This conjecture implies that within a consecutive interval where the optimal solution decides to trade, there can be a single decision value applied across the entire interval which attains the optimal competitive ratio.  Next, we give intuition for a result showing that a solution can achieve the optimal objective by only ``ramping on'' once during a sequence.

\begin{conj}\label{conj:optRampsOnce}
    For any instance of \OCSmin where the rate constraints satisfy $\{d_t = C\}_{\forall t \in [T]}$ and all cost functions are linear, there is an optimal solution where any non-zero actions are clustered into a single contiguous sequence of arbitrary length $k$.
\end{conj}
\begin{proof}[Intuition for Conjecture~\ref{conj:optRampsOnce}]
Recall Observation~\ref{obs:optBounded}, which states that $\OPT$ is upper bounded by $c_{\min} + 2\beta$, and there is always a valid solution which simply accepts $g_{\min}(1)$ (in one time step) and incurs a switching cost of $2\beta$.

Now suppose there is a sequence with two contiguous segments of interest.  The first is a segment of length $1$, containing a single cost function with coefficient $c_{\min}$.  The other is a segment of arbitrary length $k$, containing $k$ cost functions, where the average linear coefficient is $c_{\text{avg}}$.  We assume that $c_{\text{avg}} \geq c_{\min}$, and that the two segments are disjoint, i.e. there is at least one cost function separating them in time.

Then we ask the following guiding question: \textit{when should an optimal solution \OPT consider ``transferring'' some amount of trading $x_{\text{amount}}$ from the first segment $c_{\min}$ to the $k$-segment?}  Intuitively, if $c_{\text{avg}} = c_{\min}$, it makes sense for \OPT to move all of the trading onto the $k$-segment, since it is strictly better for the switching cost.  On the other hand, if $c_{\text{avg}}$ is too large, the optimal solution will ignore it completely.  We are primarily interested in the case where the optimal solution will decide to trade some amount at the first segment, followed by a different amount at the second $k$-segment.

Enumerating these cases formally, the following expression describes the case where taking some amount $x_{\text{amount}}$ in the $k$-segment gives a lower overall cost compared to trading completely in the singleton segment.
\begin{align*}
c_{\min} (1-x_{\text{amount}}) + c_{\text{avg}} x_{\text{amount}} + 2\beta \left( (1-x_{\text{amount}}) + \frac{x_{\text{amount}}}{k} \right) &\leq c_{\min} + 2\beta\\
c_{\text{avg}} x_{\text{amount}} + \frac{2\beta x_{\text{amount}}}{k} &\leq  x_{\text{amount}} \left[ c_{\min} + 2\beta \right].
\end{align*}
The above then implies that if the following inequality holds, moving $x_{\text{amount}}$ to the $k$-segment is a better or equal outcome:
\begin{align*}
c_{\text{avg}} \leq c_{\min} + \frac{(k\cdot x_{\text{amount}} - 1) 2\beta}{k \cdot x_{\text{amount}}}.
\end{align*}
Note that $x_{\text{amount}} \in (0,1)$ under the assumption that it makes sense for \OPT to accept some amount in both segments.
However, the above gives the following corollary result:
\begin{align*}
c_{\text{avg}} < \frac{(k - 1) 2\beta}{k}.
\end{align*}
Observe that if $c_{\text{avg}} = \frac{(k - 1) 2\beta}{k}$, the solutions that fully accept in the second segment and fully accept in the first segment, respectively, obtain the same objective value: $c_{\text{avg}} + 2\beta/k = c_{\min} + 2\beta$.
This causes a contradiction because it implies that if $c_{\text{avg}} < \frac{(k - 1) 2\beta}{k}$, a \textit{strictly better} objective value can be attained by moving \textit{all} of the acceptance from the first to the second segment.  

Since the condition for moving some $x_{\text{amount}} \in (0,1)$ from the first segment to the second segment implies the above is true as a corollary (i.e. $c_{\text{avg}} \leq c_{\min} + \frac{(k\cdot x_{\text{amount}} - 1) 2\beta}{k \cdot x_{\text{amount}}} < \frac{(k - 1) 2\beta}{k}$), it seems that \textit{as soon as it makes sense to move some acceptance from one segment to another, it equivalently makes sense to move all of the acceptance to the other segment.}

This subsequently gives the condition that when $c_j = c_i + \frac{2\beta}{m-1}$, $x_j$ can be any value between $0$ and $x'_i$ and the same cost will be attained.  However, also this implies that in this edge case, the optimal solution can either decide to distribute the trading equally (i.e. $x'_i$ for each of the $m$ cost functions), or completely ignore the last cost function (i.e. $x^*_i$ for each of the first $m-1$ cost functions), obtaining an identical objective value.
\end{proof}

The significance of these results is primarily in the suggestion that an optimal solution can be fully characterized by \textit{two critical values} in this simplified setting.
Namely, if Conjectures~\ref{conj:optRampsConstant} and \ref{conj:optRampsOnce} theoretically hold, then there is always an optimal solution in the simplified \OCS setting which clusters all of its acceptance decisions into a single contiguous interval, and the acceptance rate is uniform.

Such an optimal solution can be recovered by a single start index of the acceptance interval $i^\star$, and the acceptance rate $x^\star$.  Intuitively, because these parameters do not grow with the length of the input sequence $T$, this seems to suggest that the advice complexity lower bound no longer holds in the setting without rate constraints and where cost functions are linear.

Although we do not explore this direction further because rate constraints and convex cost functions are required for most of our motivating applications (particularly carbon-aware EV charging as explored in Section~\ref{sec:eval}) it would be very interesting to explore this dynamic further in future work.  It would be particularly interesting to see whether similar or improved consistency-robustness trade-offs can be achieved in this simplified \OCS setting where the predictions given to a learning-augmented algorithm are of size sublinear in the length of the input.

%% file: Z-appendix-max.tex
{\color{blue}

\section{Problem Formulation, \RORO Instantiation, and Proofs for \OCSmax} \label{apx:ocsmax}

In this section, we describe the deferred maximization variant of \OCS (\OCSmax).  In \autoref{apx:ocsmax-form}, we present the offline formulation and relevant assumptions for the \OCSmax problem, which builds on the core problem presented in \autoref{sec:probform}.

\subsection{Formulation and assumptions} \label{apx:ocsmax-form}

In \OCSmax, an online player must sell an asset with total capacity $C$ while maximizing their profit (we again assume for notational simplicity that $C=1$).  In this setting, a \textit{concave price function} $g_t(\cdot)$ arrives online at each time step $t \in [T]$, and the switching cost is subtractive rather than additive.  The offline version is formalized below:

\begin{align}
\max_{\{x_t\}_{ t \in [T] }} & \underbrace{ \ \sum_{t=1}^T g_t( x_t ) }_{\text{Selling profit}} - \underbrace{ \sum_{t=1}^{T+1}\beta \lvert x_t - x_{t-1} \rvert }_{\text{Switching cost}}, \ \text{s.t., }  \underbrace{\sum_{t=1}^T x_t = 1,}_{\text{Deadline constraint}} \underbrace{ x_t \in [0, d_t] \ \forall t \in [T]}_{\text{Rate constraint}}.  \label{align:objMax}
\end{align}

\paragraph{Assumptions and additional notation}

For \OCSmax, we assume that price functions $\{ g_t( \cdot ) \}_{t \in [T]}$ have a bounded derivative, i.e. $L \leq d g_t / d x_t \leq U$, where $L$ and $U$ are known positive constants.  Furthermore, we assume that all cost functions $g_t(\cdot)$ are \textit{concave} -- this assumption is important as a way to model \textit{diminishing returns}, and is empirically valid for the applications of interest.

The switching cost coefficient $\beta$ is assumed to be known to the player, and is bounded within an interval ($\beta \in (0, \nicefrac{L}{2})$). 
If $\beta = 0$, the problem is equivalent to one-way trading, and if $\beta$ is ``too large'' (i.e., $> \nicefrac{L}{2}$), we can show that any competitive algorithm should \textit{only} consider the switching cost.
\footnote{As brief justification for the bounds on $\beta$, consider the following solutions for \OCSmax.  
In \OCSmax, a feasible solution can have objective $L - 2\beta$.  If $\beta > \nicefrac{L}{2}$, our solution may have a negative objective value, and this possibility should be avoided to facilitate competitive analysis.
}

\paragraph{Competitive analysis} 

As mentioned in \autoref{sec:probform}, our goal is to design an online algorithm that maintains a small \textit{competitive ratio}~\cite{Manasse:88, Borodin:92}, i.e., performs nearly as well as the offline optimal solution.   Here we note that for a problem with a maximization objective (i.e., \OCSmax), we define the competitive ratio as $\max_{\mathcal{I} \in \Omega} \nicefrac{\OPT(\mathcal{I})}{\ALG(\mathcal{I})}$, where $\Omega$ is the set of all feasible input instances for the problem.

The notions of \textit{consistency} and \textit{robustness} are similarly redefined for the maximization objective.  Letting $\ALG(\mathcal{I}, \varepsilon)$ denote the cost of a learning-augmented online algorithm on input sequence $\mathcal{I}$ when provided predictions with error factor $\varepsilon$, we have that consistency is defined as $\max_{\mathcal{I} \in \Omega} \nicefrac{ \OPT(\mathcal{I}) }{ \ALG(\mathcal{I}, 0) }$, and robustness is defined as $\max_{\mathcal{I} \in \Omega} \nicefrac{ \OPT(\mathcal{I}) }{ \ALG(\mathcal{I}, \mathbf{E}) }$, where $\mathbf{E}$ is a maximum error factor (or $\infty$). 

\subsection{Solving \OCSmax using the \RORO framework}\label{apx:ocsmax-roro}

We now present our instantiation of \RORO for \OCSmax (\ROROmax), with pseudocode summarized in \autoref{alg:roro}.
Since this algorithm uses the same \RORO framework and shares much of the same conceptual structure as \ROROmin (see \autoref{sec:ocsmin-roro}), here we highlight a few important differences and defer detailed proofs for \ROROmax to \autoref{apx:comp-proofs-max} and~\ref{sec:lb-proof-max} for brevity.
As previously, \ROROmax leverages a dynamic threshold $\Phi(w)$, where $w \in [0,1]$ denotes the current utilization.  

\begin{definition}[\ROROmax threshold function $\Phi$]\label{dfn:phi-max}
For any utilization $w \in [0,1]$, $\Phi$ is defined as:
\begin{align}\label{eq:phi-max}
    \Phi(w) = L + \beta + (\omega L - L - 2 \beta) \exp(\omega w),
\end{align}
where $\omega$ is the competitive ratio and is defined in \eqref{eq:omega}.
\end{definition}

At each time step, a price function $g_t$ arrives online and \ROROmax solves two \textit{pseudo-profit maximization problems}.  Before defining the instantiation, we note here that, because the underlying problem is a maximization problem, the pseudo-cost function we define for the \RORO framework has a flipped sign to obtain the correct result on line 6 in \autoref{alg:roro}.

\begin{definition}[\RORO instantiation for \OCSmax (\ROROmax)]
The \RORO framework solves \OCSmax when instantiated with the following parameters:

Let $\Phi (\cdot) : [0,1] \rightarrow [L, U]$ be the dynamic threshold defined in~\ref{dfn:phi-max}.

Then, define the ramping-on problem $\textsc{RampOn}(\cdot)$, ramping-off problem $\textsc{RampOff}(\cdot)$, and the pseudo-cost function $\textsc{PCost}(\cdot)$ as follows:
\begin{align}
    \textsc{RampOn}(g_t( \cdot ), w^{(t-1)}, x_{t-1}) &= \quad \ \ \argmax_{\mathclap{x \in [x_{t-1}, \min (1- w^{(t-1)}, d_t) ]}} \quad \ \ g_t( x ) - \beta (x - x_{t-1}) - \int_{w^{(t-1)}}^{w^{(t-1)} + x}\Phi(u) du,\\ \label{eq:ramponmax}
    \textsc{RampOff}(g_t( \cdot ), w^{(t-1)}, x_{t-1}) &= \argmax_{x \in [0,  \min (x_{t-1}, d_t)]} g_t ( x ) - \beta (x_{t-1} - x) - \int_{w^{(t-1)}}^{w^{(t-1)} + x}\Phi(u) du,\\ \label{eq:rampoffmax}
    \textsc{PCost}(g_t( \cdot ), w^{(t-1)}, x_t, x_{t-1}) &= - \left[ g_t ( x_t ) + \beta \lvert x_t - x_{t-1} \rvert - \int_{w^{(t-1)}}^{w^{(t-1)} + x_t}\Phi(u) du \right].
\end{align}
\end{definition}\label{dfn:roromax}

The optimizations defined above are inserted into the pseudocode defined in \autoref{alg:roro} to create an instance that solves \OCSmax (\ROROmax).  In \autoref{apx:convexproofmax}, we show that both of these optimizations are concave maximization problems, which can be solved efficiently (e.g., using iterative methods).

As previously, we give upper and lower bounds on the competitive ratio of \ROROmax in Theorems~\ref{thm:roromax} and \ref{thm:lowerboundmax}, respectively.  
Both proofs are deferred to Appendices~\ref{apx:comp-proofs-max} and~\ref{sec:lb-proof-max} for brevity.

\begin{thm}\label{thm:roromax}
\autoref{alg:roro} for \OCSmax (\ROROmax) is $\omega$-competitive when the threshold function 
is given by $\Phi(w)$ (from Definition~\ref{dfn:phi-max}), where $\omega$ is the solution to $\frac{U - L - 2 \beta}{\omega L - L - 2 \beta} = \exp(\omega)$ and is given by
\begin{equation}
    \omega \defeq W \left( \frac{\nicefrac{U}{L} - 1 - \nicefrac{2\beta}{L}}{e^{1+\nicefrac{2\beta}{L}}} \right) + 1 + \frac{2\beta}{L} \label{eq:omega}.
\end{equation}
In the above, $W(\cdot)$ is the Lambert $W$ function, defined as the inverse of $f(x) = xe^x$.
\end{thm}

\begin{thm} \label{thm:lowerboundmax}
No deterministic online algorithm for \PrObmax can achieve a competitive ratio better than $\omega$, as defined in \eqref{eq:omega}.
\end{thm}

By combining Theorems~\ref{thm:roromax} and \ref{thm:lowerboundmax}, we conclude that \ROROmax is optimal for \OCSmax.
We note that the competitive ratio achievable in the maximization setting (i.e., $\omega$ in \eqref{eq:omega}) is not the same as the competitive ratio achievable in the minimization setting (i.e., $\alpha$ in \eqref{eq:alpha}).  This distinction aligns with the existing results for online search problems such as one-way trading~\cite{ElYaniv:01} and $k$-search~\cite{Lorenz:08}, which likewise have different competitive ratios in the maximization and minimization settings.  As previously, we note that when $\beta = 0$, our \RORO algorithm recovers the optimal competitive ratio for the maximization variant of one-way trading, which is $\omega \thicksim 1 + W\left( \nicefrac{ \left( \nicefrac{U}{L} - 1 \right) }{e} \right)$, proven in~\cite{Lorenz:08, SunLee:21}.  The competitive ratio degrades as $\beta$ grows, but the assumed upper bound on $\beta$ prevents it from becoming unbounded.

\subsection{\ROAdvmax: robustly incorporating machine-learned advice for \OCSmax} \label{apx:ocsmax-adv}

We now present our instantiation of the learning-augmented \ROAdv framework for \OCSmax (\ROAdvmax), with pseudocode summarized in \autoref{alg:ro-advice}.
Since this algorithm uses the same \ROAdv framework and shares much of the same conceptual structure as \ROAdvmin (see \autoref{sec:ocsmin-adv}), here we highlight differences and defer detailed proofs to \autoref{sec:advice-proofs-max} for brevity.

\begin{definition}[\ROAdv instantiation for \OCSmax (\ROAdvmax)]\label{dfn:roadvmax}
    Let $\epsilon \in [0, \omega - 1]$, where $\omega$ is the robust competitive ratio defined in \eqref{eq:omega}.  \ROAdvmax then sets a combination factor $\lambda \vcentcolon= \left( \frac{\omega}{1+\epsilon} - 1 \right) \cdot \frac{1}{\omega-1}$, which is in $[0, 1]$.  The robust advice $\{ \tilde{x}_t \}_{\forall t \in [T]}$ is given by the \RORO instantiation for \OCSmax (\ROROmax), given in Definition~\ref{dfn:roromax}.
\end{definition}

In the following theorem, we state the consistency and robustness bounds for the \ROAdvmax meta-algorithm.  We prove this result and discuss its significance in \autoref{sec:advice-proofs-max}.

\begin{thm}\label{thm:advice-consist-robust-max}
Given a parameter $\epsilon \in [0, \omega - 1]$, where $\omega$ is defined as in~\eqref{eq:omega},\\ \ROAdvmax is $(1 + \epsilon)$-consistent and $\left( \frac{(\omega - 1) (1 + \epsilon) }{\epsilon + \frac{(L - 2\beta)}{U} (\omega - 1 - \epsilon)} \right)$-robust for \OCSmax. 
\end{thm}

\section{\OCSmax Proofs}

We now prove several of the results described in \autoref{apx:ocsmax}.

In \autoref{apx:convexproofmax}, we show that the ramping-on and ramping-off optimization problems used in the \RORO instantiations for \OCSmax can be efficiently solved using iterative convex optimization techniques.

\smallskip

In \autoref{apx:comp-proofs-max}, we prove the competitive upper bound for \ROROmax (\autoref{thm:roromax}).  
In \autoref{sec:lb-proof-max}, we provide the proof of the information-theoretic lower bound for \OCSmax (\autoref{thm:lowerboundmax}), which subsequently proves that \ROROmax is optimal.  

\smallskip

In \autoref{sec:advice-proofs-max}, we prove the consistency and robustness bounds for \ROAdvmax (\autoref{thm:advice-consist-robust-max}).

\subsection{Efficiently solving the ramping-on and ramping-off problems for \ROROmax} \label{apx:convexproofmax}

Recall the ramping-on and ramping-off problems given by \eqref{eq:ramponmax} and \eqref{eq:rampoffmax}, respectively.  First, we note that because the primary difference between the two problems is the restriction on the decision space, if we are only interested in the actual online decision $x_t$, it is valid to merge these problems and consider a single optimization problem as follows:
\begin{align*}
    \textsc{RampOnRampOff}(g_t( \cdot ), w^{(t-1)}, x_{t-1}) &= \quad \ \ \argmax_{\mathclap{x \in [0, \min (1- w^{(t-1)}, d_t) ]}} \quad \ \ g_t( x ) - \beta \lvert x - x_{t-1} \rvert - \int_{w^{(t-1)}}^{w^{(t-1)} + x}\Phi(u) du.
\end{align*}

\noindent Let us define $f_t(x) : t \in [T]$ as the right-hand side of the combined ramping-on and ramping-off maximization problem defined above:
\begin{align}
    f_t(x) = g_t( x ) - \beta \lvert x - x_{t-1} \rvert - \int_{w^{(t-1)}}^{w^{(t-1)} + x}\Phi(u) du.
\end{align}

\begin{thm}
    Under the assumptions of \OCSmax, $f_t(x)$ is concave on the entire interval $x\in [0, 1]$ for any $t \in [T]$.
\end{thm}
\begin{proof}

We prove the above statement by contradiction.

By definition, the sum of two concave functions gives a concave function.  
First, note that the switching cost term can be equivalently defined in terms of the $\ell$1 norm as follows: $\beta \lVert x - x_{t-1} \rVert_1$.  By definition and by observing that $x_{t-1}$ is fixed, $\beta \lVert x - x_{t-1} \rVert_1$ is convex.  We have also assumed as part of the \OCSmax problem setting that each $g_t( x )$ is concave.  Since the negation of a convex function is concave, we have that $g_t(x) - \beta \lVert x - x_{t-1} \rVert_1$ must be concave.

We turn our attention to the term $- \int_{w^{(t-1)}}^{w^{(t-1)} + x}\Phi(u) du$.  Let $k(x) = \int_{w^{(t-1)}}^{w^{(t-1)} + x }\Phi(u) du$.  

By the fundamental theorem of calculus, $\nicefrac{d}{dx} \ k(x) = \Phi( z^{(t-1)} + x ) \cdot \nicefrac{d}{dx} \ x = \Phi( z^{(t-1)} + x )$

Let $p(x) = \Phi( z^{(t-1)} + x )$.  Then $\nicefrac{d^2}{d^2 x} \ k(x) = p'(x) \cdot \nicefrac{d}{dx} x  $.  Since $\Phi$ is monotonically increasing on the interval $[0,1]$, we know that $ p'(x) > 0$, and thus $\nicefrac{d}{dx} x \cdot p'(x)$ is positive.  This gives that $k(x)$ is convex in $x$.

Since the negation of a convex function is concave, this causes a contradiction, because the sum of two concave functions is a concave function.  Note that $\left( g_t(x) - \beta \lVert x - x_{t-1} \rVert_1\right)$ and $\left( - \int_{w^{(t-1)}}^{w^{(t-1)} + x}\Phi(u) du \right)$ are both concave.

Thus, $f_t(\cdot) = g_t( x ) - \beta \lvert x - x_{t-1} \rvert - \int_{w^{(t-1)}}^{w^{(t-1)} + x}\Phi(u) du$ is always concave under the assumptions of the problem setting.
\end{proof}

Since the right-hand side of the combined ramping-on and ramping-off maximization problem is concave, \eqref{eq:ramponmax} and \eqref{eq:rampoffmax} can be solved efficiently by casting the concave maximization as a convex minimization problem, and using an iterative convex optimization method.

\subsection{Competitive results for \texttt{RORO-max}}\label{apx:comp-proofs-max}

In the following, we prove \autoref{thm:roromax}, which states that the instantiation of \RORO for \OCSmax is $\omega$-competitive, where $\omega$ is as defined in \eqref{eq:omega}.

\begin{proof}[Proof of \autoref{thm:roromax}]
Let $\mathcal{I} \in \Omega$ denote any valid \OCSmax sequence, and
let $w^{(j)}$ denote the final utilization before the compulsory trading, which begins at time step $j \leq T$. Note that $w^{(t)} = \sum_{m\in[t]} x_t$ is non-decreasing over $t$.

\begin{lem}
\label{lem:opt-ub-max}
The offline optimum is upper bounded by $\emph{\OPT}(\mathcal{I}) \le \Phi(w^{(j)}) + \beta$.
\end{lem}
\paragraph{Proof of Lemma~\ref{lem:opt-ub-max}.} We prove this lemma by contradiction. Note that the offline optimum is to trade all items at the best price function (ignoring additional switching costs) over the sequence $\{g_t (\cdot) \}_{t\in[T]}$.

{\color{blue}
Suppose this maximum price function is at an arbitrary step $m$ ($m\in[T], \ m \leq j$), denoted by $g_m (\cdot)$.  Since cost functions are concave and additionally satisfy the conditions $g_m(0) = 0$ and $g_m(x) \ge 0 \ \forall x \in [0,1]$, the derivative of $g_m$ at $x = 0$ is an upper bound on the best marginal return (i.e. per unit sale) that the optimal solution can obtain.  We henceforth denote this derivative by $\frac{d g_m}{d x}$.  Assume for the sake of contradiction that $\frac{d g_m}{d x} = \OPT(\mathcal{I}) > \Phi(w^{(j)}) + \beta$.  

Next, we consider $\frac{d g_m}{d x} + \beta > \Phi(w^{(j)}) + 2\beta$.  Since $\Phi(z)$ is strictly increasing on $z \in [0,1]$, it follows that $\Phi(w^{(j)}) + 2\beta \ge \Phi(w^{(m)})$, as $m \leq j$.  By solving the ramping-off problem, we have $x_m^- = x_{m-1}$. Thus, the online solution of step $m$ is dominated by $x_m^+$, i.e., $x_m = x_m^+$.

Furthermore, since $\frac{d g_m}{d x} > \Phi(w^{(j)}) + \beta$ and $\Phi(w^{(j)}) + \beta > \Phi(w^{(m)}) + \beta$ as previously, by solving the ramping-on problem, we must have that the resulting decision satisfies $x_m^+ > w^{(j)} - w^{(m-1)}$.  This implies that $w^{(m)} > w^{(j)}$, which, given that $w^{(t)}$ is non-decreasing in $t \in [T]$, contradicts with the assumption that the final utilization before the compulsory trade is $w^{(j)}$. 

Thus, we conclude that $\OPT(\mathcal{I}) \leq \Phi(w^{(j)}) + \beta$.
}

\begin{lem}
\label{lem:alg-lb-max}
The return of $\ROROmax(\mathcal{I})$ is lower bounded by
\begin{align}
    \emph{\ROROmax}(\mathcal{I}) \ge \int_{0}^{w^{(j)}} \Phi(u) du - \beta w^{(j)} + (1- w^{(j)}) L.
\end{align}
\end{lem}
\paragraph{Proof of Lemma~\ref{lem:alg-lb-max}.}
By solving the ramping-off problem for any arbitrary time step $t \in [T]$, we can observe that $g_t( x_t^- ) + \beta x_t^- \ge \int_{w^{(t-1)}}^{w^{(t-1)}+x_t^-} \Phi(u) du, \forall t\in[T]$. Therefore, we have the following inequality
\begin{align}
   \max\{r_n^+, r_n^-\} \ge r_n^- \ge - \beta x_{t-1}, \forall t\in[T].
\end{align}

Thus, we have
\begin{align}
   -\beta w^{(j)}\le -\beta w^{(t-1)} = \sum_{t\in[T]} (-\beta x_{t-1}) &\le \sum_{t\in[T]}\max\{r_n^+, r_n^-\}\\
    &= \sum_{t\in[T]} \left[ g_t( x_t ) - \beta |x_t - x_{t-1}| - \int_{w^{(t-1)}}^{w^{(t-1)} + x_t}\Phi(u) du \right]\\
    &=  \sum_{t\in[T]} \left[ g_t( x_t ) - \beta |x_t - x_{t-1}|\right] - \int_{0}^{w^{(j)}}\Phi(u) du\\
    &= \ROROmax(\mathcal{I}) - (1-w^{(j)})L - \int_{0}^{w^{(j)}}\Phi(u) du.
\end{align}
 
Combining Lemma~\ref{lem:opt-ub-max} and Lemma~\ref{lem:alg-lb-max} gives
\begin{align}
    \texttt{CR} \le \frac{\OPT(\mathcal{I})}{\ROROmax(\mathcal{I})} \le \frac{\Phi(w^{(j)}) + \beta}{\int_{0}^{w^{(j)}} \Phi(u) du - \beta w^{(j)} + (1- w^{(j)}) L} \leq \omega, \label{eq:cr-omega}
\end{align}
where the last inequality holds since for any $w\in[0,1]$ 
\begin{align}
  \int_{0}^{w} \Phi(u) du - \beta w + (1- w) L &=  \int_{0}^{w}\left[L + \beta + (\omega L - L - 2 \beta) \exp(\omega w)\right] -\beta w + (1 - w)L\\
  & = (L+\beta)w + \frac{1}{\omega}(\omega L - L - 2 \beta) [\exp(\omega w) - 1] -\beta w + (1 - w)L\\
  &= \frac{1}{\omega}\left[L + 2 \beta + (\omega L - L - 2 \beta)\exp(\omega w)\right]\\
  &= \frac{1}{\omega}[\Phi(w) + \beta]. 
\end{align}

We note that the rate constraints $\{ d_t \}_{t \in [T]}$ surprisingly do not appear in this worst-case analysis.  For completeness, we state and prove Lemma~\ref{lem:rate-const-max}.

\begin{lem} \label{lem:rate-const-max}
    If $d_t < 1 \ \forall t \in [T]$, the competitive ratio of \ROROmax is still upper bounded by $\omega$.
\end{lem}
\begin{proof}[Proof of Lemma~\ref{lem:rate-const-max}]
Suppose that the presence of a rate constraint $d_t$ causes \ROROmax to make a decision which violates $\omega$-competitiveness.  At time $t$, the only difference between the setting where $x_t \in [0, 1]$ and the setting with rate constraints $< 1$, where $x_t \in [0, d_t]$, is that $x_t$ cannot be~$> d_t$.  

This implies that a challenging situation for \ROROmax under a rate constraint is the case where \ROROmax would otherwise accept $> d_t$ of a good price function, but it cannot due to the rate constraint.

We can now show that such a situation implies that \ROROmax achieves a competitive ratio which is strictly better than $\omega$ (in the maximization setting).

From \eqref{eq:cr-omega}, we know that the following holds for any value of $w \in [0, 1]$:
\begin{align*}
\int^w_0 \Phi(u) du - \beta w + (1-w) L \geq \frac{1}{\omega} [ \Phi(w) + \beta ].
\end{align*}

For an arbitrary instance $\mathcal{I} \in \Omega$ and an arbitrary time step $t$, let $w^{(t)} = w^{(t-1)} + d_t$, implying that $x_t = d_t$.  For the sake of comparison, we first consider this time step with a price function $g_t(\cdot)$, such that $g_t(x_t) = \Phi(w^{(t)}) \cdot x_t$, implying that without the presence of a rate constraint, \ROROmax would set $x_t = d_t$.  If no more price functions are accepted by \ROROmax after time step $t$ (excepting the compulsory trade), we have the following:
\begin{align}\label{eq:rate-const-omega}
\frac{\OPT(\mathcal{I})}{\ROROmax(\mathcal{I})} \leq \frac{\Phi(w^{(t)}) + \beta}{\int^{w^{(t)}}_0 \Phi(u) du - \beta w^{(t)} + (1-w^{(t)}) L} = \omega.
\end{align}

Now, consider the exact same setting as above, except with a new price function $g_t'(\cdot)$, such that $g_t'(x_t) > \Phi(w^{(t)}) \cdot x_t$.  This implies that without the presence of a rate constraint, \ROROmax would set $x_t > d_t$.  In other words, $g_t'( \cdot)$ is a good price function which $\ROROmax$ cannot accept more of due to the rate constraint.  

Since $\OPT$ is subject to the same rate constraint $d_t$, we know that $\OPT(\mathcal{I})$ is upper bounded by $[ \Phi(w^{(t)}) + \beta ] (1- d_t) + g_t'(d_t)$ -- the rest of the optimal solution is bounded by the final threshold value, since we assume that no more prices are accepted by \ROROmax after time step $t$.

The profit of \ROROmax is lower bounded by $\ROROmax(\mathcal{I}) \ge \int^{w^{(t)}}_0 \Phi(u) du - \int^{w^{(t)}}_{w^{(t-1)}} \Phi(u) du + g_t'(d_t) - \beta w^{(t)} + (1-w^{(t)}) L$.

Observe that compared to the previous setting, the $\OPT$ and $\ROROmax$ solutions have both increased -- $\OPT(\mathcal{I})$ has increased by a additive factor of $g_t'(d_t) - [ \Phi(w^{(t)}) + \beta ] d_t$, while $\ROROmax(\mathcal{I})$ has increased by a additive factor of $g_t'(d_t) - \int^{w^{(t)}}_{w^{(t-1)}} \Phi(u) du$.

Since $\Phi$ is monotonically increasing on $w \in [0,1]$, we have that by definition, $[ \Phi(w^{(t)}) + \beta ] d_t > \int^{w^{(t)}}_{w^{(t-1)}} \Phi(u) du$.  Thus, the solution obtained by $\ROROmax$ has increased \textit{more} than the solution obtained by $\OPT$.  This then implies the following:
\begin{align*}
\frac{\OPT(\mathcal{I})}{\ROROmax(\mathcal{I})} \leq \frac{[ \Phi(w^{(t)}) + \beta ] (1- d_t) + g_t'(d_t)}{\int^{w^{(t)}}_0 \Phi(u) du - \int^{w^{(t)}}_{w^{(t-1)}} \Phi(u) du + g_t'(d_t) - \beta w^{(t)} + (1-w^{(t)}) L} <  \omega,
\end{align*}
where the final inequality follows from \eqref{eq:rate-const-omega}.
\end{proof}

At a high-level, this result shows that even if there is a rate constraint which prevents \ROROmax from accepting a good price function, the worst-case competitive ratio does not change.  Combining Lemmas~\ref{lem:opt-ub-max}, \ref{lem:alg-lb-max}, and \ref{lem:rate-const-max} completes the proof.
\end{proof}

\bigskip

\subsection{Lower bound for \OCSmax} \label{sec:lb-proof-max}
\input{lb-max}

\bigskip

\subsection{\texttt{RO-Advice} consistency and robustness for \OCSmax} \label{sec:advice-proofs-max}

In the following, we prove \autoref{thm:advice-consist-robust-max}, which states that the instantiation of \ROAdv for \OCSmax is $(1+\epsilon)$-consistent and $\left( \frac{(\omega - 1) (1 + \epsilon) }{\epsilon + \frac{(L - 2\beta)}{U} (\omega - 1 - \epsilon)} \right)$-robust, where $\epsilon$ and $\omega$ are defined in Definition~\ref{dfn:roadvmax} and \eqref{eq:omega}, respectively.  %

\begin{proof}[Proof of Theorem \ref{thm:advice-consist-robust-max}]

We start by noting that the online solution given by \texttt{RO-Advice-max} is always feasible (under the assumption that $\sum_{t=1}^{T} \hat{x}_t = 1$), since 
\[
\sum_{t=1}^{T} x_t = \sum_{t=1}^{T} \left[ \lambda \hat{x}_t + (1 - \lambda) \tilde{x}_t \right] = \lambda + (1 - \lambda) = 1.
\]

Let $\mathcal{I} \in \Omega$ be an arbitrary valid \OCSmax sequence.
We denote the \textit{selling} and \textit{switching} costs of the robust advice by $\ROROmax_{\text{sold}}$ and $\ROROmax_{\text{switch}}$, respectively.  Likewise, the purchasing and switching costs of the black box advice are denoted by $\ADV_{\text{sold}}$ and $\ADV_{\text{switch}}$.

The cost of \texttt{RO-Advice-max} is bounded by
\begin{align*}
    \texttt{RO-Advice-max}(\mathcal{I}) &= \sum_{t=1}^{T} g_t(x_t) - \sum_{t=1}^{T+1} \beta \lvert x_t - x_{t-1} \rvert \\
    &= \sum_{t=1}^{T} g_t \left( \lambda \hat{x}_t + (1 - \lambda) \tilde{x}_t \right) - \sum_{t=1}^{T+1} \beta \lvert \lambda \hat{x}_t + (1 - \lambda) \tilde{x}_t - \lambda \hat{x}_{t-1} - (1 - \lambda) \tilde{x}_{t-1}  \rvert \\
    &\geq \lambda \sum_{t=1}^{T} g_t(\hat{x}_t) + (1 - \lambda) \sum_{t=1}^{T} g_t( \tilde{x}_t ) \\
    & \quad \quad \quad \quad -\sum_{t=1}^{T+1} \beta \lvert \lambda \hat{x}_t - \lambda \hat{x}_{t-1} \rvert - \sum_{t=1}^{T+1} \beta \lvert (1 - \lambda) \tilde{x}_t - (1 - \lambda) \tilde{x}_{t-1} \rvert  \\
    &\geq \lambda \ADV_{\text{sold}}(\mathcal{I}) + (1 - \lambda) \ROROmax_{\text{sold}}(\mathcal{I})\\
    & \quad \quad \quad \quad - \lambda \sum_{t=1}^{T+1} \beta \lvert \hat{x}_t - \hat{x}_{t-1} \rvert - (1 - \lambda) \sum_{t=1}^{T+1} \beta \lvert \tilde{x}_t - \tilde{x}_{t-1} \rvert  \\
    &\geq \lambda \ADV_{\text{sold}}(\mathcal{I}) + (1 - \lambda) \ROROmax_{\text{sold}}(\mathcal{I})\\
    & \quad \quad \quad \quad - \lambda \ADV_{\text{switch}}(\mathcal{I}) - (1 - \lambda) \ROROmax_{\text{switch}}(\mathcal{I})  \\
    &\geq \lambda \ADV(\mathcal{I}) + (1 - \lambda) \ROROmax(\mathcal{I})
\end{align*}

Since $\ROROmax \geq \nicefrac{\OPT}{ \omega} \geq \nicefrac{\ADV}{ \omega}$, we have
\begin{align}
\texttt{RO-Advice-max}(\mathcal{I}) & \geq \lambda \ADV(\mathcal{I}) + (1 - \lambda) \left( \nicefrac{\ADV(\mathcal{I})}{\omega} \right)\\
\texttt{RO-Advice-max}(\mathcal{I}) & \geq \left( \lambda + \frac{1 - \lambda}{\omega} \right) \cdot \ADV(\mathcal{I})\\ 
\frac{1}{\lambda + \frac{1 - \lambda}{\omega}} \cdot \texttt{RO-Advice-max}(\mathcal{I}) & \geq \ADV(\mathcal{I})\\ 
(1 + \epsilon) \cdot \texttt{RO-Advice-max}(\mathcal{I}) & \geq \ADV(\mathcal{I}).
\end{align}

Since $\ADV \geq L - 2\beta \geq (\nicefrac{(L - 2\beta)}{U}) \OPT$, we have
\begin{align}
\texttt{RO-Advice-max}(\mathcal{I}) & \geq \lambda \frac{(L - 2\beta) \OPT(\mathcal{I})}{U} + (1 - \lambda) (\nicefrac{\OPT(\mathcal{I})}{\omega})\\
\texttt{RO-Advice-max}(\mathcal{I}) & \geq \left[ \frac{\lambda (L - 2\beta)}{U} + \frac{1 - \lambda}{\omega} \right] \cdot \OPT(\mathcal{I})\\ 
\left[ \frac{1}{\frac{\lambda (L - 2\beta)}{U} + \frac{1 - \lambda}{\omega}} \right] \cdot \texttt{RO-Advice-max} &\geq \OPT(\mathcal{I})\\
\left[ \frac{U \omega}{U - U \lambda + (L - 2\beta) \lambda \omega} \right] \cdot \texttt{RO-Advice-max} &\geq \OPT(\mathcal{I})\\
\left[ \frac{(\omega - 1) (1 + \epsilon) }{\epsilon + (\frac{L}{U} - \frac{2\beta}{U}) (\omega - 1 - \epsilon)} \right] \cdot \texttt{RO-Advice-max} &\geq \OPT(\mathcal{I}).
\end{align}

Note that as $\epsilon \rightarrow \omega - 1$, both the consistency and robustness bounds approach $\omega$, which is the optimal competitive bound shown for the setting without advice.  When $\epsilon \to 0$, the robustness bound degrades to $\nicefrac{U}{(L-2\beta)}$, but the consistency bound implies that perfect advice allows \ROAdvmax to obtain the optimal solution.
\end{proof}

}

%% file: lb-max.tex
In the following, we prove Theorem~\ref{thm:lowerboundmax}, which states that for any deterministic online algorithm solving \OCSmax, the optimal competitive ratio is $\omega$, where $\omega$ is as defined in \eqref{eq:omega}.

As previously, we define a family of special instances and then show that the competitive ratio for any deterministic algorithm is lower bounded under these instances. 
Prior work has shown that difficult instances for online search problems with a maximization objective occur when inputs arrive at the algorithm in an increasing order of cost~\cite{Lorenz:08, Lechowicz:23, ElYaniv:01, SunZeynali:20}.  For $\OCSmax$, we additionally consider how an adaptive adversary can essentially force an algorithm to incur a large switching cost in the worst-case.
We now formalize such a family of instances $\{ \mathcal{I}_x \}_{x \in [L, U]}$, where $\mathcal{I}_x$ is called an $x$-increasing instance.

\begin{figure}[ht]
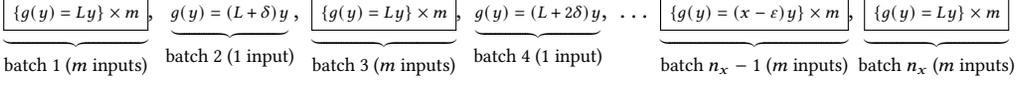

    \begin{align*}
    {\tiny 
        \underbrace{\cItem{L}}_{\text{batch 1 (} m \text{ inputs)}}, \; \underbrace{\aItem{(L + \delta)}}_{\text{batch 2 (} 1 \text{ input)}}, \; \underbrace{\cItem{L}}_{\text{batch 3 (} m \text{ inputs)}}, \; \underbrace{\aItem{(L + 2\delta)}}_{\text{batch 4 (} 1 \text{ input)}}, \; \dots \; \underbrace{\cItem{(x - \varepsilon)}}_{\text{batch $n_x-1$ (} m \text{ inputs)}}, \underbrace{\cItem{L}}_{\text{batch $n_x$ (} m \text{ inputs)}}}
    \end{align*}
    \vspace{-1em}
    \caption{$\mathcal{I}_x$ consists of $n_x$ batches of price functions, where the alternating single functions are continuously increasing from $L$ up to $x$.}\label{fig:xinstance-max}
\end{figure}

\begin{dfn}[$x$-increasing instance for \OCSmax] Let $n, m \in \mathbb{N}$ be sufficiently large, and $\delta := \nicefrac{(U-L)}{n}$. For $x \in [L, U]$, $\mathcal{I}_x \in \Omega$ is an $x$-instance if it consists of $n_x := 2 \cdot \lceil \nicefrac{(x - L)}{\delta} \rceil + 1$ alternating batches of linear price functions. The $i$-th batch ($i\in [n_x-2]$) contains $m$ price functions with coefficient $L$ if $i$ is odd, and $1$ price function with coefficient $L+(\lceil i / 2 \rceil)\delta$ if $i$ is even. The last two batches of items consist of $m$ price functions with coefficient $x - \varepsilon$, followed by $m$ price functions with coefficient $L$.
\end{dfn} \label{dfn:xinstance-max}

Note that $\mathcal{I}_{L}$ is simply a stream of $m$ price functions with coefficient $L$. See Fig. \ref{fig:xinstance-max} for an illustration of an $x$-increasing instance.  Since a competitive algorithm $\ALG$ should not accept the worst-case price function with coefficient $L$ unless it is forced to, this special $x$-increasing instance can be equivalently interpreted as an \textit{adaptive adversary}, which gives $\ALG$ $m$ worst-case inputs $L$ whenever any price function is accepted.  As $n \to \infty$, the alternating single price functions in an $x$-increasing sequence continuously increase up to $x$, and each of these ``good price functions'' is interrupted by a section of worst-case $L$ price functions.  Note that the last few price functions in an $x$-increasing instance are always $L$.

\begin{proof}[Proof of Theorem~\ref{thm:lowerboundmax}]
Let $h(x)$ denote a \textit{conversion function} $[L,U] \rightarrow [0,1]$, which fully describes the actions of a deterministic $\ALG$ for \PrObmax on an $x$-increasing instance $\mathcal{I}_x$.  Note that for large $n$, processing the instance $\mathcal{I}_{x+\delta}$ is equivalent to first processing $\mathcal{I}_x$ (besides the last two batches), and then processing batches with prices $x + \delta$ and $U$.  Since $\ALG$ is deterministic and the conversion is unidirectional (irrevocable), we must have that $h(x + \delta) \geq h(x)$, i.e. $h(x)$ is non-decreasing in $[L, U]$.  Intuitively, the entire capacity should be satisfied if the maximum possible price is observed, i.e $h(U) = 1$.

For instance $\mathcal{I}_x$, the optimal offline solution is $\OPT(\mathcal{I}_x) = x - \frac{2}{m} \beta$.  Note that as $m$ grows large, $\OPT(\mathcal{I}_x) \rightarrow x$.

Due to the adaptive nature of each $x$-instance, any deterministic $\ALG$ incurs a switching cost proportional to $h(x)$, which gives the amount of utilization used by $\ALG$ before the end of the sequence on instance $\mathcal{I}_x$.  

Whenever $\ALG$ accepts some price $L+(\lceil i / 2 \rceil)\delta$, the adversary presents prices $L$ starting in the next time step.  Any $\ALG$ which does not switch away immediately obtains a competitive ratio strictly worse than an algorithm which does switch away (if an algorithm accepts $c$ fraction of a good price, the switching cost it will pay is $2\beta c$.  An algorithm may continue accepting $c$ fraction of prices $L$ in the subsequent time steps, but a sequence exists where this decision will take up too much utilization to recover when better prices are presented later.  In the extreme case, if an algorithm continues accepting $c$ fraction of these $L$ prices, it might fill its utilization and then $\OPT$ can obtain an arbitrarily good profit)

Since accepting any price by a factor of $c$ incurs a switching cost of $2 \beta c$, the switching cost (before the final compulsory trade) paid by $\ALG$ on instance $\mathcal{I}_x$ is $2\beta h(x)$.  We assume that $\ALG$ is notified of the compulsory trade, and does not incur a significant switching cost during the final batch.

Then the total profit of an $\omega^\star$-competitive online algorithm $\ALG$ on instance $\mathcal{I}_x$ is $\ALG(\mathcal{I}_x) = h(\omega^\star L) \omega^\star L + \int^x_{\omega^\star L} u d h(u) - 2\beta h(x) + (1 - h(x))L$, where $udh(u)$ is the profit of selling $dh(u)$ utilization at price $u$, the last term is from the compulsory conversion, and the second to last term is the switching cost incurred by $\ALG$.  Note that any deterministic $\ALG$ which makes conversions when the price is larger than $\omega^\star L$ can be strictly improved by restricting conversions to prices~$\geq \omega^\star L$.

For any $\omega^\star$-competitive online algorithm, the corresponding conversion function $h(\cdot)$ must satisfy $\ALG(\mathcal{I}_x) \geq \nicefrac{\OPT(\mathcal{I}_x)}{\omega^\star} = \nicefrac{x}{\omega^\star}, \forall x \in [L, U]$.  This gives a necessary condition which the conversion function must satisfy as follows:
\[
\ALG(\mathcal{I}_x) = h(\omega^\star L) \omega^\star L + \int^x_{\omega^\star L} u d h(u) - 2\beta h(x) + (1 - h(x))L \geq \nicefrac{x}{\omega^\star}, \quad \forall x \in [L, U].
\]

By integral by parts, the above implies that the conversion function must satisfy $h(x) \geq \frac{\nicefrac{x}{\omega^\star} - L}{x - 2\beta - L} + \frac{1}{x - 2\beta - L} \int_{\omega^\star L}^x h(u) du$.  By Gr\"{o}nwall's Inequality \cite[Theorem 1, p. 356]{Mitrinovic:91}, we have that
\begin{align*}
h(x) & \geq \frac{\nicefrac{x}{\omega^\star} - L}{x - 2\beta - L} + \frac{1}{x - 2\beta - L} \int_{\omega^\star L}^x \frac{u/\omega^\star - L}{u - 2\beta - L} \cdot \exp\left( \int_u^x \frac{1}{r - 2\beta - L} dr \right) du \\
& \geq \frac{\nicefrac{x}{\omega^\star} - L}{x - 2\beta - L} + \int_{\omega^\star L}^x \frac{u/\omega^\star - L}{(u - 2\beta - L)^2} du \\
& \geq \frac{\nicefrac{x}{\omega^\star} - L}{x - 2\beta - L} + \left[ \frac{L \omega^\star - 2 \beta - L}{\omega^\star u - 2\beta \omega^\star - L \omega^\star } + \frac{1}{\omega^\star} \ln \left( u - 2\beta -L \right) \right]_{\omega^\star L}^x \\
& \geq \frac{1}{\omega^\star} \ln \left( x - 2\beta -L \right) - \frac{1}{\omega^\star} \ln \left( \omega^\star L - 2\beta -L \right), \quad \forall x \in [L, U].
\end{align*}

By the problem's definition, $h(U) = 1$ -- we can combine this with the above constraint to give the following condition for an $\omega^\star$-competitive online algorithm:
\[
\frac{1}{\omega^\star} \ln \left( U - 2\beta -L \right) - \frac{1}{\omega^\star} \ln \left( \omega^\star L - 2\beta -L \right) \leq h(U) = 1.
\]
The optimal $\omega^\star$ is obtained when the above inequality is binding, so solving for the value of $\omega^\star$ which solves $\nicefrac{1}{\omega^\star} \ln \left( U - 2\beta -L \right) - \nicefrac{1}{\omega^\star} \ln \left( \omega^\star L - 2\beta -L \right) = 1$ yields that the best competitive ratio for any $\ALG$ solving \OCSmax is $\omega^\star \geq  W \left( \frac{ \nicefrac{U}{L} - 1 - \nicefrac{2 \beta}{L} }{e^{1+ \nicefrac{2\beta}{L}}} \right) + \frac{2\beta}{L} + 1 $.
\end{proof}